\documentclass[a4paper,USenglish,cleveref,thm-restate,pdfa]{lipics-v2021}
\pdfoutput=1

\newtoggle{draft}
\togglefalse{draft}

\newtoggle{extended}
\toggletrue{extended}

\nolinenumbers

\usepackage[utf8]{inputenc}

\usepackage{xspace}

\usepackage{tolcolors}

\usepackage{float}
\usepackage{multirow}

\usepackage{xparse}

\usepackage{notes}

\usepackage{tablefootnote}

\newcommand{\ie}{i.e.\@\xspace}
\newcommand{\eg}{e.g.\@\xspace}
\newcommand{\CSharp}{C$^{\#}$\xspace}
\newcommand{\CPP}{C$^{++}$\xspace}

\let\javaId\texttt

\usepackage{adjustbox}
\usepackage{framed}

\usepackage{array}
\usepackage{booktabs}
\newcolumntype{O}{>\!c<\!} 
\newcolumntype{T}{>$r<$}

\newcommand{\hll}[1]{
  \colorlet{hlbgcolor}{gray!35}%
  \ifmmode%
    \mathchoice{%
      \colorbox{hlbgcolor}{$\displaystyle#1$}%
    }{%
      \colorbox{hlbgcolor}{$#1$}%
    }{%
      \colorbox{hlbgcolor}{$\scriptstyle#1$}%
    }{%
      \colorbox{hlbgcolor}{$\scriptscriptstyle#1$}%
    }%
  \else%
    \colorbox{hlbgcolor}{#1}%
  \fi
}

\newcommand{\setOf}[1]{\{\mkern1.5mu #1\mkern1.5mu \}}

\let\mathscr\relax
\usepackage[scr]{rsfso}

\newcommand{\alt}{\; | \;}

\usepackage{syntax}

\usepackage{scopegraphs}
\DeclareMathAlphabet\mathbfcal{OMS}{cmsy}{b}{n}

\newcommand{\edgelbl}{\ensuremath{l}}
\newcommand{\EOP}{\ensuremath{\mathdollar}}

\newcommand{\re}{\ensuremath{R}}

\newcommand{\reclos}[1]{\ensuremath{#1^\ast}}
\newcommand{\reopt}[1]{\ensuremath{#1{}^?}}

\usepackage{semantic}
\usepackage{mathtools}

\makeatletter
\newcommand\@eatpar{\@ifnextchar\par{\expandafter\@eatpar\@gobble}\relax}
\newcommand{\figuresection}[2][]{%
  \par%
  {\sffamily\bfseries #2}\hfill{#1}%
  \smallskip%
}
\makeatother

\let\oldoperatorname\operatorname
\renewcommand{\operatorname}[1]{\oldoperatorname{\mathsf{#1}}}

\newcommand{\splitAt}[2]{\operatorname{split-at}(#1, #2)}

\usepackage[noend, linesnumbered, commentsnumbered]{algorithm2e}
\usepackage{boxedminipage}

\makeatletter
\newcommand{\removelatexerror}{\let\@latex@error\@gobble}
\makeatother

\DontPrintSemicolon
\SetKwProg{Fn}{fun}{}{end fun}
\SetKwProg{Block}{}{}{}

\SetKwInOut{Input}{input}
\SetKwInOut{Output}{output}

\SetKw{LocalFnTitle}{fun}

\SetKwFunction{ResolveOrdered}{TraverseOrdered}%
\SetKwFunction{Traverse}{Traverse}%

\newcommand{\hatL}{\mkern-3.5mu\hat{\mkern4mu\mathsf{L}}}

\newcommand{\withRTNR}[1]{}
\SetKwFunction{Resolve}{Resolve}
\SetKwFunction{ResolveEOP}{Resolve-$\mathsf{\EOP}$}
\SetKwFunction{Resolvel}{Resolve-$\mathsf{\edgelbl}$}
\SetKwFunction{Resolvell}{Resolve-$\hat{\mathsf{l}}$}
\SetKwFunction{ResolveLl}{Resolve-$\hat{\mathsf{l}}\hatL$}
\SetKwFunction{ResolveL}{Resolve-$\hatL$}
\SetKwFunction{ResolveAll}{Resolve-All}
\SetKwFunction{Shadow}{Shadow}

\colorlet{query-1}{blue!50!green}
\colorlet{query-2}{colorblind-muted-5}
\colorlet{query-3}{colorblind-muted-3}
\colorlet{query-4}{colorblind-muted-8}

\usepackage{listings}
\usepackage{listings-rust}

\makeatletter
\let\old@lstKV@SwitchCases\lstKV@SwitchCases
\def\lstKV@SwitchCases#1#2#3{}
\makeatother

\usepackage{lstlinebgrd}

\makeatletter
\let\lstKV@SwitchCases\old@lstKV@SwitchCases

\lst@Key{numbers}{none}{%
    \def\lst@PlaceNumber{\lst@linebgrd}%
    \lstKV@SwitchCases{#1}%
    {none:\\%
     left:\def\lst@PlaceNumber{\llap{\normalfont
                \lst@numberstyle{\thelstnumber}\kern\lst@numbersep}\lst@linebgrd}\\%
     right:\def\lst@PlaceNumber{\rlap{\normalfont
                \kern\linewidth \kern\lst@numbersep
                \lst@numberstyle{\thelstnumber}}\lst@linebgrd}%
    }{\PackageError{Listings}{Numbers #1 unknown}\@ehc}}
\makeatother

\lstdefinestyle{defaultstyle}{
  basicstyle=\ttfamily\footnotesize,
  showstringspaces=false,
  commentstyle=\color{green!50!black},
  keywordstyle=\bfseries\color{purple},
  frame = single,
  backgroundcolor = \color{white},
  escapeinside={`}{`},
}

\lstdefinelanguage{PCF}{
  keywords={let,fun,in,true,false},
  moredelim={**[is][\color{query1}]{@1}{@}},
  moredelim={**[is][\color{query-2}]{@2}{@}},
  moredelim={**[is][\color{query-3}]{@3}{@}},
  moredelim={**[is][\color{query-4}]{@4}{@}},
}

\lstdefinelanguage{LM}{
  keywords={module,import,def,fun,fix,let,letrec,letpar,in,Int},
  moredelim={**[is][\color{query1}]{@1}{@}},
  moredelim={**[is][\color{query-2}]{@2}{@}},
  moredelim={**[is][\color{query-3}]{@3}{@}},
  moredelim={**[is][\color{query-4}]{@4}{@}},
}

\lstdefinelanguage{Statix}{
  keywords={scope,query,filter,and,min,in},
  moredelim={**[is][\color{query1}]{@1}{@}},
  moredelim={**[is][\color{query-2}]{@2}{@}},
  moredelim={**[is][\color{query-3}]{@3}{@}},
  moredelim={**[is][\color{query-4}]{@4}{@}},
}

\lstdefinelanguage{SptAML}{
  keywords={test,class,private,var,public,analysis,succeeds,run},
  moredelim={**[is][\color{query1}]{@1}{@}},
  moredelim={**[is][\color{query-2}]{@2}{@}},
  moredelim={**[is][\color{query-3}]{@3}{@}},
  moredelim={**[is][\color{query-4}]{@4}{@}},
}

\lstdefinelanguage{AML}{
  keywords={%
    module,import,%
    var,%
    class,extends,new,%
    private,protected,public,internal,%
    return,%
    int,%
  },
  moredelim={**[is][\color{query1}]{@1}{@}},
  moredelim={**[is][\color{query-2}]{@2}{@}},
  moredelim={**[is][\color{query-3}]{@3}{@}},
  moredelim={**[is][\color{query-4}]{@4}{@}},
}

\newcommand{\lbldef}[1]{\ensuremath{{\scriptstyle\mathsf{#1}}}\xspace}

\newcommand{\lblLEX}{\lbldef{LEX}}
\newcommand{\lblIMP}{\lbldef{IMP}}
\newcommand{\lblEXT}{\lbldef{EXT}}

\newcommand{\lblMOD}{\lbldef{MOD}}
\newcommand{\lblCLS}{\lbldef{CLS}}
\newcommand{\lblVAR}{\lbldef{VAR}}

\newcommand{\lblTHISMOD}{\lbldef{THIS_M}}
\newcommand{\lblTHIS}{\lbldef{THIS_C}}

\newcommand{\lblEXTPRIV}{\lbldef{EXT_{PRV}}}
\newcommand{\lblEXTPROT}{\lbldef{EXT_{PRT}}}

\newcommand{\reLEX}{\reclos{\lblLEX}\reclos{\lblEXT}\lblVAR}
\newcommand{\reLEXCPP}{\reclos{\lblLEX}\reclos{{(\lblEXT|\lblEXTPROT|\lblEXTPRIV)}}\lblVAR}

\newcommand{\reMEM}{\reclos{\lblEXT}\lblVAR}
\newcommand{\reMEMCPP}{\reclos{{(\lblEXT|\lblEXTPROT|\lblEXTPRIV)}}\lblVAR\xspace}

\newcommand{\moddecl}[2]{\lit*{mod}\: #1 \mathop{\lit*{:}} #2}
\newcommand{\clsdecl}[2]{\lit*{cls}\: #1 \mathop{\lit*{:}} #2}
\newcommand{\defdecl}[3]{\lit*{var}\: #1 \mathop{\lit*{:}} #2 \mathop{\lit*{@}} #3}

\newcommand{\tyINT}{\lit*{int}}
\newcommand{\tyINST}[1]{\lit*{inst}\;#1}

\colorlet{query1}{blue!85!black}
\colorlet{query2}{green!80!black}

\usepackage{semantic}
\usepackage{mathtools}
\usepackage{amsmath}

\newenvironment{inferrule}[1][\textwidth]{%
    \begin{minipage}{#1}
    \setpremisesend{0.25ex}%
    \addtolength{\jot}{0.25em}%
    \[%
    \hfill%
}{%
    \hfill%
    \]%
    \end{minipage}
    \ignorespacesafterend
}

\newenvironment{inferruleraw}{%
    \setpremisesend{0.25ex}%
    \addtolength{\jot}{0.25em}%
    \[%
    \hfill%
}{%
    \hfill%
    \]%
    \ignorespacesafterend
}

\newcommand{\ifdisplaystyle}[2]{\mathchoice{#1}{#2}{#2}{#2}}

\newcommand{\spacedkw}[1]{\:\mathbf{#1}\:}
\newcommand{\figquery}[4]{#1 \xreaches{#3} {\scriptstyle #2} \ifstrempty{#4}{}{\mathop{/} \text{\smaller #4}}}
\newcommand{\query}[6]{\spacedkw{query_{#1}} #2 \xreaches{#3} #4 \ifstrempty{#5}{}{\mathop{/} \text{\smaller #5}} \mapsto #6}

\newcommand\twoheaduparrow{\raisebox{-0.15em}{\rotatebox{90}{$\twoheadrightarrow$}}}
\newcommand\upexclaim{\rotatebox[origin=c]{180}{$!$}}

\newcommand{\polLit}[1]{\lit*{#1}\xspace}
\newcommand{\PUB}{\polLit{PUB}}
\newcommand{\PRT}{\polLit{PRT}}
\newcommand{\PRV}{\polLit{PRV}}
\newcommand{\MOD}{\polLit{MOD}}
\newcommand{\MODof}[1]{\MOD\; #1}
\newcommand{\SMC}{\polLit{SMC}}
\newcommand{\SMCof}[1]{\SMC\; #1}
\newcommand{\SMD}{\polLit{SMD}}
\newcommand{\SMDof}[1]{\SMD\; #1}

\newcommand{\setVar}[1]{\overline{#1}}

\makeatletter
\newcommand{\@resolvearrow}[1]{\overset{\scriptstyle\mathsf{#1}}{\rightsquigarrow}}
\newcommand{\resolvearrow}[1]{\mathbin{\ifdisplaystyle{\@resolvearrow{#1}}{\raisebox{-0.11em}{\ensuremath{\@resolvearrow{#1}}}}}}
\makeatother
\newcommand{\resolvemod}{\resolvearrow{M}}
\newcommand{\resolvecls}{\resolvearrow{C}}

\newcommand{\premProgOk}[2]{\vdash_{#1} #2\:  {\textsc{\small prog}}}
\newcommand{\premModOk}[3]{#2 \vdash_{#1} #3 \: {\textsc{\small mod}}}
\newcommand{\premModDeclOk}[3]{#2 \vdash_{#1} #3 \: {\textsc{\small md}}}
\newcommand{\premClassOk}[3]{#2 \vdash_{#1} #3 \: \textsc{\small cls}}
\newcommand{\premClassDeclOk}[3]{#2 \vdash_{#1} #3 \: \textsc{\small cd}}
\newcommand{\premExtOk}[3]{#2 \vdash_{#1} #3 \: \textsc{\small ext}}

\newcommand{\premAcc}[4]{#2 \vdash_{#1} #3 \Rightarrow #4}
\newcommand{\premCMem}[3]{#2 \vdash_{#1} #3 \: \textsc{ok}}

\newcommand{\premQCls}[4]{#2 \vdash_{#1} #3 \resolvecls #4}
\newcommand{\premQMod}[4]{#2 \vdash_{#1} #3 \resolvemod #4}
\newcommand{\premExp}[4]{#2 \vdash_{#1} #3 \mathbin{:} #4}

\newcommand{\premMAcc}[4]{#2 \vdash_{#1} #3 \mathbin{!} #4}

\newcommand{\premEncC}[3]{\vdash_{#1} #2 \mathbin{\twoheaduparrow_\mathsf{C}} #3}
\newcommand{\premEncCI}[3]{\vdash_{#1} #2 \mathbin{\uparrow_\mathsf{C}} #3}
\newcommand{\premEncM}[4][\twoheaduparrow_\mathsf{M}]{\vdash_{#2} #3 \mathbin{#1} #4}
\newcommand{\premEncMI}[4][\uparrow_\mathsf{M}]{\vdash_{#2} #3 \mathbin{#1} #4}

\newcommand{\premMPth}[3]{#2 \vdash_{#1} #3  \mathbin{\upexclaim} {}}

\newcommand{\premMatchRE}[2]{#1 \mathrel{\sim} #2}

\newcommand{\premPathOrder}[4]{#2 \vdash_{#1} #3 \mathrel{<_p} #4}

\usepackage{tikz}

\usepackage{forest}
\usetikzlibrary{angles,shadows.blur,positioning,backgrounds}

\usepackage{csquotes}

\theoremstyle{plain}
\newtheorem{assumption}[theorem]{Assumption}
\crefname{assumption}{Assumption}{Assumptions}

\definecolor{ruleref}{HTML}{8A0087}
\newcounter{rlabel} 
\makeatletter
\def\@rrefjump{2cm}
\newcommand{\rlabel}[2]{%
  \stepcounter{rlabel} 
  \edef\rref@index{\the\value{rlabel}} 
  \global\expandafter\edef\csname rule#1@idx\endcsname {\rref@index}
  \global\expandafter\edef\csname rule#1@rref\endcsname {#2}
  \smash{\raisebox{\@rrefjump}{\hypertarget{rule:#1-\rref@index}{\raisebox{-\@rrefjump}{\textsc{#2}}}}}
}
\newcommand{\rref@inner}[1]{%
  \ifcsname rule#1@rref\endcsname%
    \ifcsname rule#1@idx\endcsname%
      \edef\rref@index{\csname rule#1@idx\endcsname}
      \hyperlink{rule:#1-\rref@index}{\color{ruleref} {\textsc{\csname rule#1@rref\endcsname}}}
    \else
      \PackageError{preamble}{No index defined for rule #1.}{Internal error.}
    \fi
  \else
    \PackageError{preamble}{Rule #1 not defined.}{Use |\rlabel{|#1|}{...} to define it.|}
    \textbf{??}
  \fi
}
\newcommand{\rref}[1]{%
  \noexpand{\rref@inner{#1}}%
}
\makeatother

\newcommand{\AppendixRef}[2]{%
  \iftoggle{extended}{%
    \cref{#1}%
  }{%
    Appendix~#2~\cite{ZwaanP24artifact}%
  }%
}

\iftoggle{extended}{
  \title{Defining Name Accessibility using Scope Graphs (Extended Edition)}
}{
  \title{Defining Name Accessibility using Scope Graphs}
}



\author{Aron Zwaan}{Delft University of Technology, Delft, Netherlands}{a.s.zwaan@tudelft.nl}{0000-0002-1818-4245}{}
\author{Casper Bach Poulsen}{Delft University of Technology, Delft, Netherlands}{c.b.poulsen@tudelft.nl}{0000-0003-0622-7639}{}
\authorrunning{A.\@ Zwaan and C.\@ Bach Poulsen}
\Copyright{Aron Zwaan and Casper Bach Poulsen}

\EventEditors{Jonathan Aldrich and Guido Salvaneschi}
\EventNoEds{2}
\EventLongTitle{38th European Conference on Object-Oriented Programming (ECOOP 2024)}
\EventShortTitle{ECOOP 2024}
\EventAcronym{ECOOP}
\EventYear{2024}
\EventDate{September 16--20, 2024}
\EventLocation{Vienna, Austria}
\EventLogo{}
\SeriesVolume{313}
\ArticleNo{31}

\iftoggle{extended}{
  \relatedversiondetails[cite=ZwaanP24]{Published Version}{https://doi.org/10.4230/LIPIcs.ECOOP.2024.31}
}{
  \relatedversiondetails[cite=ZwaanP24ext]{Extended Version}{TODO: Arxiv Publication}
}

\supplement{Software Artifact}
\supplementdetails[subcategory={ECOOP'24 AEC Approved Artifact}, cite=ZwaanP24artifact]{Software}{https://zenodo.org/records/11179594}

\acknowledgements{We thank Friedrich Steimann for challenging us to specify and formalize access modifiers using scope graphs, and the anonymous reviewers for their helpful comments.}

\begin{document}

\maketitle

\begin{abstract}
%
%
%
%
%
%
Many programming languages allow programmers to regulate \emph{accessibility};
i.e., annotating a declaration with keywords such as \texttt{export} and \texttt{private} to indicate where it can be accessed.
Despite the importance of name accessibility for, e.g., compilers, editor auto-completion and tooling, and automated refactorings,
few existing type systems provide a formal account of name accessibility.








We present a declarative, executable, and language-parametric model for name accessibility, which provides a formal specification of name accessibility in Java, \CSharp, \CPP, Rust, and Eiffel.
We achieve this by defining name accessibility as a predicate on \emph{resolution paths} through \emph{scope graphs}.
Since scope graphs are a language-independent model of name resolution, our model provides a uniform approach to defining different accessibility policies for different languages.
%
%
%
%
%

Our model is implemented in Statix, a logic language for executable type system specification using scope graphs.
We evaluate its correctness on a test suite that compares it with the \CSharp, Java, and Rust compilers,
and show we can synthesize access modifiers in programs with holes accurately.
%
%

\end{abstract}


\section{Introduction}
\label{sec:introduction}

Many programming languages, especially object-oriented ones, support \emph{information hiding}, \ie, regulating from which positions in a program a declaration can be accessed.
Information hiding is used to enforce invariants of particular code units, implement design patterns (e.g. the singleton pattern), improve modularization, limit public APIs to offer guidance to library users and guarantee forward compatibility.
Support for information hiding is usually provided using \emph{access modifier keywords}%
\footnote{Other common names include `access specifier' or `visibility modifier'.}
(\emph{access modifiers} for short), such as \texttt{public}, \texttt{protected}, \texttt{internal} and \texttt{private}.
Each of these corresponds with a particular accessibility policy that is validated by the type checker.

Although recent research has not paid much attention to access modifiers,
there are still good reasons to study their semantics.
First, understanding access modifiers is required to implement (alternative) compilers and editor services correctly.
In particular, disregarding accessibility may result in incorrect name binding, and hence incorrect program behavior.
Second, formalizing access modifiers enables reasoning about the meaning of programs.
Finally, program transformation tools, such as automated refactorings, must handle the semantics of accessibility correctly.
This is especially relevant for research on large-scale automated transformations, aimed at dealing with large (legacy) codebases.
It is often infeasible to check transformations performed with such tools manually.
Thus, the correctness of these transformations must be guaranteed through other means.

\begin{figure}[tb]
  \begin{subfigure}[t]{0.32\textwidth}
    \input{figures/example1-packages}
    \vspace{-1em}
    \caption{Inheritance through Packages.}
    \label{fig:inheritance-packages}
  \end{subfigure}%
  \hspace{0.04\textwidth}%
  \begin{subfigure}[t]{0.28\textwidth}
    \input{figures/example2-shadowing}
    \vspace{-1em}
    \caption{Inaccessible or Shadowed?}
    \label{fig:shadowing-inaccessible}
  \end{subfigure}
  \hspace{0.04\textwidth}%
  \begin{subfigure}[t]{0.3\textwidth}
    \input{figures/example3-shadowing}
    \vspace{-1em}
    \caption{Accessibility and Shadowing.}
    \label{fig:shadowing-inaccessible-nested}
  \end{subfigure}
  \caption{Examples of intricate Access Modifier semantics. Classes are assumed to be public.}
  \label{fig:intricate-examples}
\end{figure}

The meaning of access modifiers can be intricate in corner cases.
We illustrate that using the examples in~\cref{fig:intricate-examples}.
In~\cref{fig:inheritance-packages}, there is an inheritance chain, where class \javaId{C} extends class~\javaId{B}, which itself extends \javaId{A}.
Classes~\javaId{A} and \javaId{C} reside in package \javaId{p1}, while \javaId{B} is in \javaId{p2}.
Class~\javaId{A} defines a package-accessible field \javaId{x}, which is accessed in \javaId{C}.
The question here is whether that access is actually allowed.
One could reason that it is correct, as the access occurs in the same package as the declaration, so a package-level declaration should be visible.
On the other hand, one could consider \javaId{x} not inherited by \javaId{B}~\cite[{\S}8.2]{JLS8}, and thus not inherited by \javaId{C} either.
In fact, the Java language designers chose the second option, rejecting this program~\cite[{\S}4.2]{SchaferTST12}.
Using \texttt{((A) this).x} is accepted however.

Something similar happens in~\cref{fig:shadowing-inaccessible}.
Here, one can consider the reference \javaId{x} in class~\javaId{C} to be invalid, as the field in class \javaId{B} is inaccessible.
Alternatively, under the assumption that \javaId{B.x} is \emph{out of scope}, the reference can be valid, pointing to \javaId{A.x}.
In this case, Java checks accessibility \emph{after shadowing}, so this program is again rejected.
However, in~\cref{fig:shadowing-inaccessible-nested}, accessibility does influence the binding.
The reference~\javaId{x} binds to the field of the \emph{enclosing} class~\javaId{B}, as the field inherited from class~\javaId{A} is inaccessible.
However, reference~\javaId{y} binds to the field inherited from~\javaId{A}.
Thus, in this case, the \emph{accessibility} of the inherited fields determines the resolution of \javaId{x} and \javaId{y}; \ie, accessibility is checked \emph{before shadowing}.
This shows that specifying accessibility is essential to defining the name binding of a language correctly.

Unintuitive semantics of accessibility occurs in non-object-oriented languages as well.
For example, the accessibility scheme of Agda seems simple: definitions are either public or module-private, and imported definitions can be re-exported.
However, issue \#5461\footnote{\url{https://github.com/agda/agda/issues/5461}} reports that re-exports in a private block are still exposed to the outside world.
While this intuitively seems wrong to most commenters, an argument is made that this is actually the intended behavior.
The discussion stalls shortly after a remark that talking about intended behavior is ``meaningless without a specification''.

These examples show that the meaning of access modifiers is not always obvious.
Hence, language designers should define their semantics unambiguously.
Ideally, that is done through \emph{specifications} containing \emph{inference rules}.
Inference rules allow unambiguous interpretation of the meaning of programming language constructs, including name binding.
However, perhaps surprisingly, a general model for defining access modifiers has never been proposed.

Perhaps closest is the work of Steimann and Thies~\cite{SteimannT09} (later incorporated in the JRRT refactoring tool~\cite{SchaferTST12}).
They propose a constraint-based approach to automating refactorings in Java, by collecting and solving \emph{accessibility constraints}.
These constraints are generated using \emph{constraint generation rules}, which cover the access rules the Java compiler enforces.
By solving these constraints, changes in accessibility implied by the refactoring can be inferred, yielding type- and behavior-preserving refactorings.


Steimann and Thies' work solves the problem of making refactorings in Java sound regarding accessibility.
However, it does not yet give a high-level explanation of the meaning of access modifiers.
This is partly because the constraint generation rules need several low-level details to catch some intricate corner cases,
but also because the function that computes the minimal required accessibility level is not given, as it was ``unpleasant to specify'' and ``of no theoretical interest''~\cite[{\S}5.2]{SteimannT09}.
Therefore, their work cannot easily be adapted to a different language or a different application (\eg, a type checker).

To advance the state of the art, we pursue the following goals:
\begin{itemize}
  \item Explain the meaning of access modifiers.
  \item Explain the (subtle) differences between access modifiers in different languages.
  \item Provide a framework for experimenting with feature combinations that do not (yet) exist in other languages.
\end{itemize}
To this end, we do not fully formalize one particular language, but rather define a toy language that incorporates and combines a large number of accessibility features.
To abstract over low-level name resolution details, we use \emph{scope graphs}~\cite{NeronTVW15,AntwerpenPRV18,RouvoetAPKV20,ZwaanA23}.
In this paper, we demonstrate this is a natural fit, because accessibility can be expressed as a predicate over paths in a scope graph.
%
%
The specification is written in the logic language Statix~\cite{AntwerpenPRV18,RouvoetAPKV20}, which has a well-defined declarative semantics and also supports generating executable type-checkers automatically.

We compare these executable type checkers with reference compilers of Java, \CSharp, and Rust, showing that we accurately captured the semantics of access modifiers in some real-world languages.
Moreover, using Statix/scope graphs as a basis for \emph{(language-parametric) refactorings} is an active topic of research~\cite{Misteli21,Gugten22,MiljakPS23,PoulsenLM24}.
We envision that this will provide accessibility-aware refactorings similar to Steimann et al., without requiring significant additional effort.
This is substantiated by the fact that Statix-based code completion~\cite{PelsmaekerAPV22} proposes an access modifier if and only if it would not cause accessibility errors elsewhere in the program.

In summary, the contributions of this paper are as follows:
\begin{itemize}
  \item We provide a systematic classification of accessibility features (\cref{sec:acc-mod-real-world});
  \item we apply our taxonomy to Java, \CPP, \CSharp, Rust, and Eiffel (\cref{sec:acc-mod-real-world});
  \item we present a specification of (various versions of) accessibility on
    modules (\cref{sec:modules}),
    subclasses (\cref{sec:subclass}), and
    their conjunctive and disjunctive combination (\cref{sec:combining-subclass-module});
  \item we extend our specification with accessibility-restricting inheritance (\cref{sec:inheritance-restriction});
  \item we prove some theorems about our model, showing it is well-behaved (\cref{sec:analysis}); and
  \item we implement our specification in Statix, and compare it with the standard compilers of Java, \CSharp, and Rust.
        Moreover, we show access modifiers can be synthesized accurately using Statix-based Code Completion~\cite{PelsmaekerAPV22} (\cref{sec:evaluation}).
\end{itemize}
\iftoggle{extended}{}{This paper comes with an artifact that allows reproducing the evaluation~\cite{ZwaanP24artifact}, and appendices containing a full specification of the access modifiers and proofs of the stated theorems~\cite{ZwaanP24ext}.}

\section{Access Modifiers in Real-World Languages}
\label{sec:acc-mod-real-world}

In this section, we explore the design space of access modifiers as they occur in real-world languages.
We first motivate why languages have access modifiers (\cref{sec:why-accessibility}).
After that, we discuss common accessibility features (\cref{subsec:accessibility-practise}), summarizing them in a feature model (\cref{subsec:classification}).

\subsection{Why Accessibility?}
\label{sec:why-accessibility}

Most programming languages allow programmers to define entities (variables, functions, types, etc.), and assign a name to them.
That name can then be used to refer to the introduced entity from other positions in the program.
However, as there is typically a large number of entities within a software project, most languages offer a notion of modularization to group related definitions.
Equally named definitions in different modules can be distinguished by qualifying them with the name of the module in which they reside.
Unqualified (or partially qualified) names by default resolve within their enclosing module, or imported modules.
Details of this scheme differ from language to language, but generally aim to make definitions easy to refer to (\eg, by minimizing the number of required qualifiers), while trying to be unambiguous to the compiler and the programmer.

However, these rules may often be too lenient with respect to the intention of the programmer.
A definition may be accessible from scopes where it is not intended to be used.
This can have detrimental effects on the quality of a software artifact.
For example, exposing all internal definitions of a library makes it
(1) less intuitive to its users,
(2) prone to forward compatibility issues and technical dept (\eg strong coupling).

For these reasons, many programming languages provide constructs that give \emph{the programmer} control over the regions of code where a definition can be accessed.
For example, in many object-oriented languages, a class can access fields from its ancestor classes by default (language-controlled).
However, if the programmer does not want a field to be accessible from subclasses, they can add a \texttt{private} access modifier.
This modifier \emph{prevents} access from all other classes (programmer-controlled).
Although many constructs that provide access control to the programmer can be envisioned, most languages settle on a limited set of keywords that can be attached to a definition.
In practice, this relatively simple scheme has proven powerful enough to cover most use cases.

\subsection{Accessibility in Practice}
\label{subsec:accessibility-practise}

Next, we explore how languages typically provide modularization and accessibility features.

\subparagraph*{Modules}
A common feature that provides modularization is \emph{modules} (also called `package' or `namespace').
A module is a syntactic
construct that introduces a named collection of definitions.
Members of modules can be accessed using the name of the module, for example in a preceding import statement, or as a qualifier to the name of the member that is accessed.

Hiding a definition from other modules is the simplest accessibility restriction that can be applied with respect to modules.
For example, Java declarations without an access modifier can only be accessed within the same package.
Rust items without a modifier behave similarly, except that declarations can still be accessed from submodules.

Some languages have multiple notions of modularization.
For example, \CSharp has assemblies, namespaces, and files, where a namespace can comprise multiple files, and/or a file can contain multiple namespaces.
The \texttt{internal} keyword in \CSharp restricts accessibility to the \emph{assembly}, and the \texttt{file} keyword (introduced in \CSharp 11~\cite{CSfile}) to the current file.
Similarly, Java 9 introduces \emph{modules}~\cite{JSR376}, with features to restrict access from external modules.

Some languages give some more control over \emph{which} modules a declaration can be accessed from.
For example, Rust has the \texttt{pub(in} \textit{path}\texttt{)} access modifier, where \textit{path} refers to some enclosing module.
This enables programmers to expose items to an arbitrary ancestor.

\emph{Imports} usually do not affect the visibility of a declaration.
A notable exception to this rule is \emph{re-exporting} (e.g., as implemented in Rust), which can actually \emph{change} the visibility of a declaration, as shown in~\cref{fig:reexport}.
In this program, the module \texttt{inner} is accessible in \texttt{outer}, but not in its parent (the root scope).
Therefore, the function \texttt{main} cannot access its field~\texttt{x}.
However, \texttt{outer} re-exports \texttt{inner::x}, which gives rise to a new definition \texttt{outer::x}.
As \texttt{outer} is accessible in the root scope, so is this definition.
Hence, via the re-export, \texttt{main} can access \texttt{x}, although the original declaration was hidden.

\begin{figure}[t]
\vspace{-1em}
\begin{minipage}{0.5\textwidth}
\begin{lstlisting}[language=Rust]
mod outer {
  mod inner {
    pub x = 42;
  }
  pub use inner::x;
}
\end{lstlisting}
\end{minipage}%
\begin{minipage}{0.5\textwidth}
\begin{lstlisting}[language={Rust}]
fn main() {
  // ERROR: inner is inaccessible:
  // let x = outer::inner::x;
  let x = outer::x;
  println!("{x}")
}
\end{lstlisting}
\end{minipage}%
\vspace{-1.2em}
\caption{Re-exports can change Accessibility}
\label{fig:reexport}
\end{figure}

From an accessibility point of view, re-exporting can typically be considered as a combination of an import and a declaration, where the declaration always points to the imported member.
The re-exported item (\texttt{inner::x} in the example) should be accessible from the location of the \emph{re-export}.
References to the re-export should have access to the location of the re-export, but not necessarily to the location of the original declaration.
In fact, for any access path, it does not matter whether the declaration is a re-export or not.

\subparagraph*{Classes}
A special modularization concept is the notion of \emph{classes},
which represent composite data types with associated operations (methods).
Where simple modules only have a static interpretation, an arbitrary number of class instances can exist at runtime.%
\footnote{
  At this point, we slightly over-simplify the reality.
  For example, neither parameterized modules (ML) nor \texttt{objects} (e.g. Scala/Kotlin) fit in this scheme.
  We made this choice deliberately, to cover the most prevalent cases.
  We conjecture that the techniques we develop for classes can be applied to parameterized modules
  (and vice versa for modules and objects)
  but leave explicating that to future~work.
}
While modules can implicitly be related to each other by their relative position, such a relation does not exist for classes.
However, classes can extend other classes, ensuring the subclass inherits the fields of its parent class.
This creates an inheritance hierarchy orthogonal to the module~hierarchy.

Object-oriented languages usually provide modifiers to control accessibility over the inheritance chain.
For example, Java and \CSharp have a \texttt{private} keyword, which prevents access outside the defining class.
Additionally, the \texttt{protected} keyword allows access from subclasses, but prevents access from any other location.

In Java and \CSharp, the accessibility level is inherited with the field.
That means, if a field in the superclass is \texttt{protected}, it will be protected in the subclass as well.
However, \CPP allows restricting the accessibility of members of the parent class.
A \texttt{private} modifier on extends-clauses will make all inherited public/protected members private on instances of the subclass.
Similarly, a \texttt{protected} modifier will make all inherited public members protected.

Finally, some languages allow specifying `friend' classes, which grant the friend access to its members.
This enables fine-grained access control, independent from module and class hierarchies.
While discouraged in \CPP, Eiffel provides only this access control mechanism.

\subparagraph*{Interaction}
Accessibility restrictions on modules and classes be combined.
This is very explicit in \CSharp, which 
has \texttt{protected internal} and \texttt{private protected} as additional modifiers.
The former permits access from within the assembly (similar to \texttt{internal}) \emph{and} to subclasses (similar to \texttt{protected}), even if they live outside the assembly.
Analogously, \texttt{private protected} grants access to subclasses in the same assembly only, which is equivalent to the conjunction of \texttt{internal} and \texttt{protected}.

\subsection{Classification}
\label{subsec:classification}

These concepts are organized
and related
in the feature model in~\cref{fig:feature-model}.
Following the previous discussion, the main features are modules and classes.
We have only a single feature for modules, because the different variants are (apart from \CSharp{s} files and namespaces) typically not mutually nested.
The \texttt{internal} keyword can either relate to the containing module (Direct) or an arbitrary parent module (Ancestor).
We explore this further in~\cref{sec:modules}.

\begin{figure}[!b]
  \centering
  \small
  \begin{forest}%
  for tree={
        %
        parent anchor=south,
        child anchor=north,
        tier/.pgfmath=level(),
        draw=black,
        delay={
          content={\strut #1}
        },
        inner xsep=2pt,
        inner ysep=2pt,
        align=center
  },
  blackcircle/.style={tikz={\node[fill=black!60,inner sep=2pt,circle]at(.north){};}},
  whitecircle/.style={tikz={\node[draw,fill=white,inner sep=2pt,circle]at(.north){};}},
  [Programming Language
    [Modules,whitecircle
      [internal,whitecircle
        [Direct,tier=lower]
        [Ancestor,tier=lower]
      ]
    ]
    [Classes,whitecircle
      [Friends,whitecircle,xshift=-2em]
      [Subclass Access\\Modifiers,whitecircle,tier=lower,name=ha
        [private,tier=ham,name=priv]
        [protected,tier=ham]
        [protected OR internal,tier=ham,name=protoi]
        [protected AND internal,tier=ham,name=protai]
      ]
      [Extends Clause\\Access Modifier,whitecircle,name=ext
        [private,tier=lower]
        [protected]
      ]
    ]
  ]
  %
  \foreach \i/\j/\k in {!111/!11/!112}
  {
    \coordinate (A) at (\i.north);
    \coordinate (O) at (\j.south);
    \coordinate (B) at (\k.north);
    \path  (A)--(O)--(B)
      pic [fill=white, draw=black, angle radius=3mm] {angle = A--O--B};
    \draw (O) -- (A);
    \draw (O) -- (B);
  }
  %
  \foreach \i/\j/\k in {priv/ha/protai,!231/!23/!232}
  {
    \coordinate (A)at (\i.north);
    \coordinate (O)at (\j.south);
    \coordinate (B)at (\k.north);
    \path  (A)--(O)--(B)
      pic [fill=black!80, angle radius=3mm] {angle = A--O--B};
  }
  %
  \coordinate (MR) at ($(!111.west) + (-2em,0)$);
  \draw[dashed] (protoi.south)
             -- ([yshift=-1em]protoi.south)
             node[fill=white, xshift=-8em] {\flqq{}requires\frqq{}}
             -- ([yshift=-1em]protoi.south -| MR)
             -- (MR)
             -- (MR |- !1.west)
             edge[->] (!1.west);
  \draw[dashed] ([yshift=-1em]protoi.south)
             -- ([yshift=-1em]protai.south)
             -- (protai.south);
  %
  \node[above right = 2.475em and 3em of ext] (lbl-opt) {Optional};
  \draw ([xshift=-0.5em, yshift=0.7em]lbl-opt.west) -- ([xshift=-0.5em, yshift=-0.3em]lbl-opt.west);
  \node[draw,fill=white,inner sep=2pt,circle, below left = 0.3em and 0.2675em of lbl-opt.west] {};
  \node[below = 2em of lbl-opt.west, anchor = west] (lbl-or) {Or};
  \coordinate (or-top)   at ($(lbl-or.west) + (-0.5em, 0.55em)$);
  \coordinate (or-left)  at ($(or-top)      + (-0.5em, -1em)$);
  \coordinate (or-right) at ($(or-top)      + (0.5em, -1em)$);
  \path[draw] (or-left)--(or-top)--(or-right)
    pic [fill=black!80, angle radius=3mm] {angle = or-left--or-top--or-right};
  \node[below = 2em of lbl-or.west, anchor = west] (lbl-alt) {Alternative (xor)};
  \coordinate (alt-top)   at ($(lbl-alt.west) + (-0.5em, 0.55em)$);
  \coordinate (alt-left)  at ($(alt-top)      + (-0.5em, -1em)$);
  \coordinate (alt-right) at ($(alt-top)      + (0.5em, -1em)$);
  \path[draw] (alt-left)--(alt-top)--(alt-right)
    pic [fill=white, draw=black, angle radius=3mm] {angle = alt-left--alt-top--alt-right};
  \node[above left = 1em and 1.5em of lbl-opt.north west, anchor = west] (header) {Legend};
  \draw (header.north east) -- (header.south east) -- (header.south west);
  \node[draw = black, fit=(header)(lbl-alt), inner sep = 0] {};
  \end{forest}
  
  \caption{Feature Model for Access Control.}
  \label{fig:feature-model}

\end{figure}
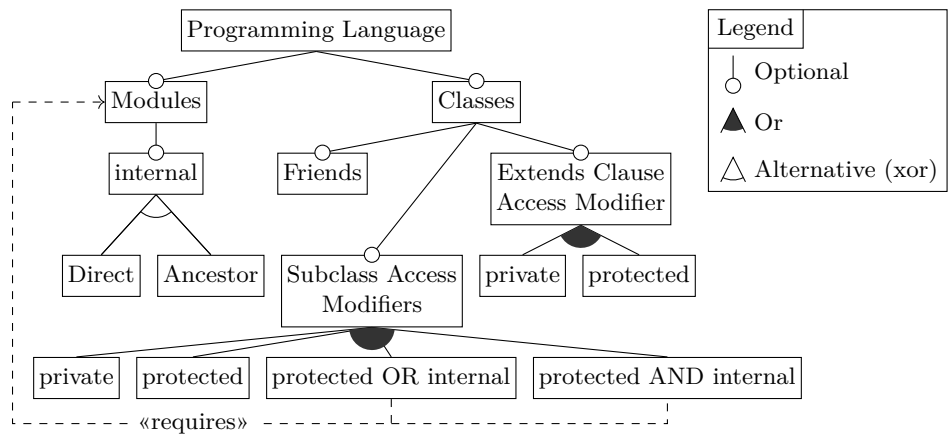

\begingroup
\newcommand{\sLang}[1]{#1}
\newcommand{\sHead}[1]{\textbf{#1}\vphantom{1.2em}}
\newcommand{\sMid}[1]{#1\vphantom{1.2em}}
\newcommand{\sBott}[1]{\textit{#1}\vphantom{1.2em}}
\newcommand{\sBot}[1]{\hspace{1.2ex}\sBott{#1}}

\def\checkmark{\tikz\fill[scale=0.4](0,.35) -- (.2,0) -- (0.8,.7) -- (.2,.15) -- cycle;}
\newcommand{\x}{\checkmark}
\newcommand{\self}{Direct}
\newcommand{\anc}{Ancestor}
\newcommand{\fnm}{%
  \tablefootnote{Either the most direct enclosing \emph{file} (\texttt{file}), or most directly enclosing \emph{assembly} (\texttt{internal}), possibly bypassing some namespaces.}%
}

\begin{table}[!b]
  \caption{Languages classified according to the feature model in~\cref{fig:feature-model}.}
  \label{tab:feature-table}
  \centering
  \begin{tabular}{l c c c c c}
    \toprule
                                          & \sLang{Java}  & \sLang{\CSharp} & \sLang{\CPP} & \sLang{Eiffel} & \sLang{Rust} \\
    \midrule
    \sHead{Modules}                       & \x            & \x              & \x           &                & \x           \\
    \sMid{Internal}                       & \self         & \self\fnm       & \self        &                & \anc         \\
    \sHead{Classes}                       & \x            & \x              & \x           & \x             &              \\
    \sMid{Friends}                        &               &                 & \x           & \x             &              \\
    \sMid{Subclass Acc.\@ Mod.}           & \x            & \x              & \x           & \x             &              \\
    \sBot{private}                        & \x            & \x              & \x           & \x             &              \\
    \sBot{protected}                      &               & \x              & \x           &                &              \\
    \sBot{protected} | \sBott{internal}   & \x            & \x              &              &                &              \\
    \sBot{protected} \& \sBott{internal}  &               & \x              & \x           &                &              \\
    \sMid{Extends Clause Acc.\@ Mod.}     &               &                 & \x           &                &              \\
    \sBot{private}                        &               &                 & \x           &                &              \\
    \sBot{protected}                      &               &                 & \x           &                &              \\
    \bottomrule
  \end{tabular}
  \vspace{-1em}
\end{table}

\endgroup

In the Classes category, the three subfeatures denote the three mechanisms for access control:
Friends allow access to other classes by name,
Subclass Access Modifiers are access modifiers on definitions that determine how it is accessible within the class hierarchy (\cref{sec:subclass,sec:combining-subclass-module}), and
Extends Clause Access Modifiers (\cref{sec:inheritance-restriction}) are access modifiers on extends clauses, as seen in \CPP.
The latter two have subfeatures for each concrete keyword associated with the access control mechanism.
For that reason, \texttt{private} and \texttt{protected} occur twice: once on definitions and once on extends clauses.
\Cref{tab:feature-table} classifies several languages according to this scheme.
In the remainder of this paper, we develop AML (Access Modifier Language), a language that covers all features.
To this end, we first introduce scope graphs~(\cref{sec:scope-graphs}), and a base language for AML~(\cref{sec:setup}).

\section{Using Scope Graphs to Model Name Binding in Programs}
\label{sec:scope-graphs}

In the previous section, we sketched the landscape of access modifiers.
This discussion was based largely on prose specifications as well as experiments with compiler implementations.
No language specification we are aware of provides a more rigorous model of accessibility (or even non-lexical name binding).
In this section, we introduce \emph{scope graphs}~\cite{NeronTVW15,AntwerpenPRV18,RouvoetAPKV20,ZwaanA23}, and argue that they provide a suitable framework for such a model.
\Cref{sec:setup} introduces AML (Access Modifier Language), a toy language with a type system defined using scope graphs.
\Cref{sec:modules,sec:subclass,sec:inheritance-restriction,sec:combining-subclass-module} will extend this language with all accessibility features from~\cref{fig:feature-model}.

\subsection{Scope Graphs as A Model for Name Binding}

From a name binding perspective, classes and modules have some similarities.
Each of these constructs can be thought of as introducing a `scope' (region of code), in which declarations live, and in which names can be resolved.
Scopes are related to each other in various ways.
First, modules are related according to their relative position in the abstract syntax tree.
In addition, imports and extends clauses relate arbitrary modules and classes, respectively.
Resolving a reference corresponds to finding a matching declaration in a scope that is reachable from the scope of the reference.
For example, a reference may resolve to a declaration if it lives in a lexically enclosing scope, or in a module that is imported in an enclosing scope.

\emph{Scope graphs}~\cite{NeronTVW15,AntwerpenPRV18,RouvoetAPKV20,ZwaanA23} make this more precise.
In this model, the name binding structure of a program is represented by a graph.
\Cref{fig:transitive} (adapted from Poulsen et al.~\cite[Fig.\@ 1]{PoulsenZH23}) gives an example program and its corresponding scope graph.
A scope is represented by a circular node in the graph.
For example, $\xrootscope$ represents the global scope, and~$\xscoperef{A}$,~$\xscoperef{B}$ and~$\xscoperef{C}$ represent the bodies of modules \id{A}, \id{B}, and \id{C}, respectively.
Scopes are related using labeled, directed edges.
For example, $\xscoperef{A}$ is lexically enclosed by $\xrootscope$, and thus the graph contains an edge from $\xscoperef{A}$ to $\xrootscope$ with label $\lblLEX$.
Similarly,~$\xscoperef{B}$ imports $\xscoperef{A}$, and thus the graph contains an edge $\xscoperef{B} \scopeedget[\lblIMP] \xscoperef{A}$.
Finally, scope graphs contain declarations.
For example, a declaration of~\texttt{i} in scope $\xscoperef{C}$ is represented by the $\xscoperef{C}\ \typeedgeg[\lblVAR]\ \texttt{i} : \tyINT$ edge/node pair.
Similarly, the modules are declared in the root scope (\eg, $\xrootscope \typeedget[\lblMOD] \id{A} \sim \xscoperef{A}$).
The language specification determines which data is included in the declaration.
Similarly, the labels for edges and declarations can be chosen to match the (binding) constructs of the language.

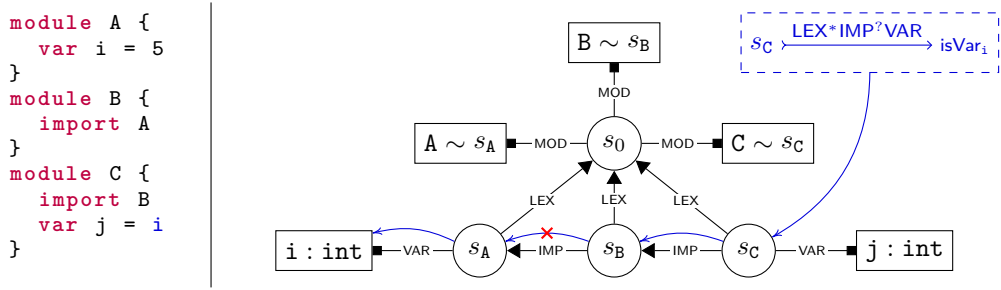
\begin{figure*}[t]
\centering
\begin{minipage}[c]{0.205\textwidth}
\centering
\begin{lstlisting}[language=AML]
module A {
  var i = 5
}
module B {
  import A
}
module C {
  import B
  var j = @1i@
}
\end{lstlisting}
\end{minipage}%
\hspace{0.0125em}%
\vrule%
\hspace{0.5em}%
\begin{minipage}[c]{0.775\textwidth}
\centering
\begin{tikzpicture}[scopegraph, node distance = 2em and 3em]
  \node[scope]              (s0) {$s_0$};

  \node[type,  left  = of s0] (mF) {$\texttt{A} \sim \xscoperef{A}$};
  \draw (s0) edge[type, lbl={\lblMOD}] (mF);

  \node[type,  above = of s0] (mG) {$\texttt{B} \sim \xscoperef{B}$};
  \draw (s0) edge[type, lbl={\lblMOD}] (mG);

  \node[type,  right = of s0] (mH) {$\texttt{C} \sim \xscoperef{C}$};
  \draw (s0) edge[type, lbl={\lblMOD}] (mH);

  \node[scope, below = of s0] (s2) {$s_{\texttt{B}}$};
  \draw (s2) edge[lbl={\lblLEX}] (s0);

  \node[scope, left = of s2] (s1) {$s_\texttt{A}$};
  \draw (s1) edge[lbl={\lblLEX}] (s0);

  \node[scope, right = of s2] (s3) {$s_\texttt{C}$};
  \draw (s3) edge[lbl={\lblLEX}] (s0);

  \draw (s2) edge[lbl={\lblIMP}] (s1);
  \draw (s3) edge[lbl={\lblIMP}] (s2);

  \node[type,  left = of s1] (Hj) {$\texttt{i} : \tyINT$};
  \draw (s1) edge[type, lbl={\lblVAR}] (Hj);

  \node[type, right = of s3] (Fi) {$\texttt{j} : \tyINT$};
  \draw (s3) edge[type, lbl={\lblVAR}] (Fi);

  \node[ref, right = of mG, draw=query1] (qy) {\color{query1}
    $\figquery{\xscoperef{C}}{\mathsf{isVar}_{\id{i}}}{\mathsf{\lblLEX}^{*}\mathsf{\lblIMP}^{?}\mathsf{\lblVAR}}{}$
  };
  \draw (qy) edge[term, color=query1, bend left] (s3);
  \draw (s3) edge[term, color=query1, bend right=20] (s2);
  \draw (s2) edge[term, color=query1, bend right=20, strike through] (s1);
  \draw (s1) edge[term, color=query1, bend right=20] (Hj);
\end{tikzpicture}
\end{minipage}
\caption{Reachability example. The $\mathsf{\lblIMP}^{?}$ part in the regular expression prevents traversal over the second $\lblIMP$ edge.}
\label{fig:transitive}
\end{figure*}

\subparagraph*{Reachability}
To resolve a reference, a \emph{query} is executed to find a valid path in the scope graph from the scope of the reference to a matching declaration.
Queries give specification writers several options to filter paths, to retain only valid paths.
First, a unary predicate selects valid declarations.
Usually, this predicate matches declarations with the name of the reference.
Second, a \emph{regular expression} over labels is used to select valid paths.
This regular expression can, for example, be used to prevent transitive imports, or accessing members in a lexical parent of an imported module.

\Cref{fig:transitive} illustrates this with the query for \id{i} in module \id{C} (dashed blue box).
The parameter on the arrow ($\mathsf{\lblLEX}^{*}\mathsf{\lblIMP}^{?}\mathsf{\lblVAR}$), is a regular expression that defines which paths to declarations are valid.
The $\mathsf{\lblLEX}^{*}$ indicates that a path may traverse an arbitrary number of $\lblLEX$-edges.
This corresponds to looking for variables in enclosing scopes.
Next, the $\mathsf{\lblIMP}^{?}$ part indicates that zero or one $\lblIMP$-edges can be traversed.
Finally, the regular expression ends with $\lblVAR$ to ensure all paths resolve in variable declarations only, excluding \eg modules.
The $\mathsf{isVar}_{\id{i}}$ parameter matches all variable definitions with name \id{i} ($\mathsf{isVar}$ is defined in the next section).
The candidate path (shown as blue edges) does not match this regular expression.
Because $\lblIMP$-labeled edges may only be traversed one time, the step to \xscoperef{A} cannot be made.
In other words: the declaration of \id{i} in \id{A} is not \emph{reachable} from \id{C}.

\subparagraph*{Visibility}
Not every declaration that is reachable (\ie, for which a valid access path exists) can actually be referenced, due to \emph{shadowing}.
For example, in most languages, local definitions have higher priority than imported ones.
We call reachable declarations that are not shadowed by any other declaration \emph{visible}.

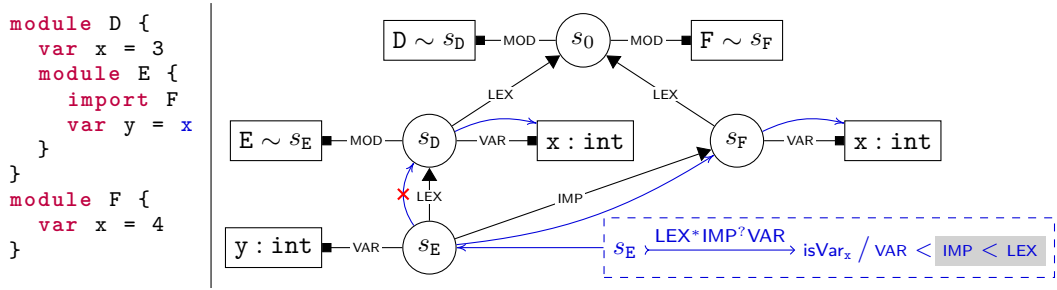
\begin{figure*}[t]
\centering
\begin{minipage}[c]{0.205\textwidth}
\centering
\begin{lstlisting}[language=AML]
module D {
  var x = 3
  module E {
    import F
    var y = @1x@
  }
}
module F {
  var x = 4
}
\end{lstlisting}
\end{minipage}%
\hspace{0.0125em}%
\vrule%
\hspace{0.5em}%
\begin{minipage}[c]{0.775\textwidth}
\centering
\begin{tikzpicture}[scopegraph, node distance = 2em and 3em]
  \node[scope]              (s0) {$s_0$};

  \node[type,  left  = of s0] (mC) {$\texttt{D} \sim \xscoperef{D}$};
  \draw (s0) edge[type, lbl={\lblMOD}] (mC);

  \node[type,  right = of s0] (mE) {$\texttt{F} \sim \xscoperef{F}$};
  \draw (s0) edge[type, lbl={\lblMOD}] (mE);

  \node[scope, below = of mC] (s1) {$\xscoperef{D}$};
  \draw (s1) edge[lbl={\lblLEX}] (s0);

  \node[scope, below = of mE] (s2) {$\xscoperef{F}$};
  \draw (s2) edge[lbl={\lblLEX}] (s0);

  \node[type,  right = of s1] (Cx) {$\texttt{x} : \tyINT$};
  \draw (s1) edge[type, lbl={\lblVAR}] (Cx);

  \node[type,  right = of s2] (Ex) {$\texttt{x} : \tyINT$};
  \draw (s2) edge[type, lbl={\lblVAR}] (Ex);

  \node[type,  left  = of s1] (mD) {$\texttt{E} \sim \xscoperef{E}$};
  \draw (s1) edge[type, lbl={\lblMOD}] (mD);

  \node[scope, below = of s1] (s3) {$\xscoperef{E}$};
  \draw (s3) edge[lbl={\lblLEX}] (s1);
  \draw (s3) edge[lbl={\lblIMP}] (s2);

  \node[type,   left = of s3] (y) {$\texttt{y} : \tyINT$};
  \draw (s3) edge[type, lbl={\lblVAR}] (y);

  \node[ref, right = of s3, draw=query1, xshift = 2.5em] (qx) {\color{query1}
      $\figquery{\xscoperef{E}}{\mathsf{isVar}_{\id{x}}}{\lblLEX^{*}\lblIMP^{?}\lblVAR}{\lblVAR < \hspace{-8pt} \hll{\lblIMP < \lblLEX}}$
  };
  \draw (qx) edge[term, color=query1] (s3);
  \draw (s3) edge[term, color=query1, bend  left=35, strike through] (s1);
  \draw (s1) edge[term, color=query1, bend  left=20] (Cx);
  \draw (s3) edge[term, color=query1, bend right=13] (s2);
  \draw (s2) edge[term, color=query1, bend  left=20] (Ex);
\end{tikzpicture}
\end{minipage}
\caption{Shadowing example. The highlighted label order causes the edge to $\xscoperef{F}$ to have priority.
}
\label{fig:shadowing}
\end{figure*}

In scope graphs, visibility can be encoded using a partial order on labels.
For example, an order $\lblVAR < \lblIMP$ encodes that (local) variable declarations shadow imported declarations.
This is illustrated in~\cref{fig:shadowing}.
The reference \id{x} in module \id{F} can refer to the declaration in module \id{D} as well as the one in module \id{E}.
Because the label order (third argument) indicates that imports shadow lexically enclosing scopes ($\lblIMP < \lblLEX$).
Thus, the variable resolves to the declaration in $\xscoperef{F}$.
Alternatively, if $\lblLEX < \lblIMP$, it would resolve to \id{x} in $\xscoperef{D}$.
Finally, if neither $\lblLEX \not< \lblIMP$ nor $\lblIMP \not< \lblLEX$, both paths would be included in the query result.

In summary, scope graphs model the name binding structure of a program using nodes for scopes and declarations, and edges for relations between those.
Queries can be used to model reference resolution.
A query selects a declaration when
(1) it matches some predicate, and
(2) there exists a path to it of which the labels match a regular expression, and
(3) no other paths that traverse labels with higher priority exist.
The result of a query is a set of paths that lead to these matching declarations.

\subparagraph*{Accessibility}

We can model extensibility using plain scope graphs by including accessibility information in the \emph{declaration}.
In other words, a declaration of a variable in a scope graph contains not only a name and a type, but also its accessibility level.
\emph{After resolution}, we check if the path that the query returns is actually valid according to the accessibility level of the declaration.
For example, if a variable is private, but an $\lblEXT$-edge (for class \emph{extension}) is traversed, an error is emitted.
With this pattern, we can model all accessibility features.


\subparagraph*{Notation}

\Cref{fig:transitive,fig:shadowing} introduce the graphical notation of scope graphs.
In text, variable~$s$ ranges over scopes, and $S$ over sets of scopes.
Moreover, we use the following notation for assertions on scope graphs:
$s_1 \scopeedget[\mathsf{L}] s_2 \in \SG$ means `scope graph $\SG$ has an $\mathsf{L}$-labeled edge from~$s_1$ to $s_2$', and
\smash{$s \typeedget[\mathsf{D}] d \in \SG$} means that $\SG$ has a declaration with data~$d$ under label~$\mathsf{D}$ in scope $s$.
Moreover, we write queries in the following way:
\begin{inferruleraw}
  \centering
  \query{\SG}{s}{\re}{\mathcal{P}}{$\mathcal{O}$}{R}
\end{inferruleraw}
where $\SG$ is the scope graph in which the query is resolved,
$s$ is the scope in which the resolution starts,
$\re$ is the regular expression that paths must adhere to,
and $\mathcal{P}$ is the predicate that declarations must match.
$\mathcal{O}$ is the strict partial order on labels used for shadowing.
It is usually written as $\mathsf{L}_1 < \mathsf{L}_2 < \cdots < \mathsf{L}_n$.
We omit the label order when there is no shadowing.
$R$ is the result set containing tuples of paths and declarations.
When we expect a single result, we use $\setOf{ \langle p, d \rangle }$ to match on the value in the set.
Paths are alternating sequences of scopes and labels, written as $s_1 \pathstept[\mathsf{L}_1] s_2 \cdots s_m$.
Paths do not include the declaration it resolved to, but stop at the scope in which the declaration occurs.
The functions $\mathsf{src}(p)$, $\mathsf{tgt}(p)$ refer to the source and target scope of a path, respectively.
$\mathsf{scopes}(p)$ denotes all scopes in a path.

\section{AML: The Base Language}
\label{sec:setup}

In the next sections, we show how scope graphs support intuitive formalization of accessibility.
We will do so by defining \emph{AML} (Access Modifier Language).
The base syntax (which will be extended later) is given in~\cref{fig:aml-syntax}.
In AML, a program consists of a list of modules.
Each module can define other modules, import other modules, and contain class definitions.
A~class can optionally extend another class, and contains a list of field declarations.
Each field has an access modifier, and is initialized by some expression.
Possible expressions include references, integer constants, class instance creation, field access, and binary operations.

\begin{figure}[!p]
  \setlength{\grammarindent}{5em}
    \begin{minipage}{0.63\textwidth}
      \addtolength{\grammarparsep}{-1pt}
      \begin{grammar}
        <prog> ::= <mod>$^{*}$

        <mod>  ::= "module" <x> "\{" <md>$^{*}$ "\}"

        <md>   ::= <mod> | "import" <x> | <cls>

        <cls>  ::= "class" <x> (":" <acc> <x>)$^{?}$ "\{" <cd>$^{*}$ "\}"

        <cd>   ::= <acc> "var" <x> "=" <e> | <cls>

        <acc>  ::= "public" | \colorbox{gray!33!white}{\vphantom{x}$\ldots$}

        <e>    ::= <n> | <x> | "new" <x> "()" | <e> "." <x> | \ldots
      \end{grammar}
    \end{minipage}%
    \begin{minipage}{0.36\textwidth}
      \addtolength{\grammarparsep}{-3pt}
      \begin{grammar}
        <l>    ::= \lblLEX | \lblIMP | \lblEXT
              \alt \lblMOD | \lblCLS | \lblVAR
              \alt \lblTHISMOD | \lblTHIS

        <d>    ::= "mod" <x> ":" <s>
              \alt "cls" <x> ":" <s>
              \alt "var" <x> ":" <T> "@" <A>
              \alt <s>

        <T>    ::= "int" | "inst" <s>

        <A>    ::= "PUB" | \colorbox{gray!33!white}{\vphantom{x}$\ldots$}
      \end{grammar}%
    \end{minipage}

    \caption{
      Syntax of AML.
      The highlighted positions indicate extensions in later sections.
      The syntax of the complete language can be found in~\AppendixRef{sec:aml-full-spec}{A}.
    }
    \label{fig:aml-syntax}
\end{figure}
\begin{figure}[!p]
    \figuresection[\fbox{$\mathcal{P}(d)$}]{Data Matching Predicates}
      \vspace{-1em}
      \begin{align*}
        \mathsf{isMod}_x(\moddecl{x'}{s}) \, &\Leftarrow \, x \mathbin{=} x' &
        \mathsf{isCls}_x(\clsdecl{x'}{s}) \, &\Leftarrow \, x \mathbin{=} x'\\
        \mathsf{isVar}_x(\defdecl{x'}{T}{A}) \, &\Leftarrow \, x \mathbin{=} x' & %
        \mathsf{isScope}_s(s') \, &\Leftarrow \, s = s'
      \end{align*}
    \figuresection[\fbox{$\premCMem{\SG}{s}{\mathit{cd}}$}]{Class Members}
    \begin{inferrule}
      \inference[\rlabel{d-def}{D-Def}]{
        \premExp{\SG}{s}{e}{T}
        \quad
        \hll{\premAcc{\SG}{s}{\mathit{acc}}{A}}
        \quad
        s \typeedget[\lblVAR] (\defdecl{x}{T}{A}) \in \SG
      }{
        \premCMem{\SG}{s}{\mathit{acc} \mathrel{\lit*{var}} x \mathrel{\lit*{=}} e}
      }%
    \end{inferrule}
    \figuresection[\fbox{$\premExp{\SG}{s}{e}{T}$}]{Type of Expression}
    \begin{inferrule}
      \hspace{-1.75em}
      \inference[\rlabel{t-var}{T-Var}]{
        %
        %
        \query{\SG}{s}{\reLEX}{\mathsf{isVar}_x}{\lblVAR < \lblEXT < \lblLEX}{\setOf{ \langle p, \defdecl{x}{T}{A} \rangle }}
        &
        \hll{\premMAcc{\SG}{s}{p}{A}}
      }{
        \premExp{\SG}{s}{x}{T}
      }
    \end{inferrule}
    \begin{inferrule}
      \inference[\rlabel{t-fld}{T-Fld}]{
        %
        %
        \premExp{\SG}{s}{e}{\tyINST{s_c}}
        \\
        \query{\SG}{s_c}{\reMEM}{\mathsf{isVar}_x}{\lblVAR < \lblEXT}{\setOf{ \langle p, \defdecl{x}{T}{A} \rangle }}
        \\
        \hll{\premMAcc{\SG}{s}{p}{A}}
      }{
        \premExp{\SG}{s}{e.x}{T}
      }
    \end{inferrule}
    \begin{minipage}{\textwidth}
      \vphantom{XXXX}
    \end{minipage}
    \begin{minipage}{0.45\textwidth}
      \figuresection[\fbox{$\premAcc{\SG}{s}{\mathit{acc}}{\mathit{A}}$}]{Access Modifier}
      \vspace{-0.5em}%
      \begin{inferrule}
        %
        %
        \inference[\rlabel{a-pub}{A-Pub}]{
        }{
          \premAcc{\SG}{s}{\lit*{public}}{\PUB}
        }
      \end{inferrule}
    \end{minipage}
    \hfill
    \begin{minipage}{0.45\textwidth}
      \figuresection[\fbox{$\premMAcc{\SG}{s}{p}{A}$}]{Access Policy}
      \begin{inferrule}
        %
        %
        \inference[\rlabel{ap-pub}{AP-Pub}]{
        }{
          \premMAcc{\SG}{s}{p}{\PUB}
        }
      \end{inferrule}
    \end{minipage}
    \vspace{1em}
    \figuresection[\fbox{$\premQMod{\SG}{s}{x}{s_m} \qquad \premQCls{\SG}{s}{x}{s_c}$}]{Module and Class References}
    \begin{inferrule}
      \inference[\rlabel{q-mod}{Q-Mod}]{
        \query{\SG}{s}{
          \reclos{\lblLEX}\lblMOD}{\mathsf{isMod}_x
        }{
          \lblMOD < \lblLEX
        }{
          \setOf{\langle p, \moddecl{x}{s_m} \rangle}
        }
      }{
        \premQMod{\SG}{s}{x}{s_m}
      }
    \end{inferrule}
    \begin{inferrule}
      \inference[\rlabel{q-cls}{Q-Cls}]{
        \query{\SG}{s}{
          \reclos{\lblLEX}\reopt{\lblIMP}\lblCLS
        }{
          \mathsf{isCls}_x
        }{
          \lblCLS < \lblIMP < \lblLEX
        }{
          \setOf{\langle p, \clsdecl{x}{s_c} \rangle}
        }
      }{
        \premQCls{\SG}{s}{x}{s_c}
      }
    \end{inferrule}

    \caption{
      Typing Rules of AML. 
      Accessibility is integrated at the highlighted positions.
      The full type system specification can be found in~\AppendixRef{sec:aml-full-spec}{A}.
    }
    \label{subfig:aml-base-typing}

  \label{fig:aml-base}
\end{figure}

At the right-hand side of~\cref{fig:aml-syntax}, the scope graph parameters are shown.
There are three labels that connect scopes.
\lblLEX denotes lexical scoping, \lblIMP denotes imports, and \lblEXT class extension.
The other three labels are used for declarations.
\lblMOD is used for module declarations, \lblCLS for classes, and \lblVAR for variables/fields.
Next, we assume that each \emph{module scope} has a \lblTHISMOD edge pointing to itself, and similarly, each class has a \lblTHIS scope pointing to itself.
This will be used to resolve enclosing classes or modules.
The sort \synt{d} denotes the data that can be associated with scopes.
Modules and classes are characterized by their name and the scope of their body.
A field has a name, a type \synt{T}, and an accessibility level~\synt{A}.
Scopes that are not declarations implicitly map to themselves.
To query declarations, we use the four predicates shown at the top of~\cref{subfig:aml-base-typing}, which each match a single kind of declaration.
Depending on the type of access control we formalize, different access modifiers will be used.
Therefore, we have left the \synt{acc} and \synt{A} productions partially unspecified.
Each section will instantiate those appropriately.

\subparagraph{Typing Rules}
\Cref{subfig:aml-base-typing} presents some typing rules of AML.
The rules are written in a declarative style, where a scope graph $\SG$ that models the program is assumed.
Constraints over the scope graph are used as premises.
The highlighted premises show where accessibility is integrated into the type system.
We now discuss each of the presented rules.

The~\rref{d-def} rule asserts a declaration is well-typed if the initialization expression $e$ has some type $T$ (first premise),
the access modifier \textit{acc} corresponds to some accessibility policy~\textit{A} (second premise), and
an appropriate declaration exists in the scope graph (third premise).
The accessibility policy is included in the declaration, which enables us to validate accessibility when type checking references.

Next, rule~\rref{t-var} defines how references are type checked in a current scope $s$.
First, it performs a query that looks into the lexical context ($\reclos{\lblLEX}$), parent classes ($\reclos{\lblEXT}$), and eventually resolves to a variable declaration ($\lblVAR$).
It matches only variables with the same name as the reference ($\mathsf{isVar}_x$).
Regarding shadowing, it prefers local variables over variables from a parent class ({$\lblVAR < \lblEXT$}), and
variables from parent classes over variables from enclosing classes ({$\lblEXT < \lblLEX$}).
The query should return a single result, as the name would otherwise be ambiguous.
From this result, the access path $p$, type $T$, and accessibility policy~$A$ are extracted.
The path and the accessibility policy are used in the second (highlighted) premise ($\premMAcc{\SG}{s}{p}{A}$), which asserts that `\emph{accessibility policy} $A$ grants access via path $p$ in scope $s$'.
In future sections, we will define new accessibility policy rules, that may prohibit access of a variable, even if the query premise resolved properly.

Note that, by having accessibility separated from the resolution, we do not capture the interaction between accessibility as shown in~\cref{fig:shadowing-inaccessible-nested}.
We made this choice because the place where accessibility is integrated does not influence the access rules themselves, and this presentation allows more concise derivations, which makes the explanations more accessible.
\AppendixRef{subsec:full-path-order}{A.1} shows how to integrate accessibility in the shadowing policy of a query, and is incorporated in the evaluation (\cref{sec:evaluation}).

For this base language, we only have the \texttt{public} access modifier.
The~\rref{a-pub} rule shows that this keyword corresponds to the \PUB policy.
The meaning of this policy is that access is allowed from any location, with any access path.
This is encoded in the~\rref{ap-pub} rule, which has no premises.

Finally, the last two rules define how references to classes and modules are resolved.
Rule~\rref{q-mod} indicates that module reference $x$ resolves to scope $s_m$ if that scope is included in the closest module declaration with name $x$ in the lexical context.
Similarly, a class reference resolves to the scope of the closest class declaration $s_c$, preferring (non-transitively) imported classes over classes in the lexical context (\rref{q-cls}).

\begin{figure}[!t]
\begin{subfigure}[t]{\textwidth}
\begin{minipage}{0.35\textwidth}
\begin{lstlisting}[language=AML]
class A {
  public var i = 42
}
class B : public A {
  public var j = @1i@
}
\end{lstlisting}
\end{minipage}%
\vrule
\begin{minipage}{0.645\textwidth}
\centering
\begin{tikzpicture}[scopegraph, node distance = 2em and 3em]
  \node[scope] (sA) {$s_\id{A}$};
  \draw (sA) edge[lbl={\lblTHIS}, loop left] (sA);

  \node[scope, below = of sA] (sB) {$s_\id{B}$};
  \draw (sB) edge[lbl={\lblTHIS}, loop left] (sB);
  \draw (sB) edge[lbl={\lblEXT}] (sA);

  \node[type,  right = of sA] (Ax) {$\defdecl{\texttt{i}}{\tyINT}{\PUB}$};
  \draw (sA) edge[type, lbl={\lblVAR}] (Ax);

  \node[type,  right = of sB] (By) {$\defdecl{\texttt{j}}{\tyINT}{\PUB}$};
  \draw (sB) edge[type, lbl={\lblVAR}] (By);

  \node[ref, below = 1em of sB, color=query1] (qy) {\color{query1}
    $\figquery{\xscoperef{B}}{\mathsf{isVar}(\id{i})}{\reLEX}{\lblVAR < \lblEXT < \lblLEX}$
  };
  \draw (qy) edge[term, color=query1] (sB);
  \draw (sB) edge[term, color=query1, bend right=30] (sA);
  \draw (sA) edge[term, color=query1, bend right=20] (Ax.south west);
\end{tikzpicture}
\end{minipage}
\caption{Example program and (partial) scope graph.}
\end{subfigure}

\begin{subfigure}[t]{\textwidth}
  \begin{inferrule}
    \inference{
        \raisebox{-0.75em}{$
          \query{\SG}{\xscoperef{B}}{\ldots}{\mathsf{isVar}_{\color{query1}\texttt{i}}}{\ldots}{
            \setOf{ \langle \xscoperef{B} \pathstept[\lblEXT] \xscoperef{A}, \defdecl{\texttt{i}}{\tyINT}{\PUB} \rangle }
          }
        $}
        &
        \hll{
          \inference{
          }{
            \premMAcc{\SG}{\xscoperef{B}}{\xscoperef{B} \pathstept[\lblEXT] \xscoperef{A}}{\PUB}
          }
        }
    }{
      \premExp{\SG}{\xscoperef{B}}{{\color{query1}\texttt{i}}}{\tyINT}
    }
  \end{inferrule}
  \caption{Part of typing derivation that shows how access is granted by the \PUB accessibility policy.}
  \label{subfig:pub-typing-derivation}
\end{subfigure}

\caption{Example AML program demonstrating the scope graph structure and name resolution with accessibility checking.}
\label{fig:ex-pub}

\end{figure}

\begin{figure}[!t]
  \vspace{1em}
  \figuresection[\fbox{$\premEncM{\SG}{s}{S} \qquad \premEncMI{\SG}{s}{s}$}]{Enclosing Modules}
  \begin{minipage}{\textwidth}
    \begin{inferrule}
      \inference[\rlabel{enc-m}{Enc-M}]{
        \query{\SG}{s}{\reclos{\lblLEX}\lblTHISMOD}{\top}{}{R}
        &
        S_M = \setOf{ s_m \alt \langle p_m, s_m \rangle \in R }
      }{
        \premEncM{\SG}{s}{S_M}
      }
    \end{inferrule}
    \begin{inferrule}
      \inference[\rlabel{enc-mi}{Enc-MI}]{
        \query{\SG}{s}{\reclos{\lblLEX}\lblTHISMOD}{\top}{\lblTHISMOD < \lblLEX}{\setOf{\langle p, s_m \rangle}}
      }{
        \premEncMI{\SG}{s}{s_m}
      }
    \end{inferrule}
  \end{minipage}

  \figuresection[\fbox{$\premEncC{\SG}{s}{S} \qquad \premEncCI{\SG}{s}{s}$}]{Enclosing Classes}
  \begin{minipage}{\textwidth}    
    \begin{inferrule}
      \inference[\rlabel{enc-c}{Enc-C}]{
        \query{\SG}{s}{\reclos{\lblLEX}\lblTHIS}{\top}{}{R}
        \quad
        S_C = \setOf{ s_c \alt \langle p_c, s_c \rangle \in R }
      }{
        \premEncC{\SG}{s}{S_C}
      }
    \end{inferrule}
    \begin{inferrule}
      \inference[\rlabel{enc-ci}{Enc-CI}]{
        \query{\SG}{s}{\reclos{\lblLEX}\lblTHIS}{\top}{\lblTHIS < \lblLEX}{\setOf{\langle p, s_c \rangle}}
      }{
        \premEncCI{\SG}{s}{s_c}
      }
    \end{inferrule}
  \end{minipage}

  \caption{Auxiliary relations for AML scope graphs.}
  \label{fig:aml-auxiliary}

\end{figure}

\subparagraph*{Example}
The example in~\cref{fig:ex-pub} shows two classes \id{A} and \id{B}.
Both classes have a $\lblTHIS$-edge pointing to itself.
Class \id{B} extends class \id{A}, which is represented by the $\xscoperef{B} \scopeedget[\lblEXT] \xscoperef{A}$ edge in the scope graph.
Class \id{A} has a public field \id{i} with type $\tyINT$.
The type as well as the corresponding \PUB access policy are included in the scope graph declaration.
Similarly, class \id{B} has a field \id{j}.
The initialization expression of \id{j} references \id{i}, which is represented with the query shown in the dashed box.

\Cref{subfig:pub-typing-derivation} shows the part of the typing derivation that checks the highlighted reference.
Reference~\id{i} is type checked in scope \xscoperef{B}, and has type $\tyINT$.
The first premise repeats the query shown in the scope graph, with the parameters and result made explicit.
In particular, the resolution path is $\xscoperef{B} \pathstept[\lblEXT] \xscoperef{A}$.
The validity of this path is checked by the second premise, which is satisfied by the~\rref{ap-pub} rule.

\subparagraph*{Auxiliary Relations}
Finally, \cref{fig:aml-auxiliary} presents some auxiliary relations that we will use later.
First, the $\premEncM{\SG}{s}{S_M}$ relation asserts that $S_M$ is the set of scopes of the enclosing modules of~$s$.
It is defined as a query that looks for a $\lblTHISMOD$ edge in the lexically enclosing scopes.
There is no shadowing, so $R$ can contain multiple results in the case of multiple nested modules.
The result $R$ is translated to the set of module scopes by discarding the access paths.

This relation is inhabited for any enclosing module scope.
The second relation $\premEncMI{\SG}{s}{s_m}$ is only inhabited for the \emph{innermost} enclosing module $s_m$.
The query in its definition finds the closest $\lblTHISMOD$-edge, which is enforced by the shadowing policy $\lblTHISMOD < \lblLEX$.
Thus, the query returns only one result, from which the module scope $s_m$ is extracted.
Analogously, $\premEncC{\SG}{s}{S_C}$ relates $s$ to all enclosing \emph{class} scopes $S_C$, and
$\premEncCI{\SG}{s}{s_c}$ is satisfied if~$s_c$ is the \emph{innermost} enclosing class of $s$.

\section{Defining Module Visibility}
\label{sec:modules}

Some languages have access modifiers that regulate the visibility of a declaration in other modules.
For example, in Rust, it is possible to write \texttt{pub(in ...)} to indicate in which module a declaration is visible.
Similarly, some languages support giving particular classes access to an item.
It is the primary accessibility mechanism for Eiffel, and \CPP's \texttt{friend} modifier enables this as well.
Less flexible approaches, such as Java's package visibility and~\CSharp's \texttt{internal} keyword can be seen as special instances of this mechanism.

To demonstrate how these access policies can be encoded using scope graphs, we extend our base language as follows.
\Cref{subfig:internal-syntax} introduces an additional modifier keyword \lit*{internal}, which can contain references to modules.
The declaration is visible in these modules only.
The corresponding accessibility policy \MOD has a set of scopes, each corresponding to a name given in the keyword argument.

\begingroup
\def\endinfskip{0em}

\begin{figure}[!b]
  \begin{subfigure}[t]{\textwidth}
    \centering
    \begin{minipage}{0.5\textwidth}
      \begin{grammar}
        <acc> ::= \ldots | "internal" "(" <x>$^{*}$ ")"
      \end{grammar}
    \end{minipage}%
    \begin{minipage}{0.2\textwidth}
      \begin{grammar}
        <A> ::= "MOD" $S$
      \end{grammar}
    \end{minipage}
    \caption{Syntax of \texttt{internal} keyword.}
    \label{subfig:internal-syntax}
  \end{subfigure}

  \begin{subfigure}[t]{\textwidth}
    \begin{inferrule}
      \inference[\rlabel{a-int}{A-Int}]{
        S = \left\{
          s'  \mathbin{\Big|}
          x_i \in \setVar{x}_{0 \ldots n}, 
          s \vdash_{\SG} x_i \resolvemod s'
        \right\}
      }{
        \premAcc{\SG}{s}{\lit*{internal(} \setVar{x}_{0 \ldots n} \lit*{)}}{\MODof{S}}
      }
      \qquad
      \inference[\rlabel{ap-int}{AP-Int}]{
        \premEncM{\SG}{s}{S_M}
        &
        s_m \in S_M
        &
        s_m \in S
      }{
        \premMAcc{\SG}{s}{p}{\MODof{S}}
      }
    \end{inferrule}
    \vspace{\endinfskip}

    \caption{Semantics of \texttt{internal} keyword.}
    \label{subfig:internal-var-1}
  \end{subfigure}

  \begin{subfigure}[b]{0.645\textwidth}
    \begin{minipage}{0.64\textwidth}
    \begin{inferrule}
      \inference[\rlabel{a-int1}{A-Int'}]{
        \hll{\premEncM{\SG}{s}{S_M}}
        \\
        S = \left\{
          s'  \mathbin{\Big|}
          x_i \in \setVar{x}_{0 \ldots n}, 
          s \vdash_{\SG} x_i \resolvemod s'
          ,
          \hll{s' \in S_M}
        \right\}
      }{
        \premAcc{\SG}{s}{\lit*{internal(} \setVar{x}_{0 \ldots n} \lit*{)}}{\MODof{S}}
      }
    \end{inferrule}
    \end{minipage}

    \caption{Variant 1: Ancestor module only.}
    \label{subfig:internal-var-2}
  \end{subfigure}%
  %
  \begin{subfigure}[b]{0.34\textwidth}
    \begin{minipage}{0.34\textwidth}
    \begin{inferrule}
      \setpremisesend{0.5ex}%
      \addtolength{\jot}{0.45em}%
      \inference[\rlabel{ap-int1}{AP-Int'}]{
        \hll{\premEncMI{\SG}{s}{s_m}}
        \quad
        s_m \in S
      }{
        \premMAcc{\SG}{s}{p}{\MODof{S}}
      }
    \end{inferrule}
    \end{minipage}

    \caption{Variant 2: Innermost module.}
    \label{subfig:internal-var-3}
  \end{subfigure}%

  \begin{subfigure}[b]{\textwidth}
    \begin{inferrule}
      \inference[\rlabel{ap-int2}{AP-Int'{'}}]{
        \cdots
        &
        \hll{
          \left[
            \
            \premEncM{\SG}{s}{S_{M_i}}
            \quad
            s_{m}' \in S_{M_i}
            \quad
            s_{m}' \in S
            \
          \right]_{s' \in \left(\mathsf{scopes}(p) \setminus \setOf{\mathsf{tgt}(p)} \right)}
        }
      }{
        \premMAcc{\SG}{s}{p}{\MODof{S}}
      }
    \end{inferrule}
    \vspace{\endinfskip}

    \caption{Variant 3: Definition exposed to all classes in path.}
    \label{subfig:internal-var-4}
  \end{subfigure}%

  \caption{Extending AML (\cref{fig:aml-base}) with module-level visibility.}
  \label{fig:internal}
\end{figure}

\endgroup

Next, we explain how this keyword is interpreted.
An \lit*{internal} declaration is accessible if the reference occurs in a module that the arguments to the \lit*{internal} modifier give access to.
This is formalized in the rules given in~\cref{subfig:internal-var-1}.
Rule~\rref{a-int} translates an \lit*{internal} access modifier to the \MOD policy.
Each module name argument to the modifier ($x_i$) is resolved relative to the current scope $s$.
This yields a collection of module scopes $s_i$, which are included in the resulting policy.
The~\rref{ap-int} rule encodes that accessing an \lit*{internal} variable is valid if~$s_m$, the scope of some enclosing module of $s$ (the scope of the reference), is in the list of scopes to which access is granted.

\begin{figure}[!t]
\begin{subfigure}[t]{\textwidth}
\begin{minipage}{0.375\textwidth}
\begin{lstlisting}[language=AML]
class A {
  internal(M) var x = 42
}
module M {
  module N {
    class B {
      public var y =
        new C().x
    }
  }
  class C : public A { }
}
\end{lstlisting}
\vspace{-0.5em}
\end{minipage}%
\vrule
\begin{minipage}{0.62\textwidth}
\centering
\vspace{-1em}
\begin{tikzpicture}[scopegraph, node distance = 2em and 3em]
  \node[scope] (s0) {$s_0$};
  \draw (s0) edge[lbl={\lblTHISMOD}, loop left] (s0);

  \node[scope, above right = of s0] (sA) {$s_\id{A}$};
  \draw (sA) edge[lbl={\lblLEX}] (s0);

  \node[type,  right = of sA] (Ax) {$\defdecl{\texttt{x}}{\tyINT}{\MODof{\setOf{\xscoperef{M}}}}$};
  \draw (sA) edge[type, lbl={\lblVAR}] (Ax);

  \node[scope, below right = of s0] (sM) {$s_\id{M}$};
  \draw (sM) edge[lbl={\lblTHISMOD}, loop above] (sM);
  \draw (sM) edge[lbl={\lblLEX}] (s0);

  \node[scope, above right = 2em and 1.5em of sM] (sC) {$s_\id{C}$};
  \draw (sC) edge[lbl={\lblLEX}] (sM);
  \draw (sC) edge[lbl={\lblEXT}] (sA);

  \node[scope, right = of sM] (sN) {$s_\id{N}$};
  \draw (sN) edge[lbl={\lblTHISMOD}, loop above] (sN);
  \draw (sN) edge[lbl={\lblLEX}] (sM);

  \node[scope, right = of sN] (sB) {$s_\id{B}$};
  \draw (sB) edge[lbl={\lblLEX}] (sN);

  \begin{scope}
    \node[type, above = of s0] (dA) {$\clsdecl{\texttt{A}}{\xscoperef{A}}$};
    \draw (s0) edge[type, lbl={\lblCLS}] (dA);
    \draw (dA) edge[dotted, term] (sA);

    \node[type, below = of s0] (dM) {$\moddecl{\texttt{M}}{\xscoperef{M}}$};
    \draw (s0) edge[type, lbl={\lblMOD}] (dM);
    \draw (dM) edge[dotted, term] (sM);

  \end{scope}
\end{tikzpicture}
\end{minipage}
\caption{
  Example program and partial scope graph demonstrating the \lit*{internal} access modifier.
}
\label{subfig:ex-internal}
\end{subfigure}

\begin{subfigure}[t]{\textwidth}
  \begin{inferrule}
    \inference[\rref{a-int}]{
      \xscoperef{A} \vdash_{\SG} \texttt{M} \resolvemod \xscoperef{M}
    }{
      \premAcc{\SG}{\xscoperef{A}}{\lit*{internal(M)}}{\MODof{\setOf{\xscoperef{M}}}}
    }
  \end{inferrule}
  \caption{Part of typing derivation that shows how accessibility policy is derived.}
  \label{subfig:internal-policy-derivation}
\end{subfigure}

\begin{subfigure}[t]{\textwidth}
  \begin{inferrule}
    \inference[\rref{ap-int}]{
      \inference{
        \cdots
      }{
        \premEncM{\SG}{\xscoperef{B}}{
          \setOf{
            \xrootscope, \xscoperef{M}, \xscoperef{N}
          }
        }
      }
      &
      \xscoperef{M} \in \setOf{\xrootscope, \xscoperef{M}, \xscoperef{N}}
      &
      \xscoperef{M} \in {\setOf{\xscoperef{M}}}
    }{
      \premMAcc{\SG}{\xscoperef{B}}{\left(\xscoperef{C} \pathstept[\lblEXT] \xscoperef{A} \right)}{\MODof{\setOf{\xscoperef{M}}}}
    }
  \end{inferrule}
  \caption{Part of typing derivation that shows how access is granted by the \MOD accessibility policy.}
  \label{subfig:internal-typing-derivation}
\end{subfigure}

\caption{Example program demonstrating the meaning of the \lit*{internal} access modifier.}
\label{fig:ex-internal}

\end{figure}

\subparagraph*{Example.}
\Cref{fig:ex-internal} gives an example of an \emph{internal} variable.
Class \id{A} has a field \id{x} that can be accessed from module \id{M}.
In the scope graph, this is indicated with the access policy $\MODof{\setOf{\xscoperef{M}}}$ on the corresponding declaration in $\xscoperef{A}$.
The derivation of this policy is shown in~\cref{subfig:internal-policy-derivation}.
Module \id{M} contains a nested module \id{N}, which contains a class \id{B}.
In class \id{B}, the field \id{x} is accessed on an instance of \id{A}.
The (partial) typing derivation in~\cref{subfig:internal-typing-derivation} shows this access is allowed by the~\rref{ap-int} rule.
The first premise asserts that $\xrootscope$, $\xscoperef{M}$ and $\xscoperef{N}$ are the enclosing modules of $\xscoperef{B}$.
This can be seen in the scope graph, as those scopes are reachable via paths with regular expression $\reclos{\lblLEX}\lblTHISMOD$ (\cref{fig:aml-auxiliary}).
As~$\xscoperef{M}$ occurs both in the enclosing modules and in the access policy, access is allowed.

\subparagraph*{Variant 1.}
Several variations on this scheme are conceivable.
For example, languages can restrict the modules to which an \lit*{internal} modifier may expose a declaration.
For example, Rust has the \syntax{\lit*{pub}\lit*{(}\lit*{in} \synt{path}\lit*{)}} visibility modifier, similar to how we defined \lit*{internal}.
However, at the \synt{path} position, only ``an ancestor module of the item whose visibility is being declared'' is allowed~\cite[\S12.6]{RustVisPriv}.
This is encoded in~\cref{subfig:internal-var-2}.
Compared to~\rref{a-int}, this rule adds premises (highlighted) that guarantee that the arguments of the \lit*{internal} modifier ($x_i$) resolve to an enclosing module ($s_i \in S_M$).

Note how these premises would make the example fail to type-check.
Only $\xrootscope$ is an enclosing module of $\xscoperef{A}$.
In particular, the derivation in~\cref{subfig:internal-policy-derivation} would have an additional premise $\xscoperef{A} \in \setOf{s_0}$, which is clearly unsatisfiable.

\subparagraph*{Variant 2.}
Next, consider the example in~\cref{subfig:ex-internal} again.
In the system above, \texttt{x} is accessible in \texttt{B}, because \texttt{x} is exposed to one of its enclosing modules (\id{M}).
However, $\xscoperef{M}$ is not its \emph{innermost} enclosing module.
Such a more lenient accessibility scheme might be desirable (\eg, Rust has this behavior), but languages such as Java do not allow this.
To model these languages, we instead use the premise that asserts $s_m$ is the \emph{innermost} enclosing module scope.
The rule for this variant is given in~\cref{subfig:internal-var-3}.

With this addition, the example would fail to type-check as well.
The access validation (\cref{subfig:internal-typing-derivation}) would now have to satisfy $\premEncMI{\SG}{\xscoperef{B}}{\xscoperef{M}}$, which is impossible, as $\xscoperef{N}$ is the innermost enclosing module.

\subparagraph*{Variant 3.}
Finally, consider example~\cref{fig:inheritance-packages} from the introduction again.
In this example, the reference to~\id{x} in class~\id{C} was not valid, as~\id{B} (by virtue of residing in a different package), did not inherit~\id{x}.
The (partial) rule in~\cref{subfig:internal-var-4} covers this case.
For each scope in the path (apart from the target), it adds premises that assert that the definition is exposed to that scope (similar to $s$ in~\cref{subfig:internal-var-1}).%
\footnote{Alternatively, the premises of~\cref{subfig:internal-var-3} can be used when direct exposure is required.}
The target is excluded because it is not inheriting the accessed field, but rather defining it.
(Recall that paths move from reference to declaration, so the target is the scope of the defining class.)
For that reason, there is no need to assert it inherits the field.

When adding this rule fragment to the derivation in~\cref{subfig:internal-typing-derivation}, there will be additional premises that validate that class~\id{C} inherits~\id{x}.
This is the case, as~\id{C} resides in module~\id{M}.

\section{Defining Subclass Visibility}
\label{sec:subclass}

Next, we consider how to define access modifiers that regulate access from other \emph{classes}:
the \lit*{private} modifier (\cref{subsec:private}), and the \lit*{protected} keyword (\cref{subsec:protected}).

\subsection{Private Access}
\label{subsec:private}

The \texttt{private} access modifier is slightly challenging to define, as languages implement it differently.
For example, \CSharp allows accessing private variables on instances of \emph{subclasses}, whereas Java does not.
Consider the example programs in~\cref{fig:diff-priv-sub}.
In the Java case, the access \texttt{b.x} is invalid, because it only allows access on instances of \texttt{A}.

On the other hand, Java exposes \texttt{private} members to the \emph{outermost} enclosing class\footnote{
  ``{[When]} the member or constructor is declared \texttt{private}, (...) access is permitted if and only if it occurs within the body of the top level class [sic!] that encloses the declaration of the member or constructor."~\cite[\S6.6.1]{JLS8}
}, while \CSharp only exposes members to the \emph{defining} (\ie, innermost enclosing) class (including nested classes), as shown in~\cref{fig:diff-priv-enc}.

\begin{figure}[t]
\begin{minipage}{0.5\textwidth}
\begin{lstlisting}[language=Java]
class A {
    private int x = 42;
    public int accessX(B b) {
        return b.x; // ERROR!
    }
}
class B extends A { }
\end{lstlisting}
\vspace{-1em}
\end{minipage}%
\begin{minipage}{0.5\textwidth}
\begin{lstlisting}[language={[Sharp]C}]
class A {
    private int x = 42;
    public int AccessX(B b) {
        return b.x; // fine
    }
}
class B : A { }
\end{lstlisting}
\vspace{-1em}
\end{minipage}%
\caption{Difference in \texttt{private} member access of subclass instances between Java and \CSharp.}
\label{fig:diff-priv-sub}
\end{figure}

\begin{figure}[t]
\begin{minipage}{0.5\textwidth}
\begin{lstlisting}[language=Java]
class A {
    class B {
        private int x = 42;
    }
    int accessX(B b) {
        return b.x; // fine
    }
}
\end{lstlisting}%
\vspace{-1em}
\end{minipage}%
\begin{minipage}{0.5\textwidth}
\begin{lstlisting}[language={[Sharp]C}]
class A {
    class B {
        private int x = 42;
    }
    int AccessX(B b) {
        return b.x; // ERROR!
    }
}
\end{lstlisting}%
\vspace{-1em}
\end{minipage}%
\caption{Difference in \texttt{private} member access from enclosing class between Java and \CSharp.}
\label{fig:diff-priv-enc}
\end{figure}%

\begin{figure}[!t]
  \begin{subfigure}[t]{0.55\textwidth}
    \begin{minipage}{\textwidth}
      \centering
      \begin{minipage}{0.525\textwidth}
        \begin{grammar}
          <acc> ::= $\ldots$ | "private"
        \end{grammar}
      \end{minipage}%
      \begin{minipage}{0.35\textwidth}
        \begin{grammar}
          <A> ::= $\ldots$ "PRV"
        \end{grammar}
      \end{minipage}
    \end{minipage}%

    \vspace{0.45em}

    \caption{Syntax of \texttt{private} keyword.}
    \label{subfig:private-syntax}
  \end{subfigure}%
  \begin{subfigure}[t]{0.4\textwidth}
    \vspace{-2.5em}
    \begin{inferrule}[0.4\textwidth]
      \inference[\rlabel{a-priv}{A-Priv}]{
      }{
        \premAcc{\SG}{s}{\lit*{private}}{\PRV}
      }
    \end{inferrule}

    \caption{\texttt{private} to \PRV access policy.}
    \label{subfig:private-policy}
  \end{subfigure}
  \begin{subfigure}[b]{0.55\textwidth}
    \begin{inferrule}[0.55\textwidth]
      \inference[\rlabel{ap-priv}{AP-Priv}]{
        \premEncC{\SG}{s}{S_C}
        &
        \mathsf{tgt}(p) \in S_C
      }{
        \premMAcc{\SG}{s}{p}{A}
      }
    \end{inferrule}

    \caption{Semantics of \texttt{private} keyword.}
    \label{subfig:private-semantics}
  \end{subfigure}%
  \begin{subfigure}[b]{0.45\textwidth}
    \begin{inferrule}[0.45\textwidth]
      \inference[\rlabel{ap-priv1}{AP-Priv'}]{
        \ldots
        &
        \hll{\premMatchRE{p}{\reclos{\lblLEX}}}
      }{
        \premMAcc{\SG}{s}{p}{A}
      }
    \end{inferrule}

    \caption{Prevent access on instances of subclasses.}
    \label{subfig:current-instance}
  \end{subfigure}

  \begin{subfigure}[t]{\textwidth}
    \centering
    \begin{inferrule}
      \inference[\rlabel{ap-priv2}{AP-Priv''}]{
        \premEncC{\SG}{s}{S_{C_{\mathit{ref}}}}
        &
        \hll{\premEncC{\SG}{\mathsf{tgt}(p)}{S_{C_{\mathit{decl}}}}}
        &
        s_c \in S_{C_{\mathit{ref}}}
        &
        s_c \in S_{C_{\mathit{decl}}}
      }{
        \premMAcc{\SG}{s}{p}{A}
      }
    \end{inferrule}

    \caption{Allow access from enclosing classes.}
    \label{subfig:private-enclosing-class}
  \end{subfigure}

  \caption{Extending AML (\cref{fig:aml-base}) with \texttt{private} visibility.}
  \label{fig:private}
\end{figure}

We start with modeling the \CSharp semantics 
in~\cref{subfig:private-syntax,subfig:private-policy,subfig:private-semantics}.
Rule \rref{ap-priv} states that the class in which the field is declared (which is the target of the
path $\mathsf{tgt}(p)$) should be an enclosing class of the scope in which the access occurs.
This permits access from nested classes of $\mathsf{tgt}(p)$, but does not expose it to enclosing classes.
On the other hand, access on instances of subclasses is allowed, as there are no constraints on the structure of the path.

Note that we did not specify that this rule matches on the~\PRV policy specifically,
but rather applies to \emph{any} access policy $A$.
This is a deliberate choice; it adds the possibility of using this rule as a fallback in case no other rule works.
This ensures other accessibility policies will never be more strict than~\PRV, which corresponds to general intuition.
By matching on an arbitrary $A$ in \rref{ap-priv}, we simplify the definition of the other policies, as they otherwise would need to define special rules for \texttt{private}-like access.

\subparagraph*{Current Instance.}
Now, we adapt these rules to match the Java semantics.
First,~\cref{subfig:current-instance} shows how to prevent access to the private field on instances of subclasses (\cref{fig:diff-priv-sub}).
It uses a new type of constraint, $\premMatchRE{p}{\re}$, which holds when the sequence of labels in path $p$ is in the language described by the regular expression $\re$.
In this case, we assert that the access path $p$ must adhere to the regular expression $\reclos{\lblLEX}$.
This prevents access from instances of subclasses of the defining class, as that requires traversing an $\lblEXT$ edge.
For example, the access path in~\cref{fig:diff-priv-sub} would be $\premMatchRE{\xscoperef{B} \scopeedget[\lblEXT] \xscoperef{A}}{\reclos{\lblLEX}}$, which is not satisfiable.

\subparagraph{Outermost Class.}
Finally,~\cref{subfig:private-enclosing-class} shows how to expose \texttt{private} fields to the outermost enclosing class.
In this rule, the set $S_{C_\mathit{ref}}$ contains the scope of the enclosing classes of the reference location,
and $S_{C_\mathit{decl}}$ contains the scope of the enclosing classes of the class in which the declaration occurs.
These sets should share a scope $s_c$, which represents the shared enclosing class of the reference and the declaration.

Note how this rule enables type-checking the program in~\cref{fig:diff-priv-enc}.
Using \rref{ap-priv} does not work, as $\premEncC{\SG}{\xscoperef{A}}{\setOf{\xscoperef{A}}}$,
which does not include $\mathsf{tgt}(p) = \xscoperef{B}$.
However, we can check it with \rref{ap-priv2}, as $\premEncC{\SG}{\mathsf{tgt}(p)}{\setOf{\xscoperef{B}, \xscoperef{A}}}$, which includes the shared enclosing class~$\xscoperef{A}$.

\begin{figure}[!t]
  \begin{subfigure}[b]{\textwidth}
    \begin{minipage}{\textwidth}
      \centering
      \begin{minipage}{0.4\textwidth}
        \begin{grammar}
          <acc> ::= $\ldots$ | "protected"
        \end{grammar}
      \end{minipage}%
      \begin{minipage}{0.3\textwidth}
        \begin{grammar}
          <A> ::= $\ldots$ | "PRT"
        \end{grammar}
      \end{minipage}
    \end{minipage}

    \caption{Syntax of \texttt{protected} keyword.}
    \label{subfig:protected-syntax}
  \end{subfigure}

  \begin{subfigure}[b]{\textwidth}
    \begin{inferrule}
      \inference[\rlabel{a-prot}{A-Prot}]{
      }{
        \premAcc{\SG}{s}{\lit*{protected}}{\PRT}
      }
      \qquad
      \inference[\rlabel{ap-prot}{AP-Prot}]{
        \premEncC{\SG}{s}{S_C}
        &
        s_c \in S_C
        &
        s_c \in \mathsf{scopes}(p)
      }{
        \premMAcc{\SG}{s}{p}{\PRT}
      }
    \end{inferrule}

    \caption{Semantics of \texttt{protected} keyword.}
    \label{subfig:protected-semantics}
  \end{subfigure}%

  \caption{Extending AML (\cref{fig:aml-base}) with \texttt{protected} visibility.}
  \label{fig:protected}
\end{figure}

\subsection{Protected Access}
\label{subsec:protected}

The \texttt{protected} access modifier (\cref{subfig:protected-syntax}) grants access to subclasses of the defining class, including classes nested in subclasses.
For field access expressions (\synt{e.x}), \textit{e} must be an instance of a class that encloses \emph{the reference}~\cite[\S6.6.2.1]{JLS8}.
This semantics (\cref{subfig:protected-semantics}) can be modeled by asserting that there should be some class $s_c$ that is both
(a) an enclosing scope of the reference location ($\premEncC{\SG}{s}{S_C}$), and
(b) occurs in the in the access path ($s_c \in \mathsf{scopes}(p)$).
The last condition implies that the enclosing class $s_c$ is a subclass of the \emph{defining class},
which is the intuitive understanding of the \texttt{protected} keyword.

\begin{figure}[!t]
\begin{subfigure}[t]{\textwidth}
\begin{minipage}{0.35\textwidth}
\begin{lstlisting}[language=AML]
class A {
  protected var x = 42
}
class B : public A {
  class I {
    public int f(b: B) {
      return b.x;
    }
  }
}
\end{lstlisting}
\end{minipage}%
\vrule
\begin{minipage}{0.645\textwidth}
\centering
\begin{tikzpicture}[scopegraph, node distance = 2em and 3em]
  \node[scope] (sA) {$s_\id{A}$};
  \draw (sA) edge[lbl={\lblTHIS}, loop left] (sA);

  \node[type,  right = of sA] (Ax) {$\defdecl{\texttt{x}}{\tyINT}{\PRT}$};
  \draw (sA) edge[type, lbl={\lblVAR}] (Ax);

  \node[scope, below = of sA] (sB) {$s_\id{B}$};
  \draw (sB) edge[lbl={\lblEXT}] (sA);
  \draw (sB) edge[lbl={\lblTHIS}, loop left] (sB);

  \node[scope, right = of sB] (sI) {$s_\id{I}$};
  \draw (sI) edge[lbl={\lblLEX}] (sB);
  \draw (sI) edge[lbl={\lblTHIS}, loop above] (sI);

  \node[scope, right = of sI] (sf) {$s_\id{f}$};
  \draw (sf) edge[lbl={\lblLEX}] (sI);

  \node[type,  below = of sf] (fb) {$\defdecl{\texttt{b}}{\tyINST{\xscoperef{B}}}{\PUB}$};
  \draw (sf) edge[type, lbl={\lblVAR}] (fb);
\end{tikzpicture}
\end{minipage}
\caption{
  Example program and partial scope graph demonstrating the \lit*{protected} access modifier.
}
\label{subfig:prot-scope-graph}
\end{subfigure}

\begin{subfigure}[t]{\textwidth}
  \begin{inferrule}
    \inference{
      \inference{
        \cdots
      }{
        \premEncC{\SG}{\xscoperef{f}}{\setOf{\xscoperef{I}, \xscoperef{B}}}
      }
      &
      \xscoperef{B} \in \setOf{\xscoperef{I}, \xscoperef{B}}
      &
      \xscoperef{B} \in \mathsf{scopes}\left(\xscoperef{B} \pathstept[\lblEXT] \xscoperef{A}\right)
    }{
      \premMAcc{\SG}{\xscoperef{f}}{\left(\xscoperef{B} \pathstept[\lblEXT] \xscoperef{A}\right)}{\PRT}
    }
  \end{inferrule}
  \caption{
    Part of typing derivation that shows how access is granted by the \PRT accessibility policy.
  }
  \label{subfig:prot-typing-derivation}
\end{subfigure}

\caption{Example program demonstrating the meaning of the \lit*{protected} access modifier.}
\label{fig:ex-protected}

\end{figure}

\Cref{fig:ex-protected} demonstrates how this rule works.
In this program, there is a class~\id{A} which has a subclass~\id{B}.
Class~\id{B} has a nested class~\id{I}, which has a method~\id{f} with a parameter~\id{b} of type~\id{B}.
The body of~\id{f} accesses field~\id{x} on the instance of \id{B}.
On the right-hand side of the picture, a part of the corresponding scope graph is shown.
The scopes for classes~$\id{A}$ and~\id{B} are connected by an $\lblEXT$-edge again.
The fact that class~\id{I} is nested in class~\id{B} is represented by the $\xscoperef{I} \scopeedget[\lblLEX] \xscoperef{B}$ edge, similar to other lexically nested constructs.
Likewise, scope $\xscoperef{f}$, which represents the body of the method~\id{f}, has a $\lblLEX$-edge to $\xscoperef{I}$.

\Cref{subfig:prot-typing-derivation} shows how the access to~\texttt{b.x} is validated.
The first premise states that~$\xscoperef{I}$ and~$\xscoperef{B}$ are the enclosing classes of~$\xscoperef{f}$.
The other premises assert that~$\xscoperef{B}$ is in the enclosing classes as well as in the access path.
Together, this allows access to the protected member.
Note how access to an instance of~\id{A} in~$\xscoperef{f}$ would not be allowed.
In that case, the access path would have been just~$\xscoperef{A}$, which is not an enclosing class of $\xscoperef{f}$.

\section{Combining Subclass and Module Visibility}
\label{sec:combining-subclass-module}

Access modifiers regulating both the module and subclass dimensions occur in real-world languages as well.
For example (as noticed earlier), Java's \texttt{protected} keyword also exposes a definition in the same package, similar to \CSharp's \texttt{protected internal}.
In addition, \CSharp has a \texttt{private protected} modifier, which allows access to subclasses in the same assembly only.
In fact, those two keywords denote the two main ways in which access modifiers can interact.
First, \texttt{protected internal} denotes \emph{disjunctive} interaction, where a declaration is accessible from the subclasses \emph{or} the same module.
Second, \texttt{private protected} denotes \emph{conjunctive} interaction, where a declaration is accessible from the subclasses \emph{in} the same module only.
These interactions are straightforward to define, with one intricate case discussed below.

\begin{figure}[!t]
  \begin{subfigure}[b]{\textwidth}
    \begin{minipage}{\textwidth}
      \centering
      \begin{minipage}{0.85\textwidth}
        \setlength{\grammarindent}{5em}
        \begin{grammar}
          <acc> ::= \ldots | "protected" "internal" "(" <x>$^{*}$ ")" | "private" "protected" "(" <x>$^{*}$ ")"

          <A>   ::= \ldots  | "SMD" $S$ | "SMC" $S$
        \end{grammar}%
      \end{minipage}
    \end{minipage}

    \caption{Syntax of policy interaction keywords.}
    \label{subfig:interaction-syntax}
  \end{subfigure}

  \begin{subfigure}[b]{\textwidth}
    \vspace{0.25em}
    \begin{inferrule}
      \inference[\rlabel{a-pprot}{A-PProt}]{
        %
        S = \left\{
          s'  \mathbin{\Big|}
          x_i \in \setVar{x}_{0 \ldots n}, 
          s \vdash_{\SG} x_i \resolvemod s'
        \right\}
      }{
        \premAcc{\SG}{s}{\lit*{private protected(} \setVar{x}_{0 \ldots n} \lit*{)}}{\SMCof{S}}
      }
    \end{inferrule}
    \vspace{0.5em}
    \begin{inferrule}
      \inference[\rlabel{a-pint}{A-PInt}]{
        %
        S = \left\{
          s'  \mathbin{\Big|}
          x_i \in \setVar{x}_{0 \ldots n}, 
          s \vdash_{\SG} x_i \resolvemod s'
        \right\}
      }{
        \premAcc{\SG}{s}{\lit*{protected internal(} \setVar{x}_{0 \ldots n} \lit*{)}}{\SMDof{S}}
      }
    \end{inferrule}

    \caption{Translation of composite keywords to their policies.}
    \label{subfig:interaction-translation}
  \end{subfigure}

  \begin{subfigure}[b]{\textwidth}
    \begin{inferrule}
      \inference[\rlabel{ap-smc}{AP-SMC}]{
        \premMAcc{\SG}{s}{p}{\MODof{S}}
        &
        \premMAcc{\SG}{s}{p}{\PRT}
      }{
        \premMAcc{\SG}{s}{p}{\SMCof{S}}
      }
    \end{inferrule}
    \vspace{0.5em}
    \begin{inferrule}
      \inference[\rlabel{ap-smd-prot}{AP-SMD-Prot}]{
        \premMAcc{\SG}{s}{p}{\PRT}
      }{
        \premMAcc{\SG}{s}{p}{\SMDof{S}}
      }
      \qquad
      \inference[\rlabel{ap-smd-mod}{AP-SMD-Mod}]{
        \premMAcc{\SG}{s}{p}{\MODof{S}}^{(*)}
      }{
        \premMAcc{\SG}{s}{p}{\SMDof{S}}
      }
    \end{inferrule}

    \caption{Semantics of interaction policies.}
    \label{subfig:interaction-semantics}
  \end{subfigure}%

  \caption{Extending AML (\cref{fig:aml-base}) with keywords to combine module-level and subclass-level accessibility.}
  \label{fig:interaction}

  \az{Lot of redundant (v)space in this figure.}
\end{figure}

\Cref{subfig:interaction-syntax} defines the syntax of the two new keywords (based on their name in \CSharp) and policies.
We add \SMD (Subclass/Module, Disjunctive) and \SMC (Subclass/Module, Conjunctive) policies, which each contain a list of module scopes to which they are exposed.
The translation from keyword to policy is given in~\cref{subfig:interaction-translation}.
Both rules resolve their module arguments, similar to~\rref{a-int}.
The \SMC policy has one rule (\rref{ap-smc}), which simply asserts that access is granted by the module (\MOD) and protected (\PRV) policies.
There are two rules for the \SMD policy.
The first one simply delegates to the \PRT access policy, permitting access wherever a \texttt{protected} member would have been accessible.
The other rule delegates to the \MOD policy, but more careful attention must be paid here (hence the $(*)$ mark).
Recall that the semantics of this policy has a variant that asserts that the whole inheritance chain has access to the declaration (\cref{subfig:internal-var-4}).
However, this extension should \emph{not} be applied here, because the \texttt{protected} part of this modifier already grants access, regardless of the module-level exposure.

\section{Defining Extends-Clause Accessibility Restriction}
\label{sec:inheritance-restriction}

Until now, we have only considered inheritance as it exists in Java and \CSharp.
In this section, we shift our focus to \CPP, in particular the access modifiers on extends clauses.
In \CPP, it is possible to add a \texttt{private} modifier to an extends clause, which reduces the accessibility of \texttt{public} and \texttt{protected} members to \texttt{private} in the derived class.
Similarly, the \texttt{protected} keyword can be used to reduce the accessibility of \texttt{public} members to \texttt{protected}.
For qualified accesses, \CPP imposes the additional constraint that the inheritance chain leading to class in which the variable is declared should be accessible from the class in which the access occurs~\cite[\S11.9.3 (4)]{CPP20ISO}.

\subparagraph*{Setup.}
In contrast to the previous sections, we cannot encode inheritance-imposed access control in our accessibility policy $A$.
Instead, we encode it in the scope graph directly.
For that purpose, we introduce two new labels: $\lblEXTPRIV$ and $\lblEXTPROT$, which model private and protected extension, respectively.
Similar to the previous sections, $\lblEXT$ will model public extension; \ie inheritance without access restriction.

Fortunately, we can validate path access independently from the declaration-level access policy.%
\footnote{
  That also holds for the subtle interaction between \texttt{internal} and \texttt{protected} discussed in~\cref{sec:combining-subclass-module}.
  \texttt{protected} or \texttt{private} inheritance in subclasses of the reference class can still compromise these access modes, and must therefore be validated.
}
We require two adaptations to the rules~\rref{t-var} and \rref{t-fld} (\cref{fig:aml-base}).
First, the regular expressions of the queries must be changed to also traverse these new edges.
Thus, in \rref{t-var}, $\reLEX$ must be changed to $\reLEXCPP$.
Similarly, \rref{t-fld} now has \reMEMCPP as regular expression instead of \reMEM.
Second, we add a premise $\premMPth{\SG}{s}{p}$ to both rules.
This premise asserts that the labels in the path $p$ do not hide the accessed definition in scope $s$.

 \begin{figure}[!t]
  \begin{inferrule}
    \inference[\rlabel{p-pub}{P-Pub}]{
      \premMatchRE{p}{\reclos{\lblLEX}\reclos{\lblEXT}}
    }{
      \premMPth{\SG}{s}{p}
    }
    \qquad
    \inference[\rlabel{p-priv-prot}{P-Priv-Prot}]{
      \premEncC{\SG}{s}{S_C}
      &
      s_c \in S_C
      &
      \splitAt{s_c}{p} = \langle p_1, p_2 \rangle
      \\
      \premMatchRE{p_1}{\reclos{\lblLEX}\reclos{\lblEXT}}
      &
      \premMatchRE{p_2}{\reopt{\lblEXTPRIV}\reclos{(\lblEXT|\lblEXTPROT)}}
    }{
      \premMPth{\SG}{s}{p}
    }
  \end{inferrule}

  \caption{Extending AML (\cref{fig:aml-base}) with path-level visibility.}
  \label{fig:path-access}
\end{figure}

Path accessibility can be captured in two rules, shown in~\cref{fig:path-access}.
First, \rref{p-pub} asserts that a path is valid when there is only public inheritance.
With this rule, the semantics of the programs that do not use private or protected inheritance has not changed.
Second, rule \rref{p-priv-prot} covers the other two cases.
This rule looks intricate, but the intuition behind it is not too complicated.
Similar to the \texttt{private} and \texttt{protected} modifiers (\cref{subsec:private,subsec:protected}), access must occur within the class where the member is \texttt{private}/\texttt{protected}.
This is now not necessarily the defining class, but rather the last class in the inheritance chain that has a non-public modifier on the extends clause.
In the rule, this is encoded as follows.
The first two premises introduce a scope $s_c$, which is an enclosing scope of the reference location $s$.
The third premise asserts that the path $p$ can be split into two segments at scope~$s_c$.
That is, $p$ consists of two segments: a part $p_1$ from~$s_1$ to~$s_c$ and a part $p_2$ from~$s_c$ to~$s_n$.
This implies that $s_c$ is in the access path.
To validate that all subclasses of $s_c$ in the path have public inheritance, $p_1$ should match regular expression $\reclos{\lblLEX}{\lblEXT}$.%
\footnote{Alternatively, one can encode the requirement that the instance type must be $s_c$ itself by using $\reclos{\lblLEX}$, similar to~\cref{subfig:current-instance}.}
The path leading from the current class to the declaration ($p_2$) may start with a private inheritance step ($\reopt{\lblEXTPRIV}$), but may have only public and protected inheritance higher in the access path.

\begin{figure}[!t]
\begin{subfigure}[t]{\textwidth}
\begin{minipage}{0.45\textwidth}
\begin{lstlisting}[language=AML]
class A {
  public var x = 42
}
class B : private A {
  public var y = new C().x
}
class C : public B { }
\end{lstlisting}
\end{minipage}%
\vrule
\begin{minipage}{0.545\textwidth}
\centering
\begin{tikzpicture}[scopegraph, node distance = 2em and 4em]
  \node[scope] (sA) {$s_\id{A}$};
  \draw (sA) edge[lbl={\lblTHIS}, loop left] (sA);

  \node[type,  above = of sA] (Ax) {$\defdecl{\texttt{x}}{\tyINT}{\PUB}$};
  \draw (sA) edge[type, lbl={\lblVAR}] (Ax);

  \node[scope, right = of sA] (sB) {$s_\id{B}$};
  \draw (sB) edge[lbl={\vphantom{\lblEXT}\smash{$\lblEXTPRIV$}}] (sA);

  \draw (sB) edge[lbl={\lblTHIS}, loop above] (sB);

  \node[scope, right = of sB] (sC) {$s_\id{C}$};
  \draw (sC) edge[lbl={\lblEXT}] (sB);
  \draw (sC) edge[lbl={\lblTHIS}, loop above] (sC);
\end{tikzpicture}
\end{minipage}
\caption{
  Example program and partial scope graph demonstrating path access restrictions.
}
\label{subfig:path-scope-graph}
\end{subfigure}

\begin{subfigure}[t]{\textwidth}
  \begin{inferrule}
    \inference{
      \inference{
        \cdots
      }{
        \premEncC{\SG}{\xscoperef{B}}{\setOf{\xscoperef{B}}}
      }
      &
      \xscoperef{B} \in \setOf{\xscoperef{B}}
      \\
      \raisebox{0.0em}{$\displaystyle\vphantom{\frac{A}{\frac{B}{C}}}$}
      \splitAt{\xscoperef{B}}{
        \xscoperef{C} \pathstept[\lblEXT] \xscoperef{B} \pathstept[\lblEXTPRIV] \xscoperef{A}
      } = \langle
        \xscoperef{C} \pathstept[\lblEXT] \xscoperef{B},
        \xscoperef{B} \pathstept[\lblEXTPRIV] \xscoperef{A}
      \rangle
      \\
      \premMatchRE{\left(\xscoperef{C} \pathstept[\lblEXT] \xscoperef{B}\right)}{\reclos{\lblLEX}\reclos{\lblEXT}}
      &
      \premMatchRE{\left(\xscoperef{B} \pathstept[\lblEXTPRIV] \xscoperef{A}\right)}{\reopt{\lblEXTPRIV}\reclos{(\lblEXT|\lblEXTPROT)}}
    }{
      \premMPth{\SG}{\xscoperef{B}}{\left(\xscoperef{C} \pathstept[\lblEXT] \xscoperef{B} \pathstept[\lblEXTPRIV] \xscoperef{A}\right)}
    }
  \end{inferrule}
  \vspace{-0.5em}
  \caption{
    Part of typing derivation that shows how access is granted by the \rref{p-priv-prot} rule.
  }
  \label{subfig:path-typing-derivation}
\end{subfigure}

\caption{Example program demonstrating path accessibility.}
\label{fig:ex-path}

\end{figure}

\Cref{fig:ex-path} gives an example that uses this rule.
There is a class \id{A} with a field \id{x}.
Class \id{A} is inherited privately by class~\id{B}, which makes~\id{x} private in~\id{B}.
Next, class~\id{C} extends~\id{B} publicly.
In class~\id{B},~\id{x} is accessed on an instance of~\id{C}.
This access should be allowed, as class~\id{B} is the class in which~\id{x} is private as well as the class in which the reference occurs.
The partial derivation in~\cref{subfig:path-typing-derivation} asserts this.
$\xscoperef{B}$ is the scope that encloses the reference.
Splitting the access path from $\xscoperef{C}$ to $\xscoperef{A}$ at that $\xscoperef{B}$ yields two segments of a single step.
The segment leading up to $\xscoperef{B}$ ($\xscoperef{C} \pathstept[\lblEXT] \xscoperef{B}$) does indeed match the regular expression $\reclos{\lblLEX}\reclos{\lblEXT}$.
Likewise, the other segment also matches its regular expressions, showing that this access is valid.
Note that, when class~\id{C} would have extended class~\id{B} with \texttt{protected} or \texttt{private} visibility instead,
the premise on the first section would not hold anymore.
This corresponds with the behavior in~\cref{sec:subclass} (the field must be accessible as if it was defined on the instance type) as well as the specification of \CPP cited above.

\section{Analysis}
\label{sec:analysis}

A comprehensive model of accessibility can be made by composing the system fragments we discussed so far (\cref{fig:aml-base,fig:internal,fig:private,fig:protected,fig:interaction,fig:path-access}).
In this section, we discuss a few properties that our system adheres to.

\subsection{Soundness of Access Policies}
\label{subsec:soundness-theorems}

First, we claim some soundness theorems for \texttt{private}, \texttt{protected} and \texttt{internal} access.
There is no soundness theorem for \texttt{public}, as access is allowed unconditionally.
Soundness theorems for \texttt{private protected} and \texttt{protected internal} are easily derived from~\cref{def:prot-sound,def:int-sound}, and hence omitted.
In the theorems, $P_\SG$ ranges over valid typing derivation for an AML program with scope graph $\SG$,
$x_r$ over references, and $x_d$ over declarations.
\AppendixRef{sec:characterizing-predicates}{D} defines the predicates used in these theorems, and proves them.

\newpage
First, soundness for \texttt{private} access is stated as follows:
\begin{theorem}[Soundness of \texttt{private} member access]
  \label{def:priv-sound}
  \begin{align*}
    &\mathsf{resolve}_{P_\SG}(x_r) = x_d 
      \mathrel{\land} 
      \mathsf{private}_{P_\SG}(x_d)\ 
    \Rightarrow{}\\
    &\hspace{1em} 
    \exists s_d.\ 
      \mathsf{definingClass}_{P_\SG}(x_d) = s_d 
      \mathrel{\land} 
      \mathsf{enclosingClass}_{P_\SG}(x_r, s_d)
  \end{align*}
\end{theorem}
This should be read as `when $x_r$ resolves to $x_d$, and $x_d$ is private, then $x_r$ must occur in the class $s_c$ that defines $x_d$'.

Likewise, soundness for \texttt{protected} access is stated as:
\begin{theorem}[Soundness of \texttt{protected} member access]
  \label{def:prot-sound}
  \begin{align*}
    &\mathsf{resolve}_{P_\SG}(x_r) = x_d 
      \mathrel{\land} 
      \mathsf{protected}_{P_\SG}(x_d)\ 
    \Rightarrow {}\\
    &\hspace{1em}
    \exists s_c, s_d.\ 
      \mathsf{definingClass}_{P_\SG}(x_d) = s_d 
      \mathrel{\land}
      \mathsf{enclosingClass}_{P_\SG}(x_r, s_c)
      \mathrel{\land}
      \mathsf{subClass}_{P_\SG}(s_c, s_d)
  \end{align*}
\end{theorem}
Compared to~\cref{def:priv-sound}, this theorem states that $x_r$ can occur in some arbitrary subclass~$s_c$ of $s_d$ if $x_d$ is \texttt{protected}.

Finally, \texttt{internal} access is specified correctly when:
\begin{theorem}[Soundness of \texttt{internal} member access]
  \label{def:int-sound}
  \begin{align*}
    &\mathsf{resolve}_{P_\SG}(x_r) = x_d 
      \mathrel{\land} 
      \mathsf{internal}_{P_\SG}(x_d, \setVar{x})\
    \Rightarrow {}\\
    &\hspace{1em}
    (\exists x,\, s_m.\
      \mathsf{enclosingMod}_{P_\SG}(x_r) = s_m
      \mathrel{\land}
      x \in \setVar{x}
      \mathrel{\land} 
      \mathsf{resolveMod}(x) = s_m) \lor {}\\
    &\hspace{1em}
    (\exists s_d.\ 
      \mathsf{definingClass}_{P_\SG}(x_d) = s_d 
      \mathrel{\land} 
      \mathsf{enclosingClass}_{P_\SG}(x_r, s_d))
  \end{align*}
\end{theorem}
This theorem states that references to declarations with modifier \texttt{internal} are valid if the enclosing module of the reference $s_m$ is referred to in the arguments of the access modifier $\setVar{x}$, or if it is accessed as a private variable.

\subsection{Equivalence of Access Policies}
\label{subsec:equivalences}

The access policy language \synt{A} we defined is not minimal.
It is possible to define equivalent policies in multiple ways.
To analyze that, we define equivalence of access policies as follows:

\begin{restatable}[Equivalence of Access Policies]{definition}{accesspolicyequiv}
  \label{def:acc-equiv}
  \begin{inferruleraw}
    \inference{
      \forall \SG, s, p.\ (\premMAcc{\SG}{s}{p}{A}) \Leftrightarrow (\premMAcc{\SG}{s}{p}{A'})
    }{A \equiv A'}
  \end{inferruleraw}
\end{restatable}
\noindent
That is, two accessibility policies are equivalent when, for any scope $s$, path $p$, scope graph~$\SG$, either both policies admit access, or neither does.

The equivalences that hold in our model are: $\PRT \equiv \SMDof{\emptyset}$ and $\PRV \equiv \SMCof{\emptyset} \equiv \MODof{\emptyset}$.
This follows from the fact that module access grants nothing if no module parameters are given.
Thus, the $\SMDof{\emptyset}$ policy reduces to $\PRT$, while $\SMCof{\emptyset}$ and $\MODof{\emptyset}$ do not elevate accessibility beyond $\PRV$.
\AppendixRef{sec:equivalence-proofs}{B} gives proofs for each of these equivalences.
Because of these equivalences, we did not include \PRT and \PRV in our implementation (\cref{sec:evaluation}).

\subsection{Order of Access Policies}
\label{subsec:policy-order}

Intuitively, there exists an ordering between accessibility policies, where \smash{$\PRV$} is the bottom most restrictive, and \smash{$\PUB$} is the least restrictive.
This order is partial, as the module-exposure dimension and subclass-exposure dimension are orthogonal.
Assuming a subset relation on scope sets \smash{($S \subset S'$)}, we can define a strict partial order \smash{$A <_A A'$} as follows:

\newpage
\noindent
\begin{minipage}{\textwidth}
  \centering
  \begin{tikzpicture}[
    plc/.style = {draw=black, rounded corners},
    cond/.style = {yshift = 0.65em, xshift=0.1em},
    node distance = 1.5em and 1em
  ]
    \node[plc] (prv) {$\PRV$};

    \node[plc, right = of prv] (smcl) {$\SMCof{S_0}$};
    \draw (smcl) edge[->] (prv);

    \node[plc, right = 3.85em of smcl] (smcr) {$\SMCof{S_1}$};
    \draw (smcr) edge[->] node[cond] {$S_0 \subsetneq S_1$} (smcl);

    \node[plc, above right = of smcr, xshift = -2em] (modl) {$\MODof{S}_1$};
    \draw (modl) edge[->] (smcr);

    \node[plc, right = 3.85em of modl] (modr) {$\MODof{S}_2$};
    \draw (modr) edge[->] node[cond] (modlr-c) {$S_1 \subsetneq S_2$} (modl);

    \node[plc] (prt) at (smcr -| modlr-c) {$\PRT$};
    \draw (prt) edge[->] (smcr);

    \node[plc, below right = of modr, xshift = -2em] (smdl) {$\SMDof{S_2}$};
    \draw (smdl) edge[->] (modr);
    \draw (smdl) edge[->] (prt);

    \node[plc, right = 3.85em of smdl] (smdr) {$\SMDof{S_3}$};
    \draw (smdr) edge[->] node[cond] {$S_2 \subsetneq S_3$} (smdl);


    \node[plc, right = of smdr] (pub) {$\PUB$};
    \draw (pub) edge[->] (smdr);
  \end{tikzpicture}

  \vspace{1em}
\end{minipage}
where the edges indicate instances of the $<_A$-relation.
The edges with a condition indicate that $\SMC$, $\MOD$, and $\SMD$ become more permissive when more scopes are added to the policy.

The intuition behind this order is not arbitrary.
In fact, we claim the following:
\begin{restatable}[The order on access policies $<_A$ is well-behaved]{theorem}{accpolorderwellbehaved}
  \label{thm:acc-pol-order-well-behaved}
  \begin{inferruleraw}
    (A <_A A') 
    \Rightarrow
    \forall \SG, s, p.\ (\premMAcc{\SG}{s}{p}{A}) \Rightarrow (\premMAcc{\SG}{s}{p}{A'})
  \end{inferruleraw}%
\end{restatable}%
\noindent
That is, when $A$ is more restrictive than $A'$, and $A$ permits access in scope $s$ via a path $p$, then~$A'$ will permit that access too.
A proof of this theorem can be found in \AppendixRef{sec:preorder-proofs}{C}.

\section{Evaluation}
\label{sec:evaluation}

So far, we have motivated our specification with examples from real-world languages such as Java and \CSharp, and stated some generic properties of our model.
However, for our specification to be usable as a basis for practical tools, it must correspond with the behavior of the actual languages.
To validate that, we evaluated our specification in two ways.
First, we systematically compared our specification with reference compilers of Java, \CSharp, and Rust.
Second, we validated the compatibility of our framework with recent work on language-parametric code completion~\cite{PelsmaekerAPV22}.

\subsection{Comparison with Reference Compilers: Implementation}
\label{subsec:eval-impl}

The comparison to compilers of real-world languages is implemented as follows:
\begin{enumerate}
  \item \label{it:analysis} Apply our type system on an AML program (the test case).
  \item Translate the AML program to the target language.
  \item \label{it:compile-ref} Compile the translated program using a compiler of the target language.
  \item Compare results: either both analyses should succeed, or both should give errors.
\end{enumerate}
We discuss these steps in more detail below.

\subparagraph*{AML Type Checker.}
To compare our model with real-world compilers, we need a way to type check concrete AML programs.
To that end, we implemented AML in the Spoofax language workbench~\cite{KatsV10,WachsmuthKV14}.
The actual type system is implemented using the Statix specification language~\cite{AntwerpenPRV18,RouvoetAPKV20}.
Statix is a suitable choice, as its declarativity allows an overall straightforward translation from our inference.
For example, the Statix encoding of rule~\rref{t-var} in~\cref{fig:t-var-stx} strongly corresponds to the original (\cref{fig:aml-base}).
Using this implementation, we can systematically check accessibility in concrete AML programs.

\subparagraph*{Compiling with Reference Compiler}

Next, we implemented source-to-source translations from AML to each of Java, \CSharp and Rust.
This translation was straightforward by design, as otherwise the results of the type checkers can be different due to semantic differences introduced by the translation.
For that reason, we do not support AML features that have no direct counterpart in the target language.
For example, the translation to Java will error when the AML program uses the \texttt{private protected} access modifier, as Java does not support that accessibility policy.
This way, we know that correspondence between the programs is guaranteed when the translation succeeds.

\begin{figure}
  \begin{subfigure}[t]{0.6\textwidth}
\begin{lstlisting}[language=Statix]
typeOfExpr: scope * Expr -> TYPE

typeOfExpr(s, Id(x)) = T :- {p A}
  query var
    filter LEX* EXT*
    and { x' :- x' == x }
    min $ < EXT, EXT < LEX
    in s |-> [(p, (x, T, A))],
  accessOk(s, p, A),
  pathOk(s, p).
\end{lstlisting}
    \vspace{-0.5em}
    \caption{Encoding of rule~\rref{t-var} in Statix.}
    \label{fig:t-var-stx}
  \end{subfigure}
  \begin{subfigure}[t]{0.4\textwidth}
\begin{lstlisting}[language=SptAML]
test `\texttt{private}` - nested [[
  class A {
    private var x = 42
    class B {
      public var y = x
    }
  }
]]
analysis succeeds
run java-compat
\end{lstlisting}
    \vspace{-0.5em}
    \caption{Example test case.}
    \label{fig:testcase-spt}
  \end{subfigure}
  \caption{Overview of Approach to Comparison with Reference Compilers.}
  \label{fig:reference-compiler-eval}
\end{figure}

After translating, we invoke the reference compiler, observe its output (success or failure),
and compare the given output with the result from our own type checker (step~\ref{it:analysis}).
If those are different (\ie, our type checker accepts the program, while the reference compiler emits errors, or vice versa), the test fails.

\subsection{Comparison with Reference Compilers: Test Cases}

To the best of our knowledge, there exists no test suite specifically aimed at verifying the semantics of access modifiers.
For that reason, we manually created an extensive test suite.
Each test contains a class \texttt{Def}, that defines some variable \texttt{x} with some access modifier $A$.
Furthermore, each test contains a class \texttt{Ref}, in which a reference to \texttt{x} occurs.
\texttt{Def} and \texttt{Ref} can be related in two different ways at the same time:
\begin{itemize}
  \item By inheritance: either (1) \texttt{Def} and \texttt{Ref} are actually the same class, (2) have no mutual inheritance, (3) \texttt{Ref} inherits \texttt{Def}, or (4) \texttt{Def} inherits \texttt{Ref}.
  \item By module position: either \texttt{Def} and \texttt{Ref} (1) occur in the same module, or (2) \texttt{Ref} occurs in a parent/sibling/child module of \texttt{Def}.
  \item By class nesting: either (1) \texttt{Def} and \texttt{Ref} are top-level classes, (2) \texttt{Ref} is nested in \texttt{Def}, (3) \texttt{Def} is nested in \texttt{Ref}, or (4) \texttt{Def} and \texttt{Ref} have a shared enclosing class.
\end{itemize}
In addition, tests for member accesses (\ie, \texttt{recv.x}) have a receiver type \texttt{Recv}.
This type must either be equal to \texttt{Def}, or inherit from it.
However, it can be related in all possible other ways to \texttt{Def} and \texttt{Ref}.
By systematically exploring all options, we derived our test suite.

We excluded cases that are 
(1) impossible (\eg, \texttt{Ref} cannot be nested in \texttt{Def} and live in a different module at the same time), 
(2) use features not supported by the target language, 
(3) invalid for another reason (\ie, inheriting from a nested class is not allowed by Java), or 
(4) do not bind properly (\ie, lexical access where \texttt{Ref} and \texttt{Def} do not inherit from each other, and are not nested in each other), 
To reduce the number of test cases, we restricted the cases that involved nested classes to have one module only.
Additionally, we only used \lit*{private}, \lit*{protected} and \lit*{protected internal} as access modifiers in these cases.
\Cref{tbl:test-suite} summarizes the results of the test suite generation.

\newcommand{\javaAMa}{\lit*{public}}
\newcommand{\javaAMb}{\lit*{protected internal}}
\newcommand{\javaAMc}{\lit*{internal}, \lit*{private}}
\newcommand{\javaFta}{packages, class inheritance}
\newcommand{\javaFtb}{class nesting}

\newcommand{\csAMa}{\lit*{protected}}
\newcommand{\csAMb}{\lit*{private protected}}
\newcommand{\csFt}{
  classes (inheritance/nesting)
}

\begin{table}
  \caption{Summary of Test Suite.}
  \label{tbl:test-suite}
  \begin{tabular}{r | r r r r}
    \toprule
                            & Java                       & \CSharp                   & Rust                  & Manual\\
    \midrule   
      \textit{Acc.\@ Mods.} & \javaAMa{}                 & Same as Java, and         & \lit*{public}         & All \\ 
                            & \javaAMb{}                 & \csAMa{}                  & \lit*{internal}       &     \\ 
                            & \javaAMc{}                 & \csAMb{}                  &                       &     \\ 
      \textit{Features}     & class inheritance          & class inheritance         & structs,              & advanced modules \\
                            & class nesting, and         & class nesting, and        & modules               & inheritance visibility\\
                            & packages                   & assemblies                & \\
      \textit{\#Cases}      & 433                        & 522                       & 60                    & 168 \\
      \textit{Compiler}     & \texttt{javac} 11.0.20.1   & \texttt{dotnet} 7.0.401   & \texttt{rustc} 1.73.0 & --- \\ 
    \bottomrule
  \end{tabular}
\end{table}

\Cref{fig:testcase-spt} shows an example test case written in the Spoofax Testing Language (SPT)~\cite{KatsVV11}.
This test validates that a private field is accessible from a nested class.
The test consists of a program (between double square brackets), and some \emph{expectations}.
In this case, we expect the (Statix-based) analysis to succeed.
Moreover, we expect the \texttt{java-compat} transformation to succeed.
This transformation is executing the steps in~\cref{subsec:eval-impl}.

\subparagraph*{Results}
There are several features present in AML that were not covered by any of the reference compilers, most notably \lit*{private}/\lit*{protected} \emph{inheritance}, and module visibility beyond what Rust supports.
To validate we cover these features to some extent, we have written 168 additional test cases.
While initially exposing a lot of edge cases, in the end all test cases succeeded.
This shows that our specification covers the languages it set out to model rather accurately.

\subsection{Code Completion}

One of our future goals is to use our framework to implement refactoring tools that are sound with respect to accessibility.
The most recent work in this area is done by Pelsmaeker et al.~\cite{PelsmaekerAPV22}.
They show how Statix specifications can be used to generate editor auto-completion proposals language-parametrically.
We applied auto-completion to the access modifiers in the \CSharp/Java and Rust tests (152, after deduplication), and validated soundness and completeness.
That is, when the analysis succeeded, code completion should propose the current modifier at that position.
Otherwise, if the access was invalid, the modifier should not be proposed, as only less restrictive ones are valid at that position.

We consider the fact that all completion tests pass a good indication that our specification can be applied with refactoring tools in the future.
Apparently, the code completion framework is sound and complete with respect to our encoding of access modifiers.
Accessibility errors introduced by a refactoring can be fixed by generating proposals for that position,
and using the ordering from~\cref{subsec:policy-order} to pick the most restrictive one.

\subsection{Threats to Validity}

In~\cref{sec:setup}, we briefly mentioned that the specification as presented in the paper did not model the interaction between shadowing and accessibility correctly.
Doing so would require a full path order, instead of ordering paths by a lexicographical order on labels.
\AppendixRef{subsec:full-path-order}{A.1} explains how we think that could be done.
However, Statix does not support full path orders.
To work around that, we emulated this behavior using a few helper predicates.
Our test suite gives confidence we modeled it correctly, but we did not prove that the specification in Statix and the full path order are semantically equivalent.

Finally, we might have modeled incorrect/unspecified behavior if the reference compilers were incorrect.
Examples such as~\cref{fig:diff-priv-sub} were derived from actual compiler behavior.
However, we could not find our interpretation of the implementation behavior explicitly specified in the JLS~\cite[§6.6.1]{JLS8}.

\section{Related Work}

In this section, we discuss previous work related to access modifiers and scope graphs.

\subsection{Access Modifier Semantics and Implementations}

The origin of access modifiers dates back to at least Simula 67,
which around 1972 introduced \texttt{protected} and \texttt{hidden} access modifiers~\cite[\S8]{Black13-1}
(the latter being equivalent to our \texttt{private}).
Later, languages such as Java and \CPP incorporated these keywords, making them well-known and often used.
Design principles and patterns~\cite{GammaHJV93} using these keywords were developed,
making contemporary software development heavily reliant on accessibility features the programming language provides.

Giurca and Savulea~(2004)~\cite{GiurcaS04} apply object-oriented notions of \texttt{public}, \texttt{protected} and \texttt{private} to logic programs,
with the purpose of better knowledge distribution and run time optimization.
Moreover, Apel et al.~\cite{ApelLKKL09,ApelKLKKL12} introduce access modifiers in feature-oriented programming.
Where we define accessibility for module nesting and class inheritance, they add the `feature refinement' dimension to this.
In particular, the \texttt{feature} keyword restricts access to the `current feature' only (comparable to \texttt{private} in the class inheritance dimension),
the \texttt{subsequent} keyword grants access to the current feature and later refinements (comparable to \texttt{protected}), and
the \texttt{program} modifier allows access from any position (similar to \texttt{public}).
In our terminology, their model supports `conjunctive' combination of the class and feature dimensions.
As~\cref{sec:combining-subclass-module} shows that combining the module and class dimensions conjunctively is straightforward,
we expect that integrating their work in our model will not pose major challenges (apart from a combinatorial explosion of policies).

\subparagraph*{Semantics.}

As access modifiers mainly originated from practical needs, it is not very surprising that little attention to them was paid from a more theoretical perspective.
A few attempts to create a more formal account have been performed, however.
In 1998, Yang~\cite{Yang01} presented a formalization of Java access modifiers using attribute grammars.
At that time, attribute grammars still lacked several convenience features, such as default attributes~\cite{WykMBK02} and collection attributes~\cite{10.1109-SCAM.2007.13}.
For that reason, all members must be propagated explicitly to the scopes where they are accessible,
which makes the specification rather verbose.
Additionally, since fields and methods are not treated equally (shadowing vs.\@ overriding), they are treated separately, doubling the specification size.
In contrast, we specify the propagation of members queries in scope graphs, which is more concise.
The additional requirements on methods (not explicitly discussed), can be handled at the definition site.
Furthermore, we cover more features than just the Java ones.
Fharkani et al.~\cite{FarkhaniRT08} present a generalized model of accessibility,
where accessibility is modeled as a set of rules granting access of a member to another member (similar to Eiffel/\texttt{friends} in \CPP).
In addition, rules can \emph{deny} access to the named member, or apply to all members except the named ones.

\subparagraph*{Tools.}

Steimann et al.~\cite{SteimannT09} observe that disregarding accessibility can result in a lot of subtle mistakes.
For example, a method may silently fail to override another method when it is moved to a different package,
which results in different dynamic dispatch.
To capture these errors, they present nine constraint generation rules that model the accessibility semantics of Java.
Refactoring tools can use these constraints to detect where the accessibility level of a member must be elevated.
This work was incorporated in the JRRT refactoring tool~\cite{SchaferTST12}, which was evaluated on a large number of real-world Java projects, showing the accuracy of their implementation.
While their work also covers overriding-specific constraints, which our specification treats rather superficially,
we think our model is more comprehensible, and also gives insight in the differences between languages.
Moreover, their work is applied in real refactoring implementations, while the quest for Statix-based refactorings is still ongoing.
Meanwhile, a similar approach was applied to Eiffel accessibility~\cite[\S6.3]{SteimannKP11}.

While these tools \emph{elevate} accessibility if needed, a different line of research aims to \emph{restrict} accessibility if possible~\cite{BouillonGS08,Muller10,ZollerS12-0}.
The purpose of these tools is to detect access modifiers that are more lenient than needed, and restrict those.
This is claimed to improve readability, enable optimizations, and increase modularity~\cite{Muller10}.
The exact underlying model is not the topic of these publications, and hence remains unclear.
Despite that, the tools appear to be useful in practice.
Zoller and Schmolitzky mention some challenges in porting their tool to other object-oriented languages~\cite[V.B]{ZollerS12-0}.
A language-parametric model such as ours helps in that regard by
(1) making differences between languages explicit, and
(2) make implementations of these (kind of) tools language-parametric.

\subsection{Scope Graphs}

Scope Graphs (\cref{sec:scope-graphs}) have been introduced by Neron et al.~\cite{NeronTVW15}, and later refined by Van Antwerpen et al.~\cite{AntwerpenPRV18} and Rouvoet et al.~\cite{RouvoetAPKV20}.
In order to bridge the gap between language specification and implementation, scope graphs have been embedded into the NaBL2 constraint language~\cite{AntwerpenNTVW16}.
Later, the Statix logic language was introduced~\cite{AntwerpenPRV18,RouvoetAPKV20}, which is more expressive than NaBL2.
Both languages allow specifying type checking as constraint programs, giving the language a declarative appeal, but also yielding an executable type checker.
Scope graphs are also available in a framework for concurrent and incremental type checkers~\cite{AntwerpenV21,ZwaanAV22} and an embedded DSL in Haskell~\cite{PoulsenZH23}.
Finally, Statix specifications have been used for language-parametric code completion~\cite{PelsmaekerAPV22} and refactorings~\cite{Misteli21,Gugten22}.
Zwaan and Van Antwerpen provide a detailed overview of the development history of scope graphs, their embeddings in type system specification DSLs, and their applications~\cite{ZwaanA23}.

\section{Conclusion}

Access modifiers occur in many real-world languages.
To implement high-quality tooling for these languages, a good understanding of access modifiers is required.
In this paper, we presented a model for access validation based on scope graphs.
Our model covers the most important accessibility features in contemporary languages, including module accessibility, and inheritance accessibility, both on declarations and extends-clauses.
Variations between different languages, both in supported features and their semantics, are made explicit in our model.
Our specification is quite declarative, partly because scope graphs abstract over low-level name resolution and scoping details.
Our model was validated using an extensive test suite, using Java, \CSharp, and Rust compilers as oracles.
This test suite was also used to show that we can synthesize access modifiers accurately using previous work on code completion~\cite{PelsmaekerAPV22}.

Our main motivation for this work is twofold.
First, we aim to provide a `language-transcendent' model for accessibility that enables comparison of different languages regarding accessibility.
To this end we identify and formalize differences in the semantics of several access modifiers.
In addition, we formulate soundness theorems of several access modifiers, and prove them.
As such, we consider our specification accurate enough to serve as a reference for future tool implementations.
Second, we aim to use our model in language-parametric refactorings, ensuring they respect accessibility properly.
As these refactoring tools are still in development, actual validation of this application is still future work.

\bibliography{local,bibliography}
\bibliographystyle{plainurl}

\iftoggle{extended}{
  \appendix
  \AddToHook{cmd/section/before}{\clearpage}

\newcommand{\ctx}{\mathcal{C}}
\newcommand\inCtx[2][]{\ctx#1{\left[#2\right]}}


\theoremstyle{claimstyle}
\newtheorem*{proofsketchbase}{Proof Sketch}

\newenvironment{proofsketch}{%
  \begin{proofsketchbase}%
}{%
  \hfill\ensuremath{\vartriangleleft}
  \end{proofsketchbase}%
}

\theoremstyle{claimstyle}
\newtheorem*{proofcase}{Case}

\newcommand{\casePar}[2]{
  \proofsubparagraph{Case {\color{black} $\mathrm{\protect#1 <_A \protect#2}$}.}
}
\newcommand{\ordCaseProof}[4][]{
  \ordCaseProofNoPar[#1]{#2}{#3}{#4}
}
\newcommand{\ordCaseProofNoPar}[4][]{
  \begin{proofcase}
    $\ifstrempty{#1}{}{
      #1
      \Rightarrow
    }
    \forall \SG, s, p.\ (\premMAcc{\SG}{s}{p}{#2}) \Rightarrow (\premMAcc{\SG}{s}{p}{#3})$
  \end{proofcase}
  \begin{claimproof}
    #4
  \end{claimproof}
}

  \section{AML Specification}
\label{sec:aml-full-spec}

In this section, we present a specification of AML~\cref{sec:setup}, with the following restrictions:
\begin{itemize}
  \item We do not cover the `subclass'-inheritance (\cref{sec:inheritance-restriction}). That is, only public inheritance (Java/\CSharp-like) is supported. The rules that validate paths are included for reference.
  \item For the variation points regarding \lit*{protected}/\lit*{private} access modifiers (\cref{sec:subclass}), we chose the variant that corresponds to \CSharp.
  \item For the module-accessibility, we chose the variant that allows the modifier to refer to any module (not only enclosing ones), but does not expose members to sub-modules of the named modules.
\end{itemize}
For this specification, we prove access modifier equivalences (\cref{sec:equivalence-proofs}), strict partial order (\cref{sec:preorder-proofs}), and characterizing predicates (\cref{sec:characterizing-predicates}).
\Cref{fig:aml-full-syntax} gives the full syntax of AML.
\Cref{fig:aml-full-modules,fig:aml-full-classes,fig:aml-full-expr,fig:aml-full-access-modifiers,fig:aml-full-access-policy} give the typing rules for AML, grouped by topic.

\begin{figure}[!t]
  \setlength{\grammarindent}{5em}
    \begin{minipage}{0.63\textwidth}
      \begin{grammar}
        <prog> ::= <mod>$^{*}$

        <mod>  ::= "module" <x> "\{" <md>$^{*}$ "\}"

        <md>   ::= <mod> | "import" <x> | <cls>

        <cls>  ::= "class" <x> (":" "public" <x>)$^{?}$ "\{" <cd>$^{*}$ "\}"

        <cd>   ::= <acc> "var" <x> "=" <e> | <cls>

        <acc>  ::= "public" | "internal" "(" <x>$^{*}$ ")"
              \alt "protected" | "protected internal" "(" <x>$^{*}$ ")"
              \alt "private" | "private protected" "(" <x>$^{*}$ ")" 

        <e>    ::= <n> | <x> | "new" <x> "()" | <e> "." <x>
      \end{grammar}
    \end{minipage}%
    \begin{minipage}{0.36\textwidth}
      \begin{grammar}
        <l>    ::= \lblLEX | \lblIMP | \lblEXT
              \alt \lblMOD | \lblCLS | \lblVAR
              \alt \lblTHISMOD | \lblTHIS

        <d>    ::= "mod" <x> ":" <s>
              \alt "cls" <x> ":" <s>
              \alt "var" <x> ":" <T> "@" <A>
              \alt <s>

        <T>    ::= "int" | "inst" <s>

        <A>    ::= "PUB" | "MOD" $S$ 
              \alt "PRT" | "SMD" $S$
              \alt "PRV" | "SMC" $S$ 
      \end{grammar}%
    \end{minipage}

    \caption{Syntax of AML.}
  \label{fig:aml-full-syntax}
\end{figure}

\begin{figure}[!t]
  \def\widthl{0.46\textwidth}
  \def\widthr{0.44\textwidth}
  \begin{minipage}[t]{\widthl}
    \figuresection[\fbox{$\premProgOk{\SG}{\mathit{mod}^\ast}$}]{Program}
  \end{minipage}%
  \hfill
  \begin{minipage}[t]{\widthr}
    \figuresection[\fbox{$\premModOk{\SG}{s}{\mathit{mod}}$}]{Module}
  \end{minipage}

  \begin{minipage}[b]{\widthl}
    \begin{inferrule}
      \inference[\rlabel{prog-ok}{Prog-Ok}]{
        s_0 \in \SG
        &
        s_0 \scopeedget[\lblTHISMOD] s_0
        \\
        \left[
          \premModOk{\SG}{s_0}{m}
        \right]_{m \in \setVar{\mathit{mod}}}
      }{
        \premProgOk{\SG}{\setVar{\mathit{mod}}}
      }%
    \end{inferrule}
  \end{minipage}%
  \hfill
  \begin{minipage}[b]{\widthr}
    \begin{inferrule}
      \inference[\rlabel{mod-ok}{Mod-Ok}]{
        s_m \in \SG
        &
        s_m \scopeedget[\lblTHISMOD] s_m
        \\
        \left(s \typeedget[\lblMOD] \moddecl{x}{s_m}\right) \in \SG
        \\
        \left[
          \premModDeclOk{\SG}{s_m}{m}
        \right]_{m \in \setVar{\mathit{md}}}
      }{
        \premModOk{\SG}{s}{\mathbf{module} \mathbin{x} \mathbf{\{} \mathbin{\setVar{\mathit{md}}} \mathbf{\}}}
      }%
    \end{inferrule}
  \end{minipage}
  \vspace{1em}
  \figuresection[\fbox{$\premModDeclOk{\SG}{s}{\mathit{md}}$}]{Module Members}
  \begin{inferrule}
    \inference[\rlabel{imp-ok}{Imp-Ok}]{
      \premQMod{\SG}{s}{x}{s_m}
      &
      \left(s \scopeedget[\lblIMP] s_m\right) \in \SG
    }{
      \premModDeclOk{\SG}{s}{\mathbf{import} \mathbin{x}}
    }
    \qquad
    \inference{
      \premModOk{\SG}{s}{\mathit{mod}}
    }{
      \premModDeclOk{\SG}{s}{\mathit{mod}}
    }
    \qquad
    \inference{
      \premClassOk{\SG}{s}{\mathit{cls}}
    }{
      \premModDeclOk{\SG}{s}{\mathit{cls}}
    }
  \end{inferrule}

  \figuresection[\fbox{$\premQMod{\SG}{s}{x}{s}$}]{Module Resolution}
  \begin{inferrule}
    \inference[\rlabel{q-cls}{Q-Cls}]{
      %
      %
      \query{\SG}{s}{\reclos{\lblLEX}\lblMOD}{\mathsf{isMod}_x}{$\lblMOD < \lblLEX$}{\setOf{ \langle p, \moddecl{x}{s_m} \rangle }}
    }{
      \premQMod{\SG}{s}{x}{s_m}
    }
  \end{inferrule}

  \caption{Typing AML Modules.}
  \label{fig:aml-full-modules}
\end{figure}

\begin{figure}[!t]
  \figuresection[\fbox{$\premClassOk{\SG}{s}{\mathit{cls}}$}]{Classes}
  \begin{inferrule}
    \inference[\rlabel{cls-ok}{Cls-Ok}]{
      s_c \in \SG
      &
      s_c \scopeedget[\lblTHIS] s_c
      &
      \left(s \typeedget[\lblCLS] \clsdecl{x}{s_c}\right) \in \SG
      &
      \premExtOk{\SG}{s}{\mathit{ext}^{?}}
      &
      \left[
        \premClassDeclOk{\SG}{s_c}{\mathit{cd}}
      \right]_{\mathit{cd} \in \setVar{\mathit{cd}}}
    }{
      \premClassOk{\SG}{s}{\mathbf{class} \mathbin{x} {} \! {} \mathbin{\mathit{ext}^{?}} \mathbf{\{} \mathbin{\setVar{\mathit{cd}}}  \mathbf{\}}}
    }
  \end{inferrule}
  \hspace{0.35em}
  \figuresection[\fbox{$\premExtOk{\SG}{s}{\mathit{ext}^{?}}$}]{Extension}
  \begin{inferrule}
    \inference[\rlabel{no-ext}{No-Ext}]{
    }{
      \premExtOk{\SG}{s_c}{}
    }
    \qquad
    \inference[\rlabel{ext}{Ext}]{
      \premQCls{\SG}{s_c}{x}{s_p}
      &
      \left(
        s_c \scopeedget[\lblEXT] s_p
      \right) \in \SG
    }{
      \premExtOk{\SG}{s_c}{\mathbf{:} \mathbin{\mathbf{public}} x}
    }
  \end{inferrule}

  \hspace{0.35em}
  \figuresection[\fbox{$\premClassDeclOk{\SG}{s}{\mathit{cd}}$}]{Class Members}
  \begin{inferrule}
    \inference[\rlabel{d-def}{D-Def}]{
      \premExp{\SG}{s}{e}{T}
      \quad
      \premAcc{\SG}{s}{\mathit{acc}}{A}
      \quad
      s \typeedget[\lblVAR] (\defdecl{x}{T}{A}) \in \SG
    }{
      \premClassDeclOk{\SG}{s}{\mathit{acc} \mathrel{\lit*{var}} x \mathrel{\lit*{=}} e}
    }%
    \qquad
    \inference{
      \premClassOk{\SG}{s}{\mathit{cls}}
    }{
      \premClassDeclOk{\SG}{s}{\mathit{cls}}
    }%
  \end{inferrule}

  \figuresection[\fbox{$\premQCls{\SG}{s}{x}{s}$}]{Class Resolution}
  \begin{inferrule}
    \inference[\rlabel{q-cls}{Q-Cls}]{
      %
      %
      \query{\SG}{s}{\reclos{\lblLEX}\reopt{\lblIMP}\lblCLS}{\mathsf{isCls}_x}{$\lblCLS < \lblIMP < \lblLEX$}{\setOf{ \langle p, \clsdecl{x}{s_c} \rangle }}
    }{
      \premQCls{\SG}{s}{x}{s_c}
    }
  \end{inferrule}

  \caption{Typing AML Classes.}
  \label{fig:aml-full-classes}
\end{figure}

\begin{figure}[!t]
  \figuresection[\fbox{$\premExp{\SG}{s}{e}{T}$}]{Expresions}
  \begin{inferrule}
    \inference[\rlabel{t-int}{T-Int}]{
      {}
    }{
      \premExp{\SG}{s}{n}{\tyINT}
    }
    \qquad
    \inference{
      \premQCls{\SG}{s}{x}{s_c}
    }{
      \premExp{\SG}{s}{\lit*{new}\ x\ \lit*{()}}{\tyINST{s_c}}
    }
  \end{inferrule}
  \begin{inferrule}
    \inference[\rlabel{t-var}{T-Var}]{
      %
      %
      \query{\SG}{s}{\reLEX}{\mathsf{isVar}_x}{$<_p$}{\setOf{ \langle p, \defdecl{x}{T}{A} \rangle }}
      \\
      \premMAcc{\SG}{s}{p}{A}
      &
      \premMPth{\SG}{s}{p}
    }{
      \premExp{\SG}{s}{x}{T}
    }
  \end{inferrule}
  \begin{inferrule}
    \inference[\rlabel{t-fld}{T-Fld}]{
      %
      %
      \premExp{\SG}{s}{e}{\tyINST{s_c}}
      \\
      \query{\SG}{s_c}{\reMEM}{\mathsf{isVar}_x}{\lblVAR < \lblEXT}{\setOf{ \langle p, \defdecl{x}{T}{A} \rangle }}
      \\
      \premMAcc{\SG}{s}{p}{A}
      &
      \premMPth{\SG}{s}{p}
    }{
      \premExp{\SG}{s}{e.x}{T}
    }
  \end{inferrule}

  \caption{Typing AML Expressions. In rule~\rref{t-var}, $<_p$ represents the path order that takes accessibility in account properly. It is further discussed in~\cref{subsec:full-path-order}.}
  \label{fig:aml-full-expr}
\end{figure}

\begin{figure}[!t]
  \figuresection[\fbox{$\premAcc{\SG}{s}{\mathit{acc}}{A}$}]{Access Modifiers}

  \begin{inferrule}
    \inference[\rlabel{a-pub}{A-Pub}]{
    }{
      \premAcc{\SG}{s}{\lit*{public}}{\PUB}
    }
    \quad
    \inference[\rlabel{a-int}{A-Int}]{
      S = \left\{
        s'  \mathbin{\Big|}
        x_i \in \setVar{x}_{0 \ldots n}, 
        s \vdash_{\SG} x_i \resolvemod s'
      \right\}
    }{
      \premAcc{\SG}{s}{\lit*{internal(} \setVar{x}_{0 \ldots n} \lit*{)}}{\MODof{S}}
    }
  \end{inferrule}

  \begin{inferrule}
    \inference[\rlabel{a-prot}{A-Prot}]{
    }{
      \premAcc{\SG}{s}{\lit*{protected}}{\PRT}
    }
    \quad
    \inference[\rlabel{a-pint}{A-PInt}]{
      %
      S = \left\{
        s'  \mathbin{\Big|}
        x_i \in \setVar{x}_{0 \ldots n}, 
        s \vdash_{\SG} x_i \resolvemod s'
      \right\}
    }{
      \premAcc{\SG}{s}{\lit*{protected internal(} \setVar{x}_{0 \ldots n} \lit*{)}}{\SMDof{S}}
    }
  \end{inferrule}

  \begin{inferrule}
    \inference[\rlabel{a-priv}{A-Priv}]{
    }{
      \premAcc{\SG}{s}{\lit*{private}}{\PRV}
    }
    \quad
    \inference[\rlabel{a-pprot}{A-PProt}]{
      %
      S = \left\{
        s'  \mathbin{\Big|}
        x_i \in \setVar{x}_{0 \ldots n}, 
        s \vdash_{\SG} x_i \resolvemod s'
      \right\}
    }{
      \premAcc{\SG}{s}{\lit*{private protected(} \setVar{x}_{0 \ldots n} \lit*{)}}{\SMCof{S}}
    }
  \end{inferrule}

  \caption{Translating Access Modifier Keywords to Access Policies.}
  \label{fig:aml-full-access-modifiers}
\end{figure}

\begin{figure}[!t]
  \figuresection[\fbox{$\premMAcc{\SG}{s}{p}{A}$}]{Access Policy}

  \begin{inferrule}
    \inference[\rlabel{ap-pub}{AP-Pub}]{
    }{
      \premMAcc{\SG}{s}{p}{\PUB}
    }
    \quad
    \inference[\rlabel{ap-prot}{AP-Prot}]{
      \premEncC{\SG}{s}{S_C}
      &
      s_c \in S_C
      &
      s_c \in \mathsf{scopes}(p)
    }{
      \premMAcc{\SG}{s}{p}{\PRT}
    }
  \end{inferrule}

  \begin{inferrule}
    \inference[\rlabel{ap-int}{AP-Int}]{
      \premEncMI{\SG}{s}{s_m}
      &
      s_m \in S
    }{
      \premMAcc{\SG}{s}{p}{\MODof{S}}
    }
    \quad
    \inference[\rlabel{ap-priv}{AP-Priv}]{
      \premEncC{\SG}{s}{S_C}
      &
      \mathsf{tgt}(p) \in S_C
    }{
      \premMAcc{\SG}{s}{p}{A}
    }
  \end{inferrule}

  \begin{inferrule}
    \inference[\rlabel{ap-smd-prot}{AP-SMD-Prot}]{
      \premMAcc{\SG}{s}{p}{\PRT}
    }{
      \premMAcc{\SG}{s}{p}{\SMDof{S}}
    }
    \qquad
    \inference[\rlabel{ap-smd-mod}{AP-SMD-Mod}]{
      \premMAcc{\SG}{s}{p}{\MODof{S}}
    }{
      \premMAcc{\SG}{s}{p}{\SMDof{S}}
    }
  \end{inferrule}

  \begin{inferrule}
    \inference[\rlabel{ap-smc}{AP-SMC}]{
      \premMAcc{\SG}{s}{p}{\MODof{S}}
      &
      \premMAcc{\SG}{s}{p}{\PRT}
    }{
      \premMAcc{\SG}{s}{p}{\SMCof{S}}
    }
  \end{inferrule}

  \figuresection[\fbox{$\premMPth{\SG}{s}{p}$}]{Path Access}
  \begin{inferrule}
    \inference[\rlabel{p-pub}{P-Pub}]{
      \premMatchRE{p}{\reclos{\lblLEX}\reclos{\lblEXT}}
    }{
      \premMPth{\SG}{s}{p}
    }
    \qquad
    \inference[\rlabel{p-priv-prot}{P-Priv-Prot}]{
      \premEncC{\SG}{s}{S_C}
      &
      s_c \in S_C
      &
      \splitAt{s_c}{p} = \langle p_1, p_2 \rangle
      \\
      \premMatchRE{p_1}{\reclos{\lblLEX}\reclos{\lblEXT}}
      &
      \premMatchRE{p_2}{\reopt{\lblEXTPRIV}\reclos{(\lblEXT|\lblEXTPROT)}}
    }{
      \premMPth{\SG}{s}{p}
    }
  \end{inferrule}

  \figuresection[\fbox{$\premEncM{\SG}{s}{S}$}]{Enclosing Modules}
  \begin{inferrule}
    \inference[\rlabel{enc-mi}{Enc-MI}]{
      \query{\SG}{s}{\reclos{\lblLEX}\lblTHISMOD}{\top}{\lblTHISMOD < \lblLEX}{\setOf{\langle p, s_m \rangle}}
    }{
      \premEncMI{\SG}{s}{s_m}
    }
  \end{inferrule}

  \figuresection[\fbox{$\premEncC{\SG}{s}{S}$}]{Enclosing Classes}
  \begin{inferrule}
    \inference[\rlabel{enc-c}{Enc-C}]{
      \query{\SG}{s}{\reclos{\lblLEX}\lblTHIS}{\top}{}{R}
      \quad
      S_C = \setOf{ s_c \alt \langle p_c, s_c \rangle \in R }
    }{
      \premEncC{\SG}{s}{S_C}
    }
  \end{inferrule}

  \caption{Checking Access Permissions.}
  \label{fig:aml-full-access-policy}
\end{figure}

\clearpage
\subsection{Full Path Order}
\label{subsec:full-path-order}

In this section, we describe how to model the interaction between shadowing and accessibility, such as implemented in Java, correctly.
In particular, we should model the following behavior:
\begin{itemize}
  \item Members declared in subclasses shadow fields from superclasses/enclosing scopes, regardless of their accessibility.
  \item Members declared in enclosing classes are shadowed by \emph{accessible} members from superclasses.
  \item Members declared in enclosing classes shadow \emph{inaccessible} members from superclasses.
\end{itemize}
We model this using the order shown in~\cref{fig:full-path-order}.
Full path orders are \emph{preorders} over path-data pairs ($\langle p, d \rangle$).
One can consider query resolution using full path orders as computing the \emph{minimum} over the query results set without an user.
For details, we refer to~\cite[§3.1]{RouvoetAPKV20}.

\begin{figure}
  \figuresection[\fbox{$\premPathOrder{\SG}{s}{\langle p, d \rangle}{\langle p, d \rangle}$}]{Full Path Order}
  \begin{inferrule}
    \inference[\rlabel{ord-local}{PO-Loc}]{
      {}
    }{
      \premPathOrder{\SG}{s_r}{\langle s, d \rangle}{\langle s \termedget[l] p', d' \rangle}
    }
    \qquad
    \inference[\rlabel{ord-prefix}{PO-Pref}]{
      \premPathOrder{\SG}{s_r}{\langle p, d \rangle}{\langle p', d' \rangle}
    }{
      \premPathOrder{\SG}{s_r}{\langle s \termedget[l] p, d \rangle}{\langle s \termedget[l] p', d' \rangle}
    }
  \end{inferrule}
  \begin{inferrule}
    \inference[\rlabel{ord-ext}{PO-Acc}]{
      d = \defdecl{x}{T}{A}
      &
      \hll{\premMAcc{\SG}{s_r}{s \termedget[l] p}{A}}
      &
      \hll{\premMPth{\SG}{s_r}{s \termedget[l] p}}
      &
      l \in \setOf{\lblEXT, \lblEXTPROT, \lblEXTPRIV}
    }{
      \premPathOrder{\SG}{s_r}{\langle s \termedget[l] p, d \rangle}{\langle s \termedget[\lblLEX] p', d' \rangle }
    }
  \end{inferrule}
  \begin{inferrule}
    \inference[\rlabel{ord-lex-decl}{PO-NoAcc}]{
      d' = \defdecl{x}{T'}{A'}
      &
      \hll{s_r \nvdash_{\SG} s \termedget[l] p' \mathbin{!} A'}
      &
      l \in \setOf{\lblEXT, \lblEXTPROT, \lblEXTPRIV}
    }{
      \premPathOrder{\SG}{s_r}{\langle s \termedget[\lblLEX] p, d\rangle}{\langle s \termedget[l] p', d' \rangle}
    }
  \end{inferrule}
  \begin{inferrule}
    \inference[\rlabel{ord-lex-path}{PO-NoPath}]{
      d' = \defdecl{x}{T'}{A'}
      &
      \hll{s_r \nvdash_{\SG} s \termedget[l] p' \mathbin{\upexclaim} {}}
      &
      l \in \setOf{\lblEXT, \lblEXTPROT, \lblEXTPRIV}
    }{
      \premPathOrder{\SG}{s_r}{\langle s \termedget[\lblLEX] p, d\rangle}{\langle s \termedget[l] p', d' \rangle}
    }
  \end{inferrule}

  \caption{Path order $<_p$ capturing the interaction between shadowing and accessibility in Java.}
  \label{fig:full-path-order}
\end{figure}

The definition of this order is parameterized by a scope~$s$, which is the scope in which the reference occurs.
Rule~\rref{ord-local} states that a path that ends in the current scope~$s$ has priority over a path that traverses some edges after~$s$.
This corresponds to the first criterion: members in sub-classes ($s$) shadow members from superclasses or enclosing scopes.
Next, rule~\rref{ord-prefix} indicates that paths from the same scope that traverse an edge with the same label have an ordering \emph{iff} their suffixes have an ordering.
In this way, paths with a shared prefix can be ordered by applying~\rref{ord-prefix} until the paths diverge.
Both of these rules are a standard, and were implied in the label order we used previously~\cite[Fig.\@ 1]{AntwerpenPRV18}.

The limitation of label orders however is that only either $\lblLEX < \lblEXT$ can be chosen, or the other way around.
However, based on the accessibility of the inherited declaration, we want to chose either one or the other.
This is encoded in the last rules.
Rule~\rref{ord-ext} states that inherited members (through either of the three extends edges) are preferred over the path that traversed a $\lblLEX$-edge if the resulting declaration is accessible, and the path itself did not hide the declaration.
Conversely, if the declaration is inaccessible, or hidden by a path, rules~\rref{ord-lex-decl} and~\rref{ord-lex-path} indicate the lexically enclosing path has priority instead.
Together, these rules model all the constraints we set out above.

  \section{Equivalence Proofs}
\label{sec:equivalence-proofs}

In this section, we prove the equivalences of accessibility policies defined in~\cref{subsec:equivalences}.
We will use the specification as presented in~\cref{sec:aml-full-spec}.
For completeness, we repeat our definition of equivalence:
\accesspolicyequiv*
\noindent
For the of the equivalence proofs, the following lemmas will be used.

\begin{lemma}
  \label{lem:equiv-sym}
  \[
    \forall A, A'.\ (A \equiv A') \Rightarrow (A' \equiv A)
  \]
\end{lemma}
This lemma proves that the relation $\equiv$ is symmetric.

\begin{proof}
  Follows from symmetry of bi-implication.
\end{proof}

\begin{lemma}
  \label{lem:equiv-trans}
  \[
    \forall A, A', A''.\ (A \equiv A') \land (A' \equiv A'') \Rightarrow (A \equiv A'')
  \]
\end{lemma}
This lemma proves that the relation $\equiv$ is transitive.

\begin{proof}
  Follows from transitivity of (bi)-implication.
\end{proof}

\begin{lemma}
  \label{lem:prv-to-any}
  \[
    \forall A, s, \SG.\ (\premMAcc{\SG}{s}{p}{\PRV}) \Rightarrow (\premMAcc{\SG}{s}{p}{A})
  \]
\end{lemma}
This lemma proves that, if access is allowed for a path $p$ in scope $s$ in scope graph $\SG$ under access policy \PRV, then access is allowed by \emph{any} access policy $A$.

\begin{proof}
  This lemma is proven by implication elimination, and applying~\rref{ap-priv} on the goal, which yields a tautology.%
  \footnote{
    Note how our design choice of making the \rref{ap-priv} rule match on any access policy made this lemma trivial.
    If this rule would have matched the \PRV policy only, several cases would not be provable.
    For example, the $\MODof{S}$ case would reduce to proving $s \in S$ for any $S$, which is clearly impossible (\eg, $S = \emptyset$ is a counter-example).
  }
\end{proof}

\subparagraph*{Equivalences.}
Now, we can prove the first equivalence:

\begin{theorem}
  \label{thm:prt-equiv-smd}
  $\PRT \equiv \SMDof{\emptyset}$.
\end{theorem}

\begin{proof}
The statement to prove here is:
\begin{inferruleraw}
\forall \SG, s, p.\ (\premMAcc{\SG}{s}{p}{\PRT}) \Leftrightarrow (\premMAcc{\SG}{s}{p}{\SMDof{\emptyset}})
\end{inferruleraw}
The forward direction can be proven by implication elimination and application of the \rref{ap-smd-prot} rule (\cref{subfig:interaction-semantics}).
The backward direction can be proven by inversion on $\premMAcc{\SG}{s}{p}{\SMDof{\emptyset}}$.
This yields three cases:
\begin{itemize}
  \item $(\premMAcc{\SG}{s}{p}{\PRT}) \Rightarrow (\premMAcc{\SG}{s}{p}{\PRT})$, which is a tautology.
  \item $(\premMAcc{\SG}{s}{p}{\MODof{\emptyset}}) \Rightarrow (\premMAcc{\SG}{s}{p}{\PRT})$.
  We prove this by implication elimination and $\premMAcc{\SG}{s}{p}{\MODof{\emptyset}}$, which simplifies to $\bot$, as there exists no $s_c$ such that $s_c \in \emptyset$.
  \item The last case follows from the \rref{ap-priv} rule, which can be proven using~\cref{lem:prv-to-any}.
\end{itemize}
\vskip -\baselineskip
\end{proof}
Next, we prove the second equivalence:

\begin{theorem}
  \label{thm:prv-equiv-mod}
  $\PRV \equiv \MODof{\emptyset}$.
\end{theorem}

\begin{proof}
The statement to prove here is:
\begin{inferruleraw}
\forall \SG, s, p.\ (\premMAcc{\SG}{s}{p}{\PRV}) \Leftrightarrow (\premMAcc{\SG}{s}{p}{\MODof{\emptyset}})
\end{inferruleraw}
The forward direction follows from~\cref{lem:prv-to-any}.
The backwards direction can be proven using implication elimination and inversion.
This yields two cases:
\begin{itemize}
  \item \rref{ap-smc}: $\premMAcc{\SG}{s}{p}{\MODof{\emptyset}} \land \premMAcc{\SG}{s}{p}{\PRT}$. The left conjunct reduces to $s \in \emptyset$, which is absurd.
  \item \rref{ap-priv}: $\premMAcc{\SG}{s}{p}{\PRV}$. This is a tautology, as the goal is in the current set of assumptions.
\end{itemize}
%
\end{proof}
The third equivalence we prove is the following:

\begin{theorem}
  \label{thm:mod-equiv-smc}
  $\MODof{\emptyset} \equiv \SMCof{\emptyset}$.
\end{theorem}

\begin{proof}
The statement to prove here is:
\begin{inferruleraw}
\forall \SG, s, p.\ (\premMAcc{\SG}{s}{p}{\MODof{\emptyset}}) \Leftrightarrow (\premMAcc{\SG}{s}{p}{\SMCof{\emptyset}})
\end{inferruleraw}
The forward direction follows from inversion on the left-hand side of the implication, which yields two cases:
\begin{itemize}
  \item \rref{ap-int}: $s \in \emptyset$, which is absurd.
  \item \rref{ap-priv}: $\premMAcc{\SG}{s}{p}{\PRV}$, from which the goal follows by~\cref{lem:prv-to-any}.
\end{itemize}
The backward direction follows from a similar inversion, which yields two cases:
\begin{itemize}
  \item \rref{ap-smc}: $(\premMAcc{\SG}{s}{p}{\MODof{\emptyset}}) \land (\premMAcc{\SG}{s}{p}{\PRT})$, where the left conjunct is the goal to prove.
  \item \rref{ap-priv}: $\premMAcc{\SG}{s}{p}{\PRV}$, from which the goal follows by~\cref{thm:prv-equiv-mod}.
\end{itemize}
\vskip-\baselineskip
\end{proof}
Finally, we have the following equivalence.

\begin{theorem}
  \label{thm:prv-equiv-smc}
  $\PRV \equiv \SMCof{\emptyset}$.
\end{theorem}

\begin{proof}
  Follows from~\cref{thm:prv-equiv-mod}, \cref{thm:mod-equiv-smc}, and~\cref{lem:equiv-trans} (transitivity).
\end{proof}

  \section{Strict Partial Order}
\label{sec:preorder-proofs}

In this appendix, we prove the well-behavedness of strict partial order on access policies defined in~\cref{subsec:policy-order}.
This order can be defined as follows:
\begin{definition}[$A <_A A$]
  \label{def:policy-order}
  \begin{alignat*}{3}
                  \PRV &<_A \SMCof{S}                                                               \\
    \SMCof{S} &<_A \SMCof{S'}\hspace{2em} &\text{ if }\ &S \subsetneq S' \\
    \SMCof{S} &<_A \PRT                                                                             \\
    \SMCof{S} &<_A \MODof{S}                                                               \\
    \MODof{S} &<_A \MODof{S'}\hspace{2em} &\text{ if }\ &S \subsetneq S' \\
    \MODof{S} &<_A \SMDof{S}                                                               \\
                  \PRT &<_A \SMDof{S}                                                               \\
    \SMDof{S} &<_A \SMDof{S'}\hspace{2em} &\text{ if }\ &S \subsetneq S' \\
    \SMDof{S} &<_A \PUB                                                                             \\[0.25em]
                     A &<_A A''                             &\text{ if }\ &\exists A'.\ (A <_A A') \land (A' <_A A'')
  \end{alignat*}
\end{definition}
To prove the well-behavedness of this relation, we will use the following lemmas.
\begin{lemma}
  \label{lem:subset-contains}
  \[
    s \in S 
      \land 
      (S \subsetneq S')
    \Rightarrow
      s \in S'
  \]
\end{lemma}
\begin{proof}
  By definition of subset, see \eg~\cite[p.~70]{epp2010discrete}.
\end{proof}
Now, we state that access to definitions with \MOD-accessibility is preserved when adding more modules to its argument.
\begin{lemma}
  \label{lem:mod-weakening}
  \[
    \forall \SG, s, p, S, S'.\ (S \subsetneq S') \land (\premMAcc{\SG}{s}{p}{\MODof{S}}) \Rightarrow (\premMAcc{\SG}{s}{p}{\MODof{S'}})
  \]
\end{lemma}
\begin{proof}
  Implication elimination and inversion yields two cases:
  \begin{itemize}
    \item $\rref{ap-int}$: this case can be proven by applying \rref{ap-int} on the goal.
      This yields the following assumptions:
      \begin{itemize}
        \item $\premEncM{\SG}{s}{S_M}$
        \item $s_m \in S_M$
        \item $s_m \in S$
      \end{itemize}
      and proof goal: $\exists M', s_m'.\; (\premEncM{\SG}{s}{S_M'}) \land (s_m' \in S_M') \land (s_m' \in S')$.
      We instantiate $S_M'$ with $S_M$ and $s_m'$ with $s_m$.
      Then, the first two conjuncts follow by assumption, and the last one by~\cref{lem:subset-contains}.
    \item \rref{ap-priv}: goal follows from~\cref{lem:prv-to-any}.
  \end{itemize}
  \vskip -\baselineskip
\end{proof}
We now restate and prove the theorem that defines well-behavedness of this order.
\accpolorderwellbehaved*

\begin{proof}
  We prove this implication by induction on the $<_A$-relation.
  Each claim in this proof corresponds to a case from~\cref{def:policy-order}.
  We first cover the base cases.

  \ordCaseProof{\PRV}{\SMCof{S}}{
    Follows from~\cref{lem:prv-to-any}.
  }

  \ordCaseProof[S \subsetneq S']{\SMCof{S}}{\SMCof{S}'}{
    By implication elimination, we need to prove $\premMAcc{\SG}{s}{p}{\SMCof{S}'}$ from the following assumptions:
    \begin{enumerate}
      \item $S \subsetneq S'$, and
      \item \label{it:smc} $\premMAcc{\SG}{s}{p}{\SMCof{S}}$.
    \end{enumerate}
    We proceed by inversion on the second assumption, which yields two cases:
    \begin{itemize}
      \item $\rref{ap-smc}$: proven by applying $\rref{ap-smc}$ on the goal, which gives two residual proof obligations:
      \begin{itemize}
        \item $\premMAcc{\SG}{s}{p}{\PRT}$: proven by assumption obtained from inversion on assumption~\ref{it:smc}.
        \item $\premMAcc{\SG}{s}{p}{\MODof{S}'}$: proven by~\cref{lem:mod-weakening}.
      \end{itemize}
      \item \rref{ap-priv}: goal follows from~\cref{lem:prv-to-any}.
      \vskip -\baselineskip
    \end{itemize}
  }

  \ordCaseProof{\SMCof{S}}{\PRT}{
    By implication elimination and inversion, we obtain two cases:
    \begin{itemize}
      \item $\rref{ap-smc}$: goal follows from assumption.
      \item \rref{ap-priv}: goal follows from~\cref{lem:prv-to-any}.
    \end{itemize}
    \vskip -\baselineskip
  }

  \ordCaseProof{\SMCof{S}}{\MODof{S}}{
    By implication elimination and inversion, we obtain two cases:
    \begin{itemize}
      \item \rref{ap-int}: goal follows from assumption.
      \item \rref{ap-priv}: goal follows from~\cref{lem:prv-to-any}.
    \end{itemize}
    \vskip -\baselineskip
  }

  \ordCaseProof[S \subsetneq S']{\SMCof{S}}{\MODof{S}}{
    Follows from~\cref{lem:mod-weakening}.
  }

  \ordCaseProof{\MODof{S}}{\SMDof{S}}{
    Follows from implication elimination and applying~\rref{ap-smd-mod}.
  }

  \ordCaseProof{\PRT}{\SMDof{S}}{
    Follows from implication elimination and applying~\rref{ap-smd-prot}.
  }

  \ordCaseProof[S \subsetneq S']{\SMDof{S}}{\SMDof{S'}}{
    By implication elimination and inversion, we obtain three cases:
    \begin{itemize}
      \item \rref{ap-smd-mod}: goal follows from applying \rref{ap-smd-mod} and~\cref{lem:subset-contains}.
      \item \rref{ap-smd-prot}: goal follows from applying \rref{ap-smd-prot}.
      \item \rref{ap-priv}: goal follows from~\cref{lem:prv-to-any}.
    \end{itemize}
    \vskip -\baselineskip
  }

  \ordCaseProof{\SMDof{S}}{\PUB}{
    Follows from~\rref{ap-pub}.
  }
  %
  Finally, there is an inductive case left:
  
  \begin{proofcase}
    \begin{alignat*}{7}
      \forall A, A''.\ &(\exists A'.\ &&A  &&<_A A'  &&\Rightarrow (\forall \SG, s, p.\ \premMAcc{\SG}{s}{p}{A}  \Rightarrow \premMAcc{\SG}{s}{p}{A'}) \land {}\\
                       &              &&A' &&<_A A'' &&\Rightarrow (\forall \SG, s, p.\ \premMAcc{\SG}{s}{p}{A'} \Rightarrow \premMAcc{\SG}{s}{p}{A''})) \\
                       & \Rightarrow  &&A  &&<_A A'' &&\Rightarrow (\forall \SG, s, p.\ \premMAcc{\SG}{s}{p}{A}  \Rightarrow \premMAcc{\SG}{s}{p}{A''})
    \end{alignat*}
  \end{proofcase}

  \begin{claimproof}
    In this case, we have the following induction hypotheses (for some $A'$):
    \begin{enumerate}
      \item \label{it:ind-1} $A  <_A A'  \Rightarrow (\forall \SG, s, p.\ \premMAcc{\SG}{s}{p}{A}  \Rightarrow \premMAcc{\SG}{s}{p}{A'})$, and
      \item \label{it:ind-2} $A' <_A A'' \Rightarrow (\forall \SG, s, p.\ \premMAcc{\SG}{s}{p}{A'} \Rightarrow \premMAcc{\SG}{s}{p}{A''})$.
    \end{enumerate}
    By implication elimination, we obtain the following assumptions:
    \begin{enumerate}
      \item \label{it:AAp} $A <_A A'$
      \item \label{it:AApp} $A' <_A A''$
      \item \label{it:Aacc} $\premMAcc{\SG}{s}{p}{A}$
    \end{enumerate}
    From assumptions~\ref{it:AAp} and~\ref{it:Aacc}, and induction hypothesis~\ref{it:ind-1}, we infer $\premMAcc{\SG}{s}{p}{A'}$.
    From this, assumption~\ref{it:AApp}, and induction hypothesis~\ref{it:ind-2}, we infer our goal.
  \end{claimproof}
  Together, these cases prove~\cref{thm:acc-pol-order-well-behaved}.
\end{proof}

  \section{Soundness Proofs}
\label{sec:characterizing-predicates}

In this appendix, we give soundness proofs for the following access policies:
\begin{itemize}
  \item \lit*{private}
  \item \lit*{protected}
  \item \lit*{internal}
\end{itemize}
As access to \lit*{public} is allowed unconditionally, no soundness proof for it is needed.
Moreover, as \lit*{private protected} and \lit*{protected internal} are defined as a conjunction or disjunction of \lit*{protected} and \lit*{internal},
their proofs are trivially derived from the soundness proofs of the access modifiers they are based on.
Thus, we omitted those for brevity.

\Cref{subsec:proof-techniques} briefly discusses the machinery we use to reason over AML proof trees.
After that,~\cref{sub:defs} presents a significant number of definitions, lemma's and assumptions we make.
Readers are recommended to skip this section, and refer back to it when needed.
\Cref{subsec:priv-sound,subsec:prot-sound,subsec:int-sound} present the actual soundness proofs.

\subsection{Proof Techniques}
\label{subsec:proof-techniques}

To prove that our language (AML) faithfully represents the characterizing properties presented in~\cref{subsec:soundness-theorems}, we need to reason about the \emph{structure} of \emph{proof trees} that are valid according to our AML specification.
To this end, we use $\ctx$ to represent a \emph{typing derivation context}; i.e., a typing derivation with a hole, where the typing derivation is valid according to the AML typing rules presented in our paper and in the Statix definition in the accompanying artifact.
This notion of typing derivation context is analogous to the well-known notion of \emph{evaluation context} \`{a} la Felleisen~and~Hieb~\cite{FelleisenH92}.%
\footnote{
  Alternatively, derivation contexts can be understood as \emph{zippers}~\cite{Huet97} over \emph{proof trees}.
}
The notation $\inCtx{A}$ represents the full typing derivation given by plugging the derivation of $A$ into the typing derivation context $\ctx$.
This assertion can be interpreted the other way around as well.
For example, it might be useful to think of a statement such as $P = \inCtx{A}$ as `$P$ contains a derivation of $A$ at the subtree pointed to by $\ctx$'.

\subsubsection{Proofs using Derivation Contexts}

Standard backward reasoning proof techniques, such as case analysis and inversion, are valid inside a derivation context as well.
For example, given the inference rules:
\[
    \inference{
      B
    }{
      A
    }
    \qquad
    \inference{
      C
    }{
      A
    }
\]
and a proof tree $P = \inCtx{A}$, by inversion we can deduce that
\[
  \exists \ctx'.\ (P = \inCtx[']{B} \lor P = \inCtx[']{C})
\]
That is, there exists a proof context $\ctx'$ where plugging either a derivation of $B$ or a derivation of $C$, is equivalent to $P$.

However, this does not hold for forward reasoning.
For example, consider the following inference rules:
\[
    \inference[\textsc{R-A}]{
      {}
    }{
      A
    }
    \qquad
    \inference[\textsc{R-AB}]{
      A
    }{
      B
    }
    \qquad
    \inference[\textsc{R-AC}]{
      A
    }{
      C
    }
    \qquad
    \inference[\textsc{R-BD}]{
      B
    }{
      D
    }
    \qquad
    \inference[\textsc{R-CD}]{
      C
    }{
      D
    }
\]
and a proof tree
\[
  P = \vcenter{
    \[\frac{
      \displaystyle
      \frac{
        A
      }{
        C
      }
    }{
      D
    }\]
  }
\]
Clearly, there exists a context $\ctx$ such that $P = \inCtx{A}$.
However, is is incorrect to derive from this fact and \textsc{R-AB} that there exists a $\ctx'$ such that $P = \inCtx[']{B}$.
Although $B$ could have been derived from $A$, this does not actually happen in $P$.

We can do the following, however.
From $P = \inCtx{A}$ in the above setup, we derive that 
\[
  P = A \lor \exists\ctx'.\ P = \inCtx{B} \lor P = \inCtx{C}
\]
That is, either $A$ was the final conclusion of $P$, or $A$ was derived in a derivation of either $B$ or $C$.

Next, we use the following lemma to equate propositions in equal positions in the proof tree.
\begin{lemma}[Equality of Substitutions]
  \label{lem:eq-in-ctx}
  \[
    \forall \ctx, A, A'.\ (\inCtx{A} = \inCtx{A'}) \Rightarrow A = A'
  \]
\end{lemma}

Finally, propositions can vacuously be lifted out of proof contexts.
\begin{lemma}[Derivation Lifting]
  \label{lem:lift-derivation}
  \[
    \forall P, \ctx, A.\ P = \inCtx{A} \Rightarrow A
  \]
\end{lemma}
The rationale behind this lemma is that $\inCtx{A}$ is substituting \emph{a derivation} of $A$.
If there exists a valid proof $P$ in which such a derivation occurs, there must exist a derivation of $A$, from which $A$ follows.

\subsection{Auxiliary Definitions}
\label{sub:defs}

In this section, we will present some auxiliary definitions, lemma's and assumptions that we will repeatedly use in the proof.
Readers are recommended to skip this section, and refer back to it when reading the proofs in the subsequent sections.

\subsubsection{Scope Graphs}

The first lemma states that each result in the query answer should have a corresponding declaration in the scope graph.
\begin{assumption}[Query Declaration]
  \label{lem:query-decl}
  \[
    \query{\SG}{s}{\mathit{re}}{D}{\textit{ord}}{R} \Rightarrow (\forall \langle p, d \rangle \in R.\, \mathsf{tgt}(p) \typeedget[l] d \in \SG)
  \]
\end{assumption}
\begin{proofsketch}
This follows immediately from the query resolution calculus~\cite[Fig.\@ 1, (NR-Rel)]{AntwerpenPRV18}.
\end{proofsketch}

The first lemma states that each result in the query answer should have a corresponding declaration in the scope graph.
\begin{assumption}[Query Paths]
  \label{lem:query-path}
  \[
    \query{\SG}{s}{\mathit{re}}{D}{\textit{ord}}{R} \Rightarrow (\forall \langle p, d \rangle \in R.\, \premMatchRE{p}{\mathit{re}})
  \]
\end{assumption}
\begin{proofsketch}
This follows immediately from the query resolution calculus~\cite[Fig. 1, (NR-Rel)]{AntwerpenPRV18}.
\end{proofsketch}

Second, we pose a lemma that states that queries have only a single result.
\begin{assumption}[Query Answer Uniqueness]
  \label{lem:query-answer-unique}
  \[
    \left(\query{\SG}{s}{\mathit{re}}{D}{\textit{ord}}{R}\right) \land \left(\query{\SG}{s}{\mathit{re}}{D}{\textit{ord}}{R'}\right)
    \Rightarrow
    R = R'
  \]
\end{assumption}
This should follow from the set-comprehension semantics of queries~\cite[\S3.1]{RouvoetAPKV20}.
We omit the proof for brevity.

\subsubsection{AML}

In this section, we pose some lemmas about AML typing derivations and the scope graphs that support those.

For convenience, we assume that each AST term has a particular \emph{index} associated with it, that distinguishes it from (structurally) equivalent AST nodes at \emph{different positions}.
For example, we often need to reason about references and declarations, which are structurally equivalent (\ie, both represented by a name~$x$).
To distinguish those, we assign indices~$r$ and~$d$ (written as~$x_r$ and~$x_d$) to make clear if we are referring to the reference or the declaration.
For terms that we annotate with an index, we consider the index part of the equality relation; \eg,~$x_r \neq x_d$.

\begin{assumption}
  \label{asmp:index-unique}
  Each AST node in a typing derivation~$P_\SG$ has a unique index that distinguishes it from all other AST nodes in~$P_\SG$.
\end{assumption}

For our proofs, we need some `minimality' guarantees on the scope graph $\SG$ that we did not make explicit in our initial presentation of AML's type system.
That is, the scope graph should contain exactly the declarations asserted by the program; no more, no less.
More specifically, we want each declaration/edge in the scope graph to correspond to some construct in the underlying program.
We state this in the form of \emph{support} lemmas.

\begin{lemma}[Scope Support]
  \label{lem:scope-support}
  \[
    \forall P_\SG.\, \left(s \in \SG\right) \Leftrightarrow \left(\exists \ctx.\, P_\SG = \inCtx{s \in \SG}\right)
  \]
\end{lemma}

\begin{lemma}[Edge Support]
  \label{lem:edge-support}
  \[
    \forall P_\SG.\, \left(s \scopeedget[l] s' \in \SG\right) \Leftrightarrow \left(\exists \ctx.\, P_\SG = \inCtx{s \scopeedget[l] s' \in \SG}\right)
  \]
\end{lemma}

\begin{lemma}[Declaration Support]
  \label{lem:decl-support}
  \[
    \forall P_\SG.\, \left(s \typeedget[l] d \in \SG\right) \Leftrightarrow \left(\exists \ctx.\, P_\SG = \inCtx{s \typeedget[l] d \in \SG}\right)
  \]
\end{lemma}

\begin{proofsketch}
  The forward direction of these lemmas is equivalent to the notion of \emph{support} by Rouvoet et al.~\cite[§3.3]{RouvoetAPKV20}.
  Their declarative semantics `collects' all assertions on the scope graph in a support parameter~$\sigma$, and propagates them to the root of the proof tree.
  The path through which scope graph assertions are propagated to the root is, in our notation, represented by the context $\ctx$.
  At the root of the proof tree, it is asserted that~$\sigma$ supports the full scope graph~$\SG$.

  This notion of support is intrinsic to the semantics of Statix.
  As such, by virtue of being expressed in Statix, our specification of AML adheres to these lemmas.
  The backward direction of these lemmas holds vacuously (\cref{lem:lift-derivation}).
\end{proofsketch}

Additionally, we assert that each scope assertion is unique.
\begin{lemma}[Scope Uniqueness]
  \label{lem:scope-unique}
  \[
    \forall \ctx_1,\, \ctx_F2.\, \inCtx[_1]{s \in SG} = \inCtx[_2]{s \in SG} \Leftrightarrow \ctx_1 = \ctx_2
  \]
\end{lemma}

\begin{proofsketch}
  The forward direction of this lemmas enforced by the \emph{disjoin union} used to define \emph{support} by Rouvoet et al.~\cite[§3.3]{RouvoetAPKV20}.
  This notion of support is intrinsic to the semantics of Statix.
  As such, by virtue of being expressed in Statix, our specification of AML adheres to these lemmas.
  The backward direction of these lemmas holds vacuously (\cref{lem:lift-derivation}).
\end{proofsketch}

\subparagraph{Uniqueness of Variable Declarations.}
Next, we state that every declaration in the program corresponds to a unique declaration in the scope graph:
\begin{lemma}[Uniqueness of Declarations]
  \label{lem:decl-ctx-unique}

  For any scope graph $\SG$, declaration $x_d$, contexts $\ctx_1$ and $\ctx_2$, scopes $s_1$ and $s_2$, types $T_1$ and $T_2$, and access policies $A_1$ and $A_2$, it holds that
    \[
      \left(
        \inCtx[_1]{s_1 \typeedget[\lblVAR] (\defdecl{x_d}{T_1}{A_1}) \in \SG}
        =
        \inCtx[_2]{s_2 \typeedget[\lblVAR] (\defdecl{x_d}{T_2}{A_2}) \in \SG}
      \right)
      \Rightarrow
      \ctx_1 = \ctx_2
    \]
\end{lemma}
This lemma states that, if plugging in an assertion of the declaration of $x_d$ in $\ctx_1$ and~$\ctx_2$ yields the same proof, it follows that $\ctx_1$ and $\ctx_2$ are equal.
That is, there is only one single position in the proof tree where the existence of a declaration containing $x_d$ is asserted.

\begin{proofsketch}
This lemma should follow from the facts that:
\begin{itemize}
  \item $x_d$ is unique in the program (\cref{asmp:index-unique}), and
  \item no typing rules in the AML specification that match on a part of the AST containing~$x_d$ propagate it to multiple premises.
\end{itemize}
These conditions guarantee that $x_d$ is propagated to at most one declaration assertion.
\end{proofsketch}
The second condition of this proof sketch is necessary, because otherwise $x_d$, despite of being unique in the AST, could be part of multiple sub-trees that contain an assertion of the declaration.

By uniqueness of derivations in some context, it follows that the declaration scope, type and access policy should be unique for each declaration too.

\begin{corollary}[Uniqueness of Declaration Parameters]
  \label{lem:decl-params-unique}
  Given 
  \[
    d_1 = \defdecl{x_d}{T_1}{A_1} 
    \qquad
    d_2 = \defdecl{x_d}{T_2}{A_2}
  \]
  the following holds:
  \[
    \left(
      \inCtx[_1]{s_1 \typeedget[\lblVAR] d_1 \in \SG}
      =
      \inCtx[_2]{s_2 \typeedget[\lblVAR] d_2 \in \SG}
    \right)
    \Rightarrow
    s_1 = s_2 
    \land 
    T_1 = T_2 
    \land 
    A_1 = A_2
  \]
\end{corollary}

\begin{proof}
  From \cref{lem:decl-ctx-unique}, it follows that $\ctx_1 = \ctx_2$.
  Now, from~\cref{lem:eq-in-ctx}, we infer that
  \[
    \left(s_c \typeedget[\lblVAR] d_1 \in \SG\right) = \left(s_c' \typeedget[\lblVAR] d_2 \in \SG\right)
  \]
  from which the desired equalities follow immediately.
\end{proof}
Similarly, module declarations are unique:

\begin{assumption}[Uniqueness of Modules]
  \label{lem:mod-decl-ctx-unique}

  For any scope graph $\SG$, module scope $s_m$, declarations $x_1$ and $x_2$, contexts $\ctx_1$ and $\ctx_2$, scopes $s_1$ and $s_2$, it holds that
    \[
      \left(
        \inCtx[_1]{s_1 \typeedget[\lblMOD] (moddecl{x_2}{s_m}) \in \SG}
        =
        \inCtx[_2]{s_2 \typeedget[\lblMOD] (\moddecl{x_1}{s_m}) \in \SG}
      \right)
      \Rightarrow
      \ctx_1 = \ctx_2
    \]
\end{assumption}

\begin{proofsketch}
  Using~\cref{lem:scope-unique}, and the fact that~\rref{mod-ok} takes ownership over modules scopes,
  one can infer that module declarations are unique up to module scope $s_m$.
\end{proofsketch}

\begin{corollary}[Uniqueness of Module Parameters]
  \label{lem:mod-decl-params-unique}

  For any scope graph $\SG$, module scope $s_m$, declarations $x_1$ and $x_2$, contexts $\ctx_1$ and $\ctx_2$, scopes $s_1$ and $s_2$, it holds that
    \[
      \left(
        \inCtx[_1]{s_1 \typeedget[\lblMOD] (\moddecl{x_2}{s_m}) \in \SG}
        =
        \inCtx[_2]{s_2 \typeedget[\lblMOD] (\moddecl{x_1}{s_m}) \in \SG}
      \right)
      \Rightarrow
      s_1 = s_2 \land x_1 = x_2
    \]
\end{corollary}

\begin{proofsketch}
  Similar to~\cref{lem:decl-params-unique}.
\end{proofsketch}

\subparagraph{Resolve Predicates.}
Using this, we define a $\mathsf{resolveTo}$ predicate, which captures what it means for a variable to resolve to some declaration in a particular scope.
\begin{definition}[$\mathsf{resolveTo}$]
  \label{def:resolveTo}
  \begin{inferruleraw}
    \inference{
      P_\SG =
      \inCtx{
        \query{\SG}{s}{\reLEX}{\mathsf{isVar}_{x_r}}{\lblVAR < \lblEXT < \lblLEX}{\setOf{ \langle p, \defdecl{x_d}{T}{A} \rangle }}
      }
    }{
      \mathsf{resolveTo}_{P_\SG}(s, x_r) =  \defdecl{x_d}{T}{A}
    }
  \end{inferruleraw}
  \begin{inferruleraw}
    \inference{
      P_\SG = \inCtx[_1]{\premExp{\SG}{s}{e_r.x_r}{T}}
      &
      P_\SG = \inCtx[_2]{\premExp{\SG}{s}{e_r}{\tyINST{s_c}}}
      \\
      P_\SG = \inCtx[_3]{\query{\SG}{s_c}{\reMEM}{\mathsf{isVar}_{x_r}}{\lblVAR < \lblEXT}{\setOf{ \langle p, \defdecl{x_d}{T}{A}\rangle }}}
    }{
      \mathsf{resolveTo}_{P_\SG}(s, x_r) =  \defdecl{x_d}{T}{A}
    }
  \end{inferruleraw}
\end{definition}
This predicate states that the typing derivation $P_\SG$ witnesses that $x_r$ resolves to $x_d$ \emph{in scope}~$s$ in scope graph $\SG$.
The first case asserts there is a query for lexical references looking for $x_r$ in scope $s$, resulting in declaration $x_d$.
The second case matches on field accesses.
To this end, the first premise asserts there exists some expression $e_r$ on which field $x_r$ is accessed.
This expression should have type $\tyINST{s_c}$, for some class scope $s_c$ (second premise).
Then, the last premise asserts the existence of a query for member access that resolve $x_r$ in $s_c$ to~$x_d$.

Next, we the following lemma states that resolving $x_r$ to $x_d$ implies that there exists some path~$p$ that grants access to $x_d$ in the scope $s_r$ in which $x_r$ is resolved.

\begin{lemma}[Resolve implies Accessible]
  \label{lem:res-to-acc}
  \begin{inferruleraw}
    \mathsf{resolveTo}_{P_\SG}(s_r, x_r) = \defdecl{x_d}{T}{A}\ 
    \Rightarrow\ 
    \exists p.\
    \premMAcc{\SG}{s_c}{p}{A}
  \end{inferruleraw}
\end{lemma}

\begin{proof}
Inversion on $\mathsf{resolveTo}_{P_\SG}(s_r, x_r)$ gives two cases: one for lexical access, and one for member access.

\proofsubparagraph{Lexical Access.}
In this case, inversion gives us:
\begin{inferruleraw}
  P_\SG = 
  \inCtx{
    \query{\SG}{s_r}{\reLEX}{\mathsf{isVar}_{x_r}}{\lblVAR < \lblEXT < \lblLEX}{\setOf{ \langle p, \defdecl{x_d}{T}{A} \rangle }}
  }
\end{inferruleraw}
By forward reasoning (under the assumption that $P_\SG$ is a typing derivation for a complete program),
we infer that there exists a context $\ctx'$ such that
\[
  P_\SG = \inCtx[']{\premExp{\SG}{s}{x_r}{T}}
\]
By inversion, we infer that there exists a context $\ctx''$ such that
\[
  P_\SG = \inCtx['']{\premMAcc{\SG}{s}{p}{A}}
\]

\proofsubparagraph{Member Access.}
Follows from a similar argument.
\end{proof}

Using $\mathsf{resolveTo}$, $\mathsf{resolve}$ is then defined as as follows.
\begin{definition}[$\mathsf{resolve}$]
  \label{def:resolve}
  \begin{inferruleraw}
    \inference{
      \mathsf{resolveTo}_{P_\SG}(s, x_r) = \defdecl{x_d}{T}{A}
    }{
      \mathsf{resolve}_{P_\SG}(x_r) = x_d
    }
  \end{inferruleraw}
\end{definition}
That is, $x_r$ resolves to $x_d$ if there is some scope $s$ in which it was resolved.

Finally, we define module resolution as follows:
\begin{definition}[$\mathsf{resolveModIn}$]
  \label{def:resolve-mod-in}
  \begin{inferruleraw}
    \inference{
      P_\SG = 
      \inCtx{
        \query{\SG}{s_r}{\reclos{\lblLEX}\lblMOD}{\mathsf{isMod}_{x_r}}{\lblMOD < \lblLEX}{\setOf{ \langle p, d \rangle }}
      }
    }{
      \mathsf{resolveModIn}_{P_\SG}(x_r, s_r) = d
    }
  \end{inferruleraw}
\end{definition}
\begin{definition}[$\mathsf{resolveMod}$]
  \label{def:resolve-mod}
  \begin{inferruleraw}
    \inference{
      \mathsf{resolveModIn}_{P_\SG}(x_r, s_r) = \moddecl{x_d}{s_m}
    }{
      \mathsf{resolveMod}_{P_\SG}(x_r) = s_m
    }
  \end{inferruleraw}
\end{definition}
These definitions are analogous to $\mathsf{resolve}$ and $\mathsf{resolveTo}$.

\subparagraph*{Module Relation Predicates}

Now, we state he $\mathsf{enclosingMod}$ predicate as follows:
\begin{definition}[$\mathsf{enclosingMod}$]
  \label{def:defEncMod}
  \begin{inferruleraw}
    \inference{
      \mathsf{resolveTo}_{P_\SG}(s_r, x_r) = d
      &
      \premEncMI{\SG}{s_r}{s_m}
    }{
      \mathsf{enclosingMod}_{P_\SG}(x_r) = s_m
    }
  \end{inferruleraw}
\end{definition}

\subparagraph*{Class Relation Predicates}

Next, we define some functions/relations between classes.
Next, we define $\mathsf{definingClass}$ as follows:
\begin{definition}[$\mathsf{definingClass}$]
  \label{def:defCls}
  \begin{inferruleraw}
    \inference{
      P_\SG =
      \inCtx{
        \premCMem{\SG}{s}{\mathit{acc} \mathrel{\lit*{var}} x_d \mathrel{\lit*{=}} e}
      }
    }{
      \mathsf{definingClass}_{P_\SG}(x_d) = s
    }
  \end{inferruleraw}
\end{definition}
This rule states that the declaration of $x_d$ occurs in some scope $s$.
This must necessarily be the scope of the class that defines $x_d$, so we return that as a canonical identifier for the defining class of $x_d$.

Furthermore, we define the $\mathsf{subClass}$ predicate as follows:
\begin{definition}[$\mathsf{subClass}$]
  \label{def:defSub}
  \begin{inferruleraw}
    \inference{
      \query{\SG}{s_c}{\reclos{\lblEXT}}{\mathsf{isScope}_{s_p}}{}{R}
      &
      R \neq \emptyset
    }{
      \mathsf{subClass}_{P_\SG}(s_c, s_p)
    }
  \end{inferruleraw}
\end{definition}
This definition states that $s_c$ is a subclass of $s_p$ when there exists a path from $s_c$ to $s_p$ that traverses only $\lblEXT$-edges.

Finally, we define the $\mathsf{enclosingClass}$ predicate as follows:
\begin{definition}[$\mathsf{enclosingClass}$]
  \label{def:defEnc}
  \begin{inferruleraw}
    \inference{
      \mathsf{resolveTo}_{P_\SG}(s_r, x_r) = d
      &
      \premEncC{\SG}{s_r}{S_C}
      &
      s_c \in S_C
    }{
      \mathsf{enclosingClass}_{P_\SG}(x_r, s_c)
    }
  \end{inferruleraw}
\end{definition}
First, we capture the scope $s_r$ in which $x_r$ resolved using $\mathsf{resolveTo}$.
Next we assert that $s_c$ is an enclosing class of $s_r$.

\subsection{Soundness of Private Access Validation}
\label{subsec:priv-sound}

First, we define a $\mathsf{private}(x_d)$ that holds when $x_d$ is a variable declared with the \lit*{private} access modifier:
\begin{definition}[$\mathsf{private}$]
  \label{def:priv}
  \begin{inferruleraw}
    \inference{
      P_\SG =
      \inCtx{
          \premCMem{\SG}{s}{\lit*{private} \mathrel{\lit*{var}} x_d \mathrel{\lit*{=}} e}
      }
    }{
      \mathsf{private}_{P_\SG}(x_d)
    }
  \end{inferruleraw}
\end{definition}
The premise of this rule states that $x_d$ was declared with the \texttt{private} access modifier.

\begin{lemma}
  \label{lem:x-decl}
  \begin{inferruleraw}
    \mathsf{private}_{P_\SG}(x_d)\ 
    \Rightarrow\
    \exists s_c,\, T.\
    s_c \typeedget[\lblVAR] (\defdecl{x_d}{T}{\PRV}) \in \SG
  \end{inferruleraw}
\end{lemma}

\begin{proof}
From $\mathsf{private}_{P}(x_d)$, we derive $\inCtx{\premCMem{\SG}{s_c}{\lit*{private} \mathrel{\lit*{var}} x_d \mathrel{\lit*{=}} e}}$, for some $\ctx$.
Inversion yields a single case, using rule~\rref{d-def}, from which we infer that there exists some $\ctx_1$, $\ctx_2$ and $\ctx_3$ such that
\[
  \inCtx[_1]{\premExp{\SG}{s_c}{e}{T}}
  \qquad
  \inCtx[_2]{\premAcc{\SG}{s_c}{\lit*{private}}{\PRV}}
  \qquad
  \inCtx[_3]{s_c \typeedget[\lblVAR] (\defdecl{x_d}{T}{\PRV}) \in \SG}
\]
From the third premise, the goal follows by~\cref{lem:lift-derivation}.
\end{proof}
We can now prove the correctness of our \texttt{private} access modifier specification.
The desired semantics of the \texttt{private} access modifier is defined as follows:
\begin{theorem}[Soundness of \texttt{private} member access, informally]
\label{def:priv-inf}
If $x_d$ is a private field, any reference $x_r$ that resolves to $x_d$ must live in the class that defines $x_d$.
\end{theorem}
Captured in more mathematical notation, \cref{def:priv-inf} is equivalent to the following:

\begin{theorem}[Soundness of \texttt{private} member access, formally]
\label{def:priv-form}
Let $P_\SG$ be a valid typing derivation for an AML program with scope graph $\SG$.
Then, for all possible $x_r$ and $x_d$,
\begin{align*}
  &\mathsf{resolve}_{P_\SG}(x_r) = x_d 
    \mathrel{\land} 
    \mathsf{private}_{P_\SG}(x_d)\ 
  \Rightarrow {}\\
  &\hspace{1em}
  \exists s_c. 
    \mathsf{definingClass}_{P_\SG}(x_d) = s_c 
    \mathrel{\land} 
    \mathsf{enclosingClass}_{P_\SG}(x_r, s_c)
\end{align*}
\end{theorem}
Here, $\mathsf{resolve}_{P_\SG}(x_r) = x_d$ means that the typing derivation of $P_{\SG}$ witnesses that reference $x_r$ resolves to declaration $x_d$,
$\mathsf{private}_{P_\SG}(x_d)$ means that $x_d$ is a class field with \texttt{private} access,
$\mathsf{definingClass}_{P_\SG}(x_d) = c$ means that $x_d$ is a member of the class represented by scope $s_c$, and
$\mathsf{enclosingClass}_{P_\SG}(x_r, c)$ holds if class $c$ contains the reference $x_r$.

Next, we prove that~\cref{def:priv-form} is satisfied by our specification.
\begin{proof}

  First, we prove that there is a class $s_c$ in which $x_d$ is defined.
  Then, we prove that this class is an enclosing class of $x_r$.
  
  \proofsubparagraph{Defining Class.}
  We get $\mathsf{definingClass}_{P_\SG}(x_d) = s_c$ in the conclusion as follows.
  By inversion on the premise $\mathsf{private}_{P_\SG}(x_d)$, we get a $\ctx$, $s_c$, and $e$ for which $\inCtx{\premCMem{\SG}{s_c}{\lit*{private} \mathrel{\lit*{var}} x_d \mathrel{\lit*{=}} e}}$.
  From this and~\cref{def:defCls}, the goal follows.

  \proofsubparagraph{Enclosing Class.}
  Next, we apply inversion on resolve (\cref{def:resolve}), from which we infer that there exists some scope~$s_r$, such that
  \[
    \mathsf{resolveTo}_{P_\SG}(s_r, x_r) = \defdecl{x_d}{T'}{A}
  \]
  By further inversion of $\mathsf{resolveTo}$ (\cref{def:resolveTo}), and~\cref{lem:query-decl}, we infer that,   for some scope~$s_c'$:
  \[
    s_c' \typeedget[\lblVAR] \defdecl{x_d}{T'}{A} \in \SG
  \]
  Next, by inversion on $\inCtx{\premCMem{\SG}{s_c}{\lit*{private} \mathrel{\lit*{var}} x_d \mathrel{\lit*{=}} e}}$
  and~\cref{lem:lift-derivation}, we infer that
  \[
    s_c \typeedget[\lblVAR] \defdecl{x_d}{T}{\PRV} \in \SG
  \]
  Then, by~\cref{lem:decl-params-unique}, we infer that
  \[
    s_c = s_c'
    \qquad
    T = T'
    \qquad
    A = \PRV
  \]
  Thus, we have established that the access modifier included in the query result is \PRV.
  To prove the right conjunct of our goal (\ie, that $s_c$ is an enclosing class of $s_r$), we first apply~\cref{lem:res-to-acc}.
  This gives us
  \[
    \premMAcc{\SG}{s_r}{p}{\PRV}
  \]
  for the resolution path $p$.
  We know by~\cref{lem:query-decl} that $\mathsf{tgt}(p) = s_c$, as~$p$ was the path towards declaration~$x_d$.
  By inversion on this premise (\rref{ap-priv}), we obtain $\premEncC{\SG}{s_r}{S_C}$ and $\mathsf{tgt}(p) \in S_C$ for some class scope set $S_C$.
  Combined with our earlier result that $\mathsf{resolveTo}_{P_\SG}(s_r, x_r) = \defdecl{x_d}{T}{A}$, this proves $\mathsf{enclosingClass}(x_r, s_c)$ by~\cref{def:defEnc}.
\end{proof}
\subsection{Soundness of Protected Access Validation}
\label{subsec:prot-sound}

First, we define a $\mathsf{protected}(x_d)$ that holds when $x_d$ is a variable declared with the \lit*{protected} access modifier:
\begin{definition}[$\mathsf{protected}$]
  \label{def:prot}
  \begin{inferruleraw}
    \inference{
      P_\SG =
      \inCtx{
          \premCMem{\SG}{s}{\lit*{protected} \mathrel{\lit*{var}} x_d \mathrel{\lit*{=}} e}
      }
    }{
      \mathsf{protected}_{P_\SG}(x_d)
    }
  \end{inferruleraw}
\end{definition}
The premise of this rule states that $x_d$ was declared with the \texttt{protected} access modifier.

\begin{lemma}
  \label{lem:x-decl-prot}
  \begin{inferruleraw}
    \mathsf{protected}_{P_\SG}(x_d)\ 
    \Rightarrow\
    \exists s_c,\, T.\
    s_c \typeedget[\lblVAR] (\defdecl{x_d}{T}{\PRT}) \in \SG
  \end{inferruleraw}
\end{lemma}

\begin{proof}
From $\mathsf{protected}_{P}(x_d)$, we derive $\inCtx{\premCMem{\SG}{s_c}{\lit*{protected} \mathrel{\lit*{var}} x_d \mathrel{\lit*{=}} e}}$, for some $\ctx$.
Inversion yields a single case, using rule~\rref{d-def}, from which we infer that there exists some $\ctx_1$, $\ctx_2$ and $\ctx_3$ such that
\[
  \inCtx[_1]{\premExp{\SG}{s_c}{e}{T}}
  \qquad
  \inCtx[_2]{\premAcc{\SG}{s_c}{\lit*{protected}}{\PRT}}
  \qquad
  \inCtx[_3]{s_c \typeedget[\lblVAR] (\defdecl{x_d}{T}{\PRT}) \in \SG}
\]
From the third premise, the goal follows by~\cref{lem:lift-derivation}.
\end{proof}
We can now prove the correctness of our \texttt{protected} access modifier specification.
The desired semantics of the \texttt{protected} access modifier is defined as follows:

\begin{theorem}[Soundness of \texttt{protected} member access]
\label{def:prot-form}
Let $P_\SG$ be a valid typing derivation for an AML program with scope graph $\SG$.
Then, for all possible $x_r$ and $x_d$,
\begin{align*}
  &\mathsf{resolve}_{P_\SG}(x_r) = x_d 
    \mathrel{\land} 
    \mathsf{protected}_{P_\SG}(x_d)\ 
  \Rightarrow {}\\
  &\hspace{1em}
  \exists s_c, s_d.\ 
    \mathsf{definingClass}_{P_\SG}(x_d) = s_d 
    \mathrel{\land}
    \mathsf{enclosingClass}_{P_\SG}(x_r, s_c)
    \mathrel{\land}
    \mathsf{subClass}_{P_\SG}(s_c, s_d)
\end{align*}
\end{theorem}
Here, $\mathsf{resolve}_{P_\SG}(x_r) = x_d$ means that the typing derivation of $P_{\SG}$ witnesses that reference $x_r$ resolves to declaration $x_d$,
$\mathsf{protected}_{P_\SG}(x_d)$ means that $x_d$ is a class field with \texttt{protected} access,
$\mathsf{definingClass}_{P_\SG}(x_d) = c$ means that $x_d$ is a member of the class represented by scope $s_c$, and
$\mathsf{enclosingClass}_{P_\SG}(x_r, c)$ holds if class $c$ contains the reference $x_r$.
$\mathsf{subClass}_{P_\SG}(s_c, s_d)$ states that class $s_c$ is a subclass of $s_d$.

Next, we prove that~\cref{def:prot-form} is satisfied by our specification.
\begin{proof}

  First, we prove that there is a class $s_d$ in which $x_d$ is defined.
  Then, we prove that some enclosing class $s_c$ of $x_r$ is a sub-class of $x_d$.
  
  \proofsubparagraph{Defining Class.}
  We get $\mathsf{definingClass}_{P_\SG}(x_d) = s_d$ in the conclusion as follows.
  By inversion on the premise $\mathsf{protected}_{P_\SG}(x_d)$, we get a $\ctx$, $s_d$, and $e$ for which $\inCtx{\premCMem{\SG}{s_d}{\lit*{protected} \mathrel{\lit*{var}} x_d \mathrel{\lit*{=}} e}}$.
  From this and~\cref{def:defCls}, the goal follows.

  \proofsubparagraph{Enclosing Class.}
  From an argument similar to the one in the proof of~\cref{def:priv-form}, we infer that
  $x_r$ is resolved in some scope $s_r$, and
  the access modifier included in the query result is \PRT.
  From~\cref{lem:res-to-acc}, we infer that
  \[
    \premMAcc{\SG}{s_r}{p}{\PRT}
  \]
  for the resolution path $p$.

  Inversion on this premise yields two cases.
  For \rref{ap-priv}, we infer $\mathsf{enclosingClass(x_r, s_d)}$ (similar to the proof of~\cref{def:priv-form}).
  In that case, choosing $s_c = s_d$ satisfies $\mathsf{subClass}(s_c, s_d)$, from which our goal follows.
  In the \rref{ap-prot}-case, we have the following premises:
  \begin{itemize}
    \item $\premEncC{\SG}{s_r}{S_C}$
    \item $s_c' \in S_C$
    \item $s_c' \in \mathsf{scopes}(p)$
  \end{itemize}
  for some scope $s_c'$.
  From the third premise, we infer that we can split $p$ in two segments $p_1$ and $p_2$, such that $\mathsf{tgt}(p_1) = \mathsf{src}(p_2) = s_c'$.
  From~\cref{lem:query-path}, we infer that $\premMatchRE{p_2}{\reLEX}$.
  Now, we pick the scope $s_c$ to be a scope in $p_2$ whose \emph{incoming} edge has label $\lblLEX$,
  but has no outgoing $\lblLEX$.
  That is $p_2$ is either the last scope in the path, or has an $\lblLEX$-edge as outgoing step.
  In both cases, $s_c$ is a class scope, because no paths to variables outside classes exist in AML.
  Therefore $s_c \in S_C$, and hence $\mathsf{enclosingClass(x_r, s_c)}$.

  When $s_c$ is the end of the path, $s_c = s_d$, and hence $\mathsf{subClass}(s_c, s_d)$.
  In the other case, by the regular expression $\reLEX$, we know only $\lblEXT$ edges can follow after $s_c$.
  For that reason, we can also infer $\mathsf{subClass}(s_c, s_d)$.
  Together, this proves our initial goal.
\end{proof}
\subsection{Soundness of Internal Access Validation}
\label{subsec:int-sound}

First, we define a $\mathsf{internal}(x_d, \setVar{x})$ that holds when $x_d$ is a variable declared with the \lit*{internal} access modifier with module parameters $\setVar{x}$:
\begin{definition}[$\mathsf{internal}$]
  \label{def:int}
  \begin{inferruleraw}
    \inference{
      P_\SG =
      \inCtx{
        \premCMem{\SG}{s}{\lit*{internal} \lit*{(} \setVar{x} \lit*{)} \mathrel{\lit*{var}} x_d \mathrel{\lit*{=}} e}
      }
    }{
      \mathsf{internal}_{P_\SG}(x_d, \setVar{x})
    }
  \end{inferruleraw}
\end{definition}
The premise of this rule states that $x_d$ was declared with the \texttt{internal} access modifier.

\begin{lemma}
  \label{lem:x-decl-int}
  \begin{align*}
    &\mathsf{internal}_{P_\SG}(x_d, \setVar{x})\ 
    \Rightarrow {}\\
    &\hspace{1em}
    \exists s_c,\, S,\, T.\
    S = \left\{
      s_m \mathbin{\Big|}
      x \in \setVar{x}, \premQMod{\SG}{s_c}{x}{s_m}
    \right\}
    \mathrel{\land}
    s_c \typeedget[\lblVAR] (\defdecl{x_d}{T}{\MODof{S}}) \in \SG
  \end{align*}
\end{lemma}

\begin{proof}
From $\mathsf{internal}_{P}(x_d, \setVar{x})$, we derive $\inCtx{\premCMem{\SG}{s_c}{\lit*{internal} \lit*{(} \setVar{x} \lit*{)}\mathrel{\lit*{var}} x_d \mathrel{\lit*{=}} e}}$, for some $\ctx$.
Inversion yields a single case, using rule~\rref{d-def}, from which we infer that there exists some $\ctx_1$, $\ctx_2$ and $\ctx_3$ such that
\[
  \inCtx[_1]{\premExp{\SG}{s_c}{e}{T}}
  \qquad
  \inCtx[_2]{\premAcc{\SG}{s_c}{\lit*{internal}}{\MODof{S}}}
  \qquad
  \inCtx[_3]{s_c \typeedget[\lblVAR] (\defdecl{x_d}{T}{\MODof{S}}) \in \SG}
\]
for some $S$.
From the third premise and~\cref{lem:lift-derivation}, we infer that $s_c \typeedget[\lblVAR] (\defdecl{x_d}{T}{\MODof{S}}) \in \SG$ holds for $S$.
By inversion on the second premise (using \rref{a-int}), we infer that $S$ is constructed precisely as defined in~\cref{lem:x-decl-int}.
\end{proof}
We can now prove the correctness of our \texttt{internal} access modifier specification.
The desired semantics of the \texttt{internal} access modifier is defined as follows:

\begin{theorem}[Soundness of \texttt{internal} member access]
\label{def:int-form}
Let $P_\SG$ be a valid typing derivation for an AML program with scope graph $\SG$.
Then, for all possible $x_r$ and $x_d$,
\begin{align*}
  &\mathsf{resolve}_{P_\SG}(x_r) = x_d 
    \mathrel{\land} 
    \mathsf{internal}_{P_\SG}(x_d, \setVar{x})\
  \Rightarrow {}\\
  &\hspace{1em}
  (\exists x\, s_m.\
    \mathsf{enclosingMod}_{P_\SG}(x_r) = s_m
    \mathrel{\land}
    x \in \setVar{x}
    \mathrel{\land} 
    \mathsf{resolveMod}(x) = s_m) \lor {}\\
  &\hspace{1em} 
  (\exists s_c.\ 
    \mathsf{definingClass}_{P_\SG}(x_d) = s_c 
    \mathrel{\land} 
    \mathsf{enclosingClass}_{P_\SG}(x_r, s_c))
\end{align*}
\end{theorem}
Here, $\mathsf{resolve}_{P_\SG}(x_r) = x_d$ means that the typing derivation of $P_{\SG}$ witnesses that reference $x_r$ resolves to declaration $x_d$,
$\mathsf{internal}_{P_\SG}(x_d, \setVar{x})$ means that $x_d$ is a class field with \texttt{internal} access, with module names $\setVar{x}$.
$\mathsf{enclosingMod}_{P_\SG}(x_r) = s_m$ means that $s_m$ is the scope of the most directly enclosing module of $x_r$, and
$\mathsf{resolveMod}(x) = s_m$ holds if module reference $x$ resolves to the module with scope $s_m$.
The right-hand side of the disjunct is similar to the private property (\cref{def:priv-form}).
In this way, we allow access to a \lit*{internal} variable as if it were private.

Next, we prove that~\cref{def:int-form} is satisfied by our specification.
\begin{proof}
  Similar to the proofs of~\cref{def:priv-form,def:prot-form}, using~\cref{lem:x-decl-int}, we infer that
  \[
    s_d \typeedget[\lblVAR] (\defdecl{x_d}{T}{\MODof{S}}) \in \SG   
  \]
  where
  \[
    S = \left\{
      s_m \Big|
      x \in \setVar{x}, \premQMod{\SG}{s_d}{x}{s_m}
    \right\}
  \]
  in some scope $s_d$.
  Likewise, we infer that $x_r$ is resolved in some scope $s_r$, for which the path is validated as follows:
  \[
    \premMAcc{\SG}{s_r}{p}{\MODof{S}}
  \]
  Inversion on this premise yields two cases.
  First, the case for~\rref{ap-priv} makes the right-hand side of the disjunction true (similar to the proof of~\cref{def:priv-form}), which proves our goal.

  The case for~\rref{ap-int} yields the following assumptions:
  \begin{itemize}
    \item $\premEncMI{\SG}{s_r}{s_m'}$, and
    \item $s_m' \in S$.
  \end{itemize}
  Using these premises, we will prove the left-hand disjunct of our goal:
  \begin{align*}
    \exists s_m,\, x.\
    \mathsf{enclosingMod}_{P_\SG}(x_r) = s_m
    \mathrel{\land}
    x \in \setVar{x}
    \mathrel{\land} 
    \mathsf{resolveMod}(x) = s_m
  \end{align*}
  By instantiating the existentially quantified variable $s_m$ to $s_m'$, we obtain $\mathsf{enclosingMod}_{P_\SG}(x_r) = s_m'$ as a proof goal,
  which is proven using the first premise, and the initial assumption $\mathsf{resolve}_{P_\SG}(x_r) = x_d$.
  Finally, from the construction of $S$, we infer that
  \[
    \premQMod{\SG}{s_d}{x'}{s_m'}
  \]
  for some $x' \in \setVar{x}$.
  By instantiating $x$ with $'$, \rref{q-mod}, and~\cref{def:resolve-mod-in}, we infer $\mathsf{resolveMod}(x') = s_m'$.
  This, together with $\mathsf{enclosingMod}_{P_\SG}(x_r) = s_m'$, proves our goal.
\end{proof}

}{
}

\end{document}